\documentclass[11pt]{article}
\usepackage{amsmath,amsthm,amssymb,amsfonts,graphicx,multicol,multirow,subfigure}
\usepackage{cite}
\usepackage[all]{xy}
\usepackage{fullpage,tabularx}
\usepackage[small,compact]{titlesec}
\usepackage[small,it]{caption}
\usepackage{hyperref}

\usepackage{times}

 \newenvironment{cit}
{
    \begin{list}{- \ }{}
        \setlength{\topsep}{0pt}
        \setlength{\parskip}{0pt}
        \setlength{\partopsep}{0pt}
        \setlength{\parsep}{0pt}         
        \setlength{\itemsep}{0pt}
     \setlength{\leftmargin}{0em}
     \setlength{\labelwidth}{1.5em}
     \setlength{\labelsep}{0.5em} 
}
{
    \end{list} 
}

\def\etal{\emph{et al.} }

\newcommand{\stackanchor}[2]{{ \scriptsize \begin{array}{c} \scriptsize #1 \\ \scriptsize #2 \end{array}}}

\theoremstyle{plain}
\newtheorem{theorem}{Theorem}

\newtheorem{observation}[theorem]{Observation}

\theoremstyle{definition}

\usepackage[section]{algorithm}
\usepackage{algorithmic}

\newcommand{\comment}[1]{}
\def\cF{\mathcal{F}}
\def\R{\mathbb{R}}
\def\forb{\operatorname{forb}}
\def\etal{\textit{et al.}}

\begin{document}
\setlength{\parskip}{0em}
\setlength{\parindent}{.25in}

\title{\Large Automated Discharging Arguments for Density Problems in Grids \\
			 (Extended Abstract\thanks{ %
			 	Pages 2-\pageref{paper:end} of this PDF contain the extended abstract, 
				with some 
				figures and tables
				% figures, tables, and implementation details 
				appearing in appendices (pp. \pageref{apx:begin}-\pageref{apx:end}).  
			 } )}
\author{Derrick Stolee\\
		Department of Computer Science\\
		Department of Mathematics\\
		Iowa State University\\
		\texttt{dstolee@iastate.edu}
	}

\maketitle

\begin{abstract}
Discharging arguments demonstrate a connection between local structure and global averages.
This makes it an effective tool for proving lower bounds on the density of special sets in infinite grids.
However, the minimum density of an identifying code in the hexagonal grid remains open, with an upper bound of $\frac{3}{7} \approx 0.428571$ and a lower bound of $\frac{5}{12}\approx 0.416666$.
We present a new, experimental framework for producing discharging arguments using an algorithm.
This algorithm replaces the lengthy case analysis of human-written discharging arguments with a linear program that produces the \emph{best possible} lower bound using the specified set of discharging rules.
We use this framework to present a lower bound of $\frac{23}{55} \approx 0.418181$ on the density of an identifying code in the hexagonal grid, and also find several sharp lower bounds for variations on identifying codes in the hexagonal, square, and triangular grids.
\end{abstract}

\clearpage
\section{Introduction}

\def\figtwoheight{0.8in}
\begin{figure}[bp]
\centering
	\begin{tabular}{ccc}
		\includegraphics[height=\figtwoheight]{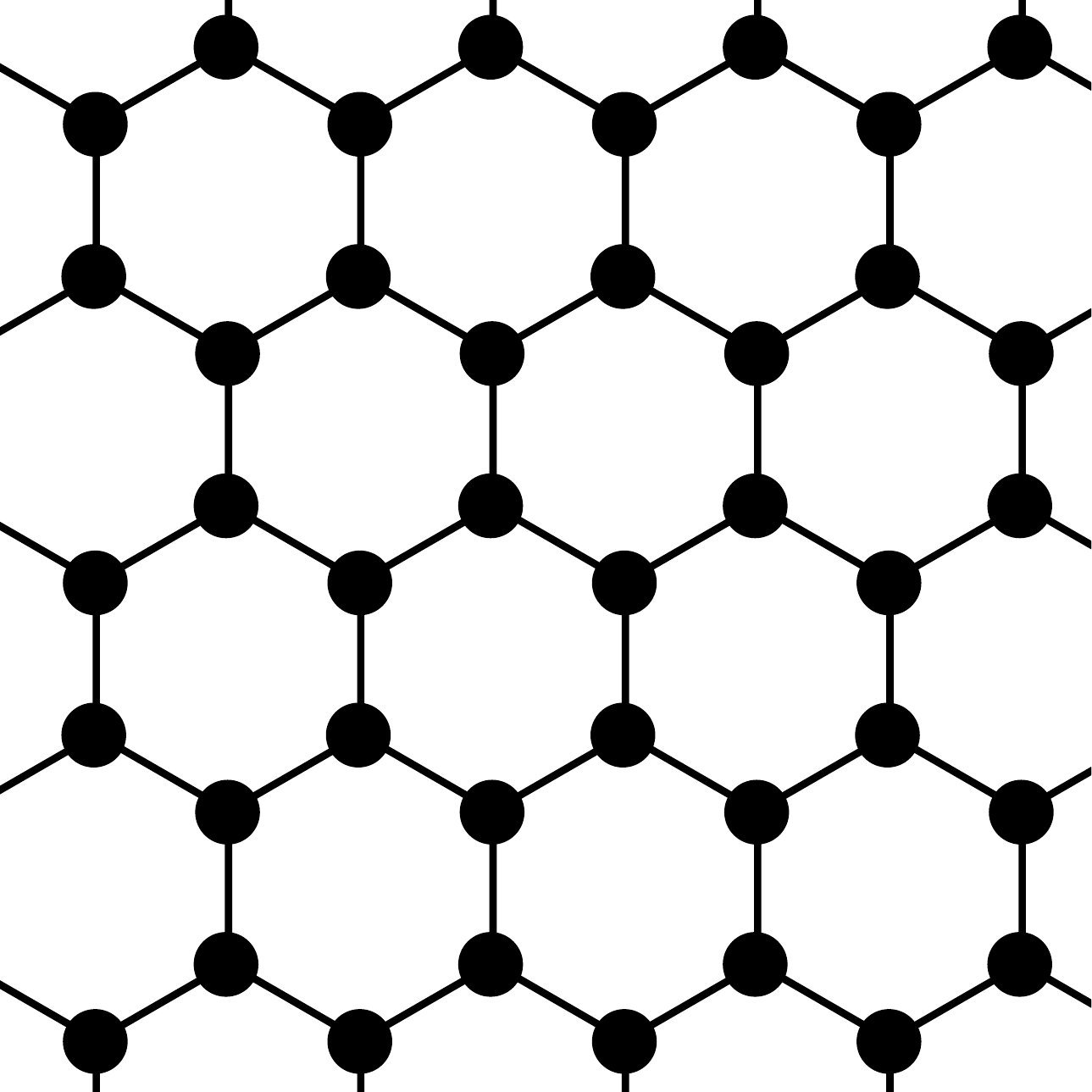}&
		\includegraphics[height=\figtwoheight]{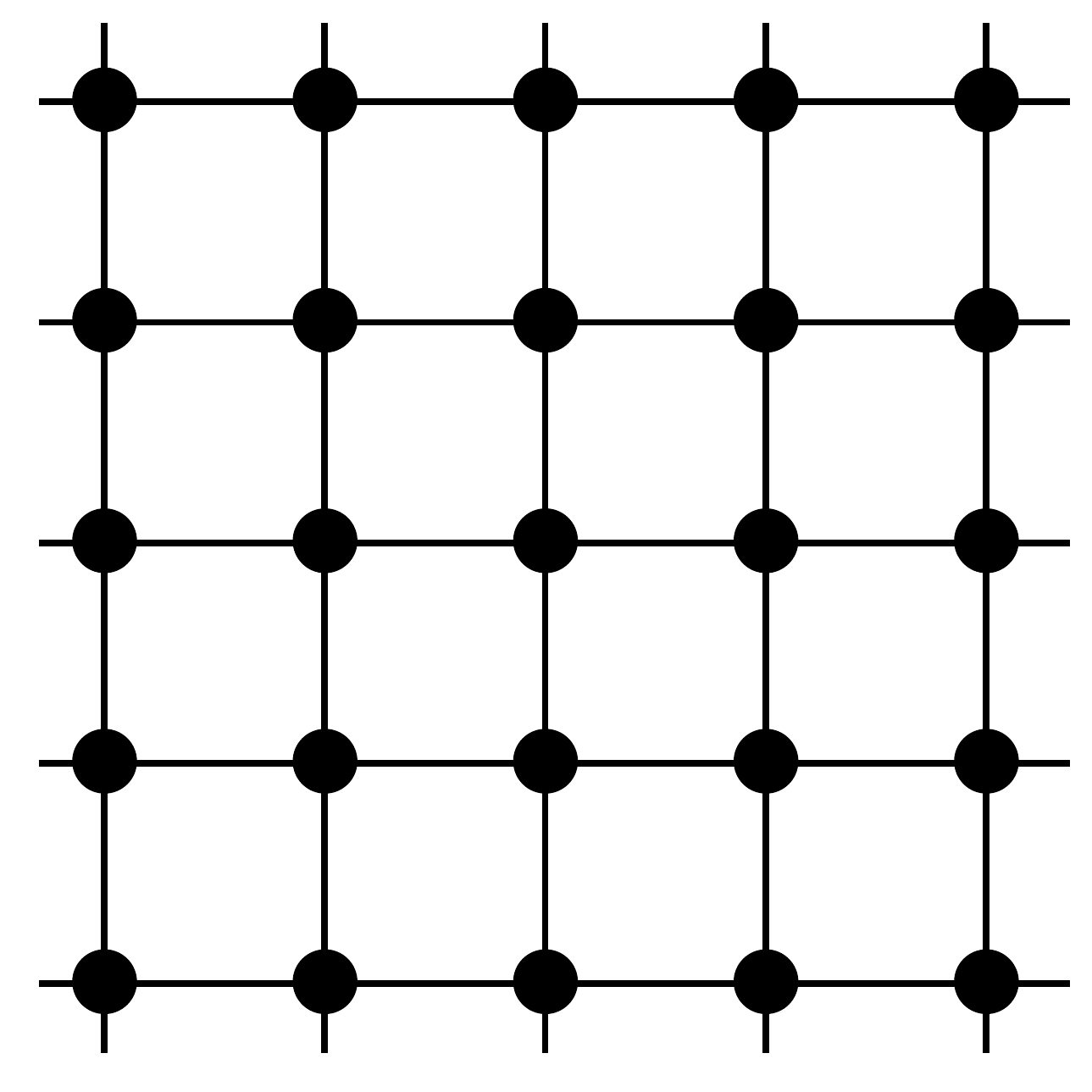}&
		\includegraphics[height=\figtwoheight]{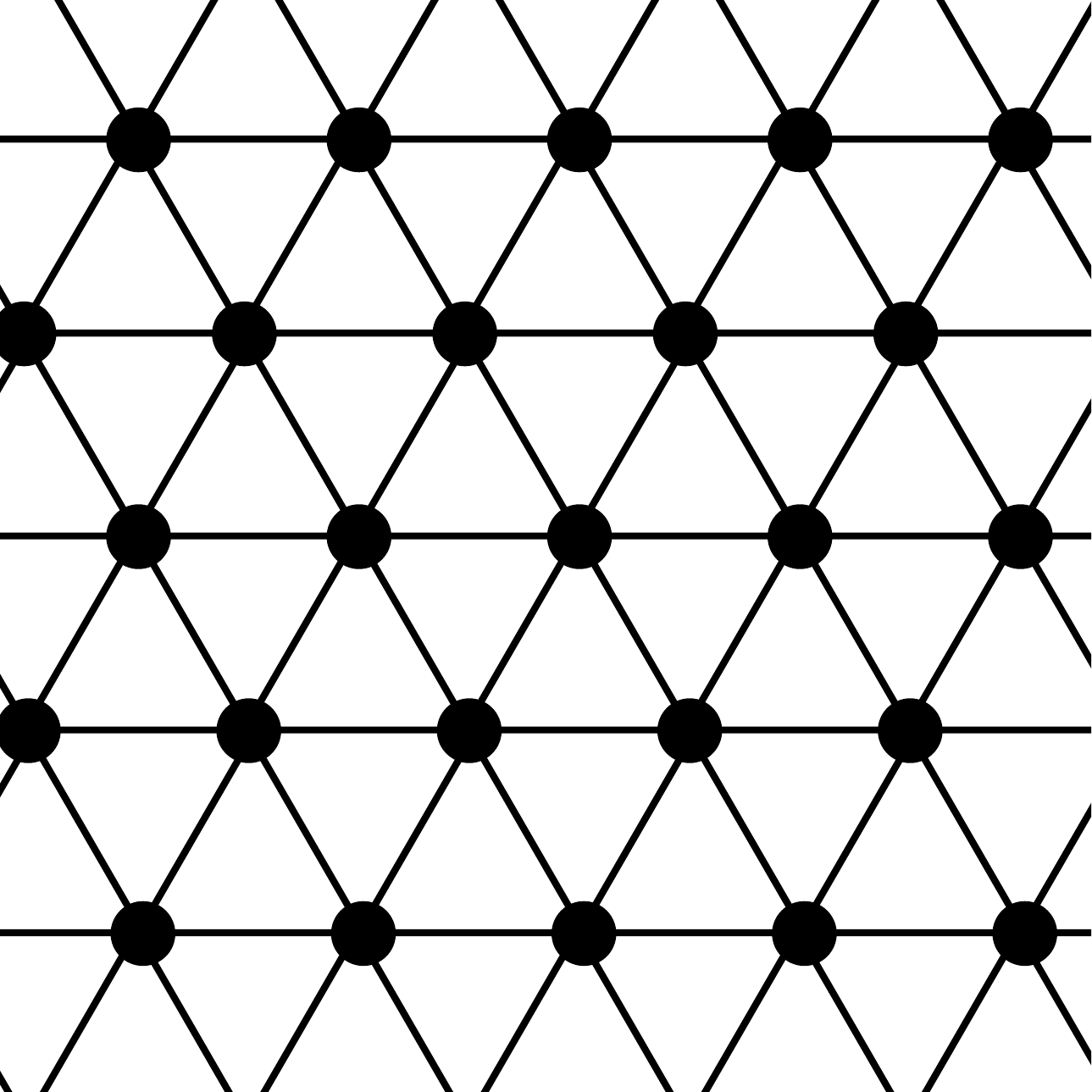}\\
		 (a) Hexagonal Grid & (b) Square Grid & (c) Triangular Grid 
	\end{tabular}
	\caption{\label{fig:grids}Examples of plane grids.}
\end{figure}

The discharging method is a well-known technique in discrete mathematics, especially due to its use in the computer-assisted proofs of the Four Color Theorem~\cite{appel1977every,appel1977every2,robertson1997four}.
Since that first incredible achievement, almost all other discharging proofs have been done manually.
%Even the Four Color Theorem proof was constructed with significant human interaction, and the computer verified the proof instead of creating it.
Applications of discharging include coloring planar graphs~\cite{cranston2013guide}, density problems in grids~\cite{cranston2009new,cukierman2013new}, and structural problems on circulant graphs~\cite{hartke2012uniquely}.
Despite its wide use, producing an effective discharging argument is very challenging and the proofs become weighted down by a lengthy case analysis.
We present an experimental method for automatically producing discharging arguments and apply this method to prove lower bounds on the density of sets in infinite grids.

A \emph{plane grid} is an infinite graph embedded in the plane. %where every vertex has bounded degree and every face has bounded length.
We will consider three plane grids: the hexagonal grid, the square grid, and the triangular grid, as shown in Figure~\ref{fig:grids}.
These grids model the structure of a wireless sensor network where the nodes are placed in a rigid lattice, as would be typical for use in a field for drought monitoring~\cite{dong2013autonomous}.
Due the low cost of wireless sensor nodes, these networks can be so large that the boundary of the network is a small portion of the entire network, so using an infinite grid is an effective way to approximate the network.
Karpovsky, Chakrabarty, and Levitin~\cite{karpovsky1998new} considered the problem of detecting faults in such a network and defined an \emph{identifying code} to be a set $X$ of vertices in the grid where the intersection of $X$ with the closed neighborhood of each vertex is distinct. 
If $N(v)$ is the set of vertices adjacent to $v$, then the closed neighborhood $N[v]$ is the set $N(v) \cup \{v\}$.
Thus an identifying code $X$  in a grid $G$ satisfies
\[
	N[v] \cap X \neq \varnothing, \forall v \in V(G), \qquad\text{and}\qquad N[v]\cap X \neq N[u]\cap X, \forall u,v \in V(G), u \neq v.
\]
An identifying code exists in a graph $G$ if and only if there are no \emph{twins} (distinct vertices $u, v$ where $N[u] = N[v]$) since using the entire vertex set can identify all vertices.
The interesting problem is to determine the \emph{smallest} identifying code in $G$, to minimize the cost of placing fault-detection devices on the nodes representing elements of the code.

In an infinite grid, any identifying code will be infinite, so we need a notion of density instead of size.
For a vertex $v$, let $B_r(v)$ be the set of vertices within distance $r$ of $v$ in $G$.
The \emph{density} of a set $X \subseteq V(G)$, denoted $\delta(G)$, is defined as a limit of the proportion of elements of $X$ in the ball of radius $r$, as $r$ grows.
\[
	\delta(X) = \limsup_{r\to\infty}  \frac{|B_r(v) \cap X|}{|B_r(v)|} \text{ for some vertex $v \in V(G)$.}
\]
Cohen \etal~\cite{cohen2000bounds} demonstrated an identifying code in the hexagonal grid of density $\frac{3}{7}$, and Cukierman and Yu~\cite{cukierman2013new} found several other constructions of this density.
However, lower bounds on the optimal density have not reached this upper bound, despite several attempts~\cite{karpovsky1998new,cohen2000bounds,cranston2009new,cukierman2013new}.
Table~\ref{tab:identbounds} summarizes these efforts, including the new lower bound demonstrated in this paper.

\begin{table}[t]
\centering
\renewcommand{\arraystretch}{1.25}
\begin{tabular}[h]{|r|r@{\ }c@{\ }c@{\ }c@{\ }c@{\ }c@{\ }c@{\ }c@{\ }c|}
\hline
Karpovsky, Chakrabarty, and Levitin~\cite{karpovsky1998new} & 0.400000 & $=$ & $\frac{2}{5}$ & $\leq$ & $\delta$ &$\leq$ & $\frac{1}{2}$ & $=$ & $0.500000$ \\
\hline
Cohen, Honkala, Lobstein, and Z\'emor \cite{cohen2000bounds} & 0.410256 & $\approx$ & $\frac{16}{39}$ & $\leq$ & $\delta$ &$\leq$ & $\frac{3}{7}$ & $\approx$ & $0.428571$ \\
\hline
Cranston and Yu \cite{cranston2009new}& 0.413793 & $\approx$ & $\frac{12}{29}$ & $\leq$ & $\delta$ &&&& \\
\hline
Cukierman and Yu \cite{cukierman2013new} & 0.416666 & $\approx$ & $\frac{5}{12}$ & $\leq$ & $\delta$ &&&& \\
\hline
\textbf{New Lower Bound} & 0.418182 & $\approx$ & $\frac{23}{55}$ & $\leq$ & $\delta$ &&&& \\
\hline
\end{tabular}

\caption{\label{tab:identbounds}Existing bounds on $\delta$, the minimum density of an identifying code in the hexagonal grid.}
\end{table}

\begin{theorem}\label{thm:main}
If $X$ is an identifying code in the hexagonal grid, then $\delta(X) \geq \frac{23}{55} = 0.4\overline{18}$.
\end{theorem}

We prove Theorem~\ref{thm:main} using a new computer-automated method for constructing discharging arguments.
Due to the generality of the method and the modular nature of the implementation, the method also demonstrates lower bounds for other grids and for other variants on identifying codes.
Thus, several sharp lower bounds are proven for these variants in the hexagonal, square, and triangular grids, which were major theorems of previous work~\cite{ben2005exact,karpovsky1998new,honkala2006locating,slater2002fault,seo2010open,kincaid2014optimal}.

Our main contribution is the development of this new computational approach to producing discharging arguments.
In Section~\ref{sec:charge}, we discuss the structure of our discharging arguments.
A recent survey by Cranston and West~\cite{cranston2013guide} includes a more general perspective, but focuses more on coloring planar graphs.
Briefly, discharging arguments all feature three main steps: (1) assigning initial charge, (2) distributing charge according to certain discharging rules, and (3) verifying that all objects have the goal amount of charge.
The most novel contribution of this new framework is that steps (2) and (3) are done in opposite order.
By using a combinatorial generation algorithm, we enumerate every possible way that the discharging rules can interact on a single object and create a linear program based on those interactions.
Every feasible solution corresponds to a correct discharging argument, and an optimal solution provides the largest lower bound possible using that set of rules.
This pairing of combinatorial generation and linear programming is similar to the use of generating and solving a semidefinite program in Razborov's flag algebra method~\cite{razborov2007flag}, which has gained significant attention in recent years (see Razborov's survey~\cite{razborov2013flag}).

The purpose of this extended abstract is to demonstrate how a custom computer algorithm can produce a discharging argument using minimal human interaction.
We name our method \emph{ADAGE} for ``Automated Discharging Arguments using GEneration,'' and a specific proof using the framework is an \emph{adage}.
While an adage is a short saying that conveys a general truth, an adage proof has a very short description while the computer handles the significant case analysis.
In Section~\ref{sec:grids} we discuss the important properties of the plane grids in more detail, followed by a definition of a \emph{configuration} and \emph{forbidden configuration}.
In Section~\ref{sec:charge} we set up the discharging argument, demonstrate how it is linked to the density of a set $X$, and discuss the structure of discharging rules.
The most crucial step is discussed in Section~\ref{sec:lp}, where a linear program is built to satisfy the assertions of the discharging argument.
Section~\ref{sec:variations} defines several variations on an identifying code, and the results on these variations are listed in Table~\ref{tab:bounds}.
Appendices~\ref{sec:ruletables} and \ref{sec:bounds} list several options for possible discharging rules in the three grids and the lower bounds demonstrated by the adage proofs using those rules.
%Appendix~\ref{sec:details} includes some information on the software implementation\footnote{See {\scriptsize\url{http://www.math.iastate.edu/dstolee/r/adage.htm}} for all software and data.}.

%\subsection{A Short Introduction to Discharging Arguments}\label{sec:discharging}
%
%The discharging rules must maintain a delicate balance: objects that began with low charge must gain enough to become average, and objects that began with high charge must maintain enough to remain average.
%Verifying that such properties hold usually requires technical case analysis, and designing such rules is a difficult task.
%
%``There exist rules such that for all chargeable objects and all surrounding features, the resulting charge value is above the goal value.''
%

\section{Grids, Density, and Configurations}
\label{sec:grids}

For a grid $G$, we use $V(G)$ to denote the set of vertices in $G$ and $F(G)$ to denote the set of faces in $G$.
We shall use standard graph theory terminology (see West~\cite{west2001introduction}) to treat a grid $G$ as an infinite plane graph with vertex set $V(G)$, edge set $E(G)$, and face set $F(G)$.

The grids $G$ have an automorphism group as graphs, but we will restrict the automorphisms to be affine linear maps on the plane with determinant 1.
Specifically, we allow rigid motions that are translations and rotations, but do not allow reflection (as such maps would have determinant $-1$).
We make this restriction based on intuition since previous discharging arguments in grids have made use of distinguishing between clockwise and counter-clockwise arrangements, which would be lost if we allowed reflection.
Under these automorphisms, all three grids are vertex- and face-transitive.
If we did not allow rotation, then the hexagonal grid would not be vertex-transitive and the triangular grid would not be face-transitive.

For a vertex $v \in V(G)$ and an integer $r \geq 0$, the \emph{ball of radius $r$ around $v$} is  the set of vertices within distance $r$ of $v$ and is denoted by $B_r(v)$.
The \emph{faces within distance $r$ of $v$} is the set of faces incident to vertices in $B_{r-1}(v)$ and is denoted by $F_r(v)$. 
For a face $f \in F(G)$, define $B_r(f)$ to be the vertices $v$ where $f \in B_r(v)$ and $F_r(f)$ to be the faces incident to vertices in $B_{r-1}(f)$.

A grid is \emph{amenable} if the maximum degree of a vertex in $G$ is finite, the maximum length of a face in $G$ is finite, and $\lim_{r\to\infty} \frac{|B_{r+d}(v)\setminus B_{r}(v)|}{|B_r(v)|} = 0$ for all finite values $d \geq 0$.
Amenable grids have the property that the boundary of a ball is a vanishing proportion of the volume of the ball as the radius of the ball grows.
Using this basic property, it is not difficult to make the following observation.

\begin{observation}\label{obs:densitydiffs}
Let $G$ be an amenable plane grid.
Then, for $u, v \in V(G)$, $\lim_{r \to \infty} \frac{|B_r(u) \cap B_r(v)|}{|B_r(u) \cup B_r(v)|} = 1,$ and hence for any set $X \subseteq V(G)$,
\[
	\limsup_{r \to \infty} \frac{|B_r(u) \cap X|}{|B_r(u)|} =  \limsup_{r \to \infty} \frac{|B_r(v) \cap X|}{|B_r(v)|}.
\]
\end{observation}

By this observation, we can define one vertex in $V(G)$ as the ``zero vertex,'' denoted by $v_0$, and define the \emph{density} of a set $X \subseteq V(G)$ to be the limit $\delta(X) = \limsup_{r\to\infty}  \frac{|B_r(v_0) \cap X|}{|B_r(v_0)|}$.
We also define one face in $F(G)$ as the ``zero face'' denoted by $f_0$.

%\subsection{Forbidden Configurations}
%\label{sec:forbidden}

%Let $G$ be a plane grid with vertices $V(G)$ and faces $F(G)$, and let $X$ be a subset of $V(G)$.
%Let $\ell$ be the maximum length of a face in $G$.
A \emph{configuration} is a tuple $(V, S_0, S_1, F)$ where $V$ is a finite set of vertices in $V(G)$,  $S_0$ and  $S_1$ are disjoint subsets of $V$, and $F$ is a finite set of faces in $F(G)$.
We call $S_1$ the \emph{elements} and $S_0$ the \emph{nonelements}.
While we require that $S_0 \cap S_1 = \varnothing$, we do not require that $S_0 \cup S_1 = V$.
Vertices in $V \setminus (S_0 \cup S_1)$ are considered \emph{undetermined vertices}.
%For a face $f \in F$, the value $s(f)$ represents the \emph{size} of the face in terms of the number of elements on the boundary of $f$; we therefore require that $|S_1| \leq s(f) \leq |\partial(f) \setminus S_0|$.
Frequently, we will denote a configuration by $C$ and refer to the entries of the tuple $(V,S_0,S_1,F)$ by $V(C)$, $S_0(C)$, $S_1(C)$, and $F(C)$. %, and $s_C$.
The automorphism group of $G$ naturally acts on configurations to produce a notion of isomorphism between configurations.

Such a configuration $C$ represents a finite induced subgraph of the grid $G$ and its planar dual $G^*$, as well as some information about $X$ on that induced subgraph.
Specifically, a configuration $C$ is \emph{embedded} in $X$ if $S_1(C) \subseteq X$, and $S_0(C) \subseteq V(G) \setminus X$. % and for all $f \in F(C)$, $s_C(f) = | X \cap \partial(f) |$.
Further, $C$ is \emph{embeddable} in $X$ if there exists a configuration $C'$ isomorphic to $C$ such that $C'$ is embedded in $X$.

% key to configuration figures???

Many families of subsets of $V(G)$, such as {identifying codes}, can be defined in terms of \emph{forbidden configurations}.
Given a collection $\cF = \{ C_1, \dots, C_k \}$ of configurations, the family $\forb(\cF)$ consists of sets $X \subset V(G)$ where for every $C_i \in \cF$, the configuration $C_i$ is not embeddable in the set $X$.
For a single configuration $C$, we use $\forb(C)$ to denote $\forb(\{C\})$.

For example, a \emph{dominating set} is a set $X \subseteq V(G)$ such that $N[v] \cap X \neq \varnothing$ for all vertices $v \in V(G)$.
If $v \in V(G)$ and $C$ is the configuration with $V(C) = S_0(C) = N[v]$, then $\forb(C)$ is the family of dominating sets in $G$.
Observe that for the family $\cF$ of configurations in Figure~\ref{fig:forb}, $\forb(\cF)$ is the family of identifying codes in the hexagonal grid.

\begin{figure}[htp]
\centering
\scalebox{1}{
\begin{tabular}[h]{m{0.5in}m{0.75in}m{0.75in}m{1in}}
\includegraphics[width=0.5in]{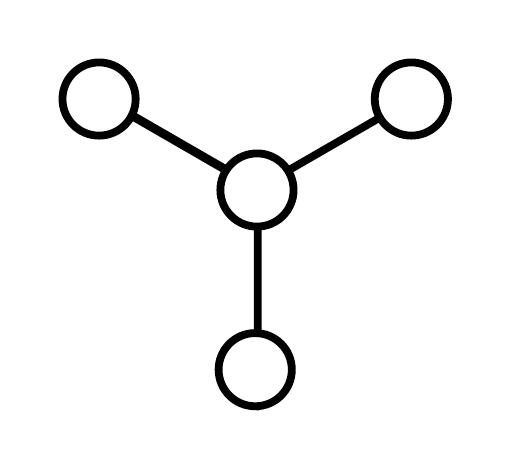} &
\includegraphics[width=0.75in]{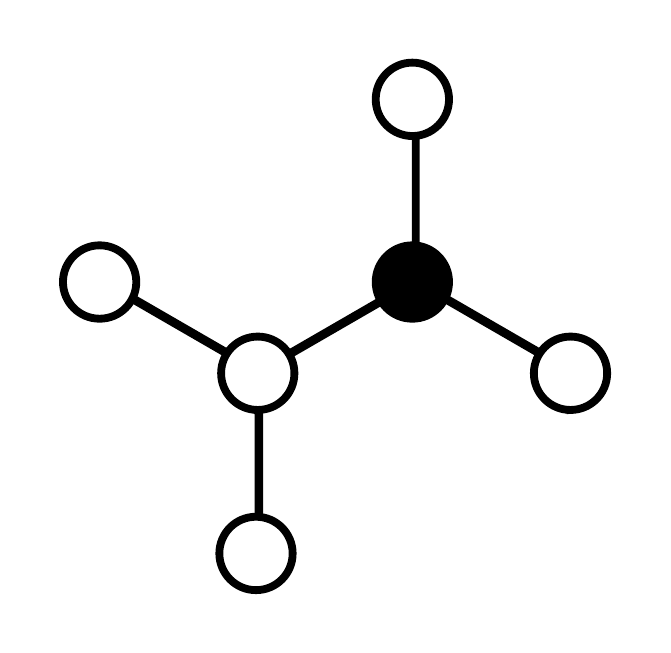} &
\includegraphics[width=0.75in]{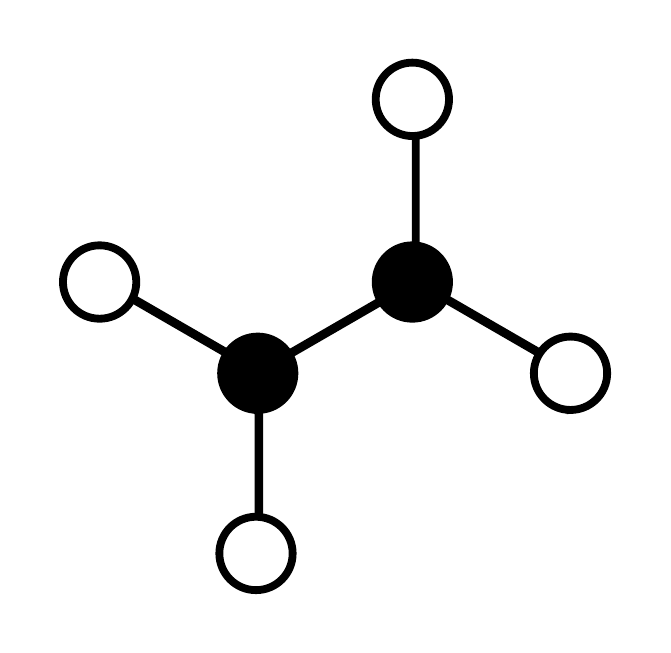} &
\includegraphics[width=1in]{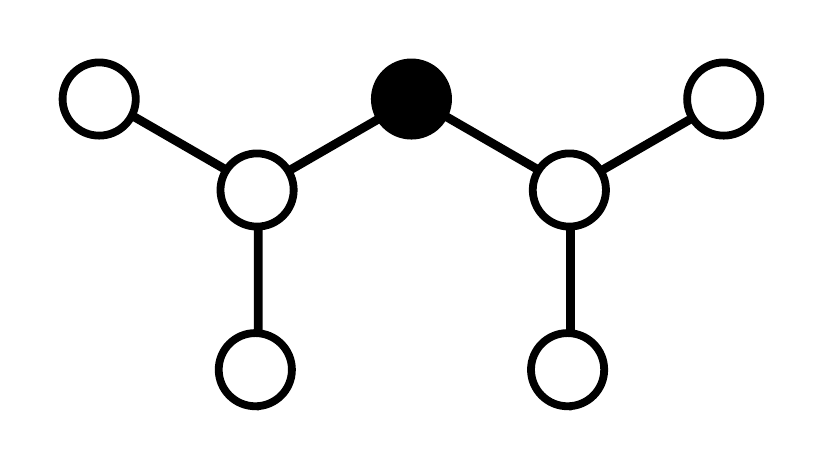} 
\end{tabular}}
\vspace{-1.5em}
\caption{\label{fig:forb}Forbidden configurations for identifying codes in the hexagonal grid.}
\end{figure}

For a family $\cF = \{C_1,\dots,C_k\}$ of forbidden configurations and a configuration $C$ where $S_0(C) = S_1(C) = \varnothing$, we say that a configuration $C'$ is an $\cF$-realization of $C$ if $V(C') = V(C)$, $S_0(C') \cup S_1(C') = V(C')$, and $C'$ does not contain any configuration $C_i \in \cF$.
It is not difficult to generate all $\cF$-realizations of a configuration $C$ up to isomorphism using standard techniques.
If $\cF = \varnothing$, then there are $2^{|V(C)|}$ realizations of $C$, and possibly fewer when $\cF \neq \varnothing$.

We will use this method of generating $\cF$-realizations of a configuration $C$ to examine all cases of how an embedding of $C$ in the grid $G$ can intersect a set $X$.
But first, we must discuss the structure of our discharging argument.

%
%When discussing extremal sets $X \in \forb(\cF)$, it can be helpful to define the concept of a \emph{reducible configuration}.
%If our goal is to minimize the density of $X$, then we can assume that $X$ is minimal within $\forb(\cF)$ under set containment.
%Thus, for a minimal set $X$ and an element $x \in X$, the set $X \setminus \{ x\}$ is not in $\forb(\cF)$, and hence $X \setminus \{x\}$ contains one of the forbidden configurations.
%A configuration $C$ is \emph{reducible} (\emph{with respect to $\cC$}) if for every set $X$ in $\forb(\cF)$ we have that if $C$ is embeddable in $X$ then $X$ is not minimal in $\forb(\cF)$.
%
%For example, consider dominating sets in the hexagonal grid.
%Let $v$ be a vertex and if $C$ is the configuration where $V(C) = S_1(C) = B_1(v)$.
%If $C$ is embedded in a dominating set $X$, then let $X' = X \setminus \{v\}$ and observe that $X'$ remains a dominating set.
%Thus, $X$ is not minimal and $C$ is reducible with respect to dominating sets.
%See Figure~\ref{fig:reduce} to see examples of reducible configurations for other types of sets within the three grids.
%
%\begin{figure}[htp]
%\centering
%\begin{tabular}[h]{|c|c|c|}
%\hline
%&&\\
%\hline
%\end{tabular}
%\caption{\label{fig:reduce}Reducible configurations for certain monotone families.}
%\end{figure}

\section{Charge Assignment and Discharging Rules}
\label{sec:charge}

For our discharging argument, we will consider vertices and faces of $G$ to be \emph{chargeable objects} in that we will associate them with a numerical value, called a \emph{charge}.
We assign a charge function $\mu \colon V(G) \to \R$ on the vertices of $G$ and a charge function $\nu \colon F(G) \to \R$ on the faces of $G$.
These functions are based on the positions of the elements in a set $X \subseteq V(G)$:
for every vertex $v \in V(G)$, $\mu(v) = \begin{cases}1 & v\in C,\\ 0 & v\notin C\end{cases}$, and for every face $f \in F(G)$, let $\nu(f) = 0$.
Observe that $\sum_{v \in B_r(v_0)} \mu(v) + \sum_{f \in F_r(v_0)} \nu(f) = |X \cap B_r(v_0)|$ and hence 
\[
\delta(X) = \limsup_{r\to\infty} \frac{|X \cap B_r(v_0)|}{|B_r(v_0)|} = \limsup_{r\to\infty}\frac{\sum_{v \in B_r(v_0)} \mu(v) + \sum_{f \in F_r(v_0)} \nu(f)}{|B_r(v_0)|}.
\]

A \emph{discharging function} is a function $D_X \colon \left(V(G) \cup F(G)\right)\times \left(V(G) \cup F(G)\right) \to \R$ where $D_X(x,y) = -D_X(y,x)$ for all $x,y \in V(G) \cup F(G)$. 
Specifically, we can say that for two chargeable objects $x, y \in V(G) \cup F(G)$, the value $D_X(x,y)$ is the amount of charge to \emph{exchange from $x$ to $y$}.
Given a discharging function $D_X$, we define the resulting charge functions $\mu' \colon V(G) \to \R$ and $\nu' \colon F(G) \to \R$ as follows:
\begin{align*}
	\mu'(v) &= \mu(v) + \sum_{u \in V(G)} D_X(u,v) + \sum_{g \in F(G)} D_X(g,v).\\
	\nu'(v) &= \nu(f) + \sum_{u \in V(G)} D_X(u,f) + \sum_{g \in F(G)} D_X(g,f).
\end{align*}

For values $c, d > 0$, we say that $D_X$ is \emph{$(c,d)$-local} if  $|D_X(x,y)| \leq c$ always, and $D_X(x,y) = 0$ whenever the distance between $x$ and $y$ in $G$ exceeds $d$.
If a discharging function $D_X$ is $(c,d)$-local, then as a ball grows, the change in the total charge between the initial charge functions $\mu$, $\nu$ and the final charge functions $\mu'$, $\nu'$ are negligible compared to the size of the ball.

Our main assertion for a ``good'' discharging function is that the resulting charge functions satisfy $\mu'(v) \geq w$ and $\nu'(f) \geq 0$ for all vertices $v \in V(G)$ and faces $f \in F(G)$.
Roughly, this means that the initial charge on the vertices was ``spread out'' evenly so that every vertex has at least $w$ units of charge, and the faces did not contribute any positive charge to the vertices and instead were simply ``messengers'' of charge.
In the hexagonal grid, passing charge between vertices and faces can be particularly effective, since a face is incident to three antipodal pairs of vertices.

We make this assertion of a good discharging function concrete in the following theorem.

\begin{theorem}\label{thm:dischargingworks}
Let $G$ be an amenable grid.
Let $X \subseteq V(G)$, $c, d, w \geq 0$, and let $D_X$ be a $(c,d)$-local discharging function.
Define the charge functions $\mu, \mu', \nu, \nu'$ by the discharging process using $X$ and $D_X$.
If $\mu'(v) \geq w$ for all $v \in V(G)$ and $\nu'(f) \geq 0$ for all $f \in F(G)$, then $\delta(X) \geq w$.
\end{theorem}

\begin{proof}
By Observation~\ref{obs:densitydiffs}, we can select the zero vertex $v_0$ and zero face $f_0$ such that $|B_d(v_0)| = \max \{ |B_d(v)| : v \in V(G)\}$ and $|F_d(f_0)| = \max \{ |F_d(f)| : f \in F(G)\}$.

Recall that by the definitions of $\mu$ and $\nu$, 
\begin{equation}
	\delta(X) = \limsup_{r\to\infty} \frac{|X \cap B_r(v_0)|}{|B_r(v_0)|} = \limsup_{r\to\infty}\frac{\sum_{v \in B_r(v_0)} \mu(v) + \sum_{f \in F_r(v_0)} \nu(f)}{|B_r(v_0)|} \label{eqn:delta}
\end{equation}
By hypothesis,
\begin{equation}
\limsup_{r\to\infty} \frac{\sum_{v \in B_r(v_0)} \mu'(v) + \sum_{f \in F_r(v_0)} \nu'(f)}{|B_r(v_0)|} \geq \limsup_{r\to\infty} \frac{\sum_{v \in B_r(v_0)} w + \sum_{f \in F_r(v_0)} 0}{|B_r(v_0)|} = w. \label{eqn:discharged}
\end{equation}
Our goal is to prove that the limit at the end of (\ref{eqn:delta}) and the limit at the beginning of (\ref{eqn:discharged}) are equal, thereby showing that $\delta(X) \geq w$.

Using the definition of $\mu'$ and $\nu'$ and the fact that $D_X$ is $(c,d)$-local, we find that the absolute difference $
\left|\sum_{v\in B_r(v_0)} \left[\mu'(v)-\mu(v)\right] + \sum_{f \in F_r(v_0)} \left[\nu'(f)-\nu(f)\right]\right|$ is equal to the magnitude of the charge that $D_X$ exchanges across the boundaries of $B_r(v_0)$ and $F_r(f)$:
\begin{align*}
&\hspace{-0.1in}\left| \sum_{v\in B_r(v_0)}\left[\sum_{u \notin B_r(v_0)} D_X(u,v) + \sum_{g \notin F_r(v_0)} D_X(g,v)\right]  + \sum_{f\in F_r(v_0)}\left[\sum_{u \notin B_r(v_0)} D_X(u,v) + \sum_{g \notin F_r(v_0)} D_X(g,f)\right]\right|\\
%&\leq \sum_{\stackanchor{v\in B_{r+d}(v_0)\setminus B_r(v_0)}{u \in B_d(v)}} |D_X(u,v)| + \sum_{\stackanchor{v\in B_{r+d}(v_0)\setminus B_r(v_0)}{g \in F_d(v)}} |D_X(g,v)|  + \sum_{\stackanchor{f\in F_{r+d}(v_0)\setminus F_r(v_0)}{u \in B_d(f)}} |D_X(u,f)| + \sum_{\stackanchor{f\in F_{r+d}(v_0)\setminus F_r(v_0)}{g \in F_d(f)}} |D_X(g,f)|\\
&\leq \sum_{u \in B_{r+d}(v_0)\setminus B_r(v_0)} \left[\sum_{v \in B_d(u)} |D_X(u,v)| + \sum_{f \in F_d(u)} |D_X(u,f)|\right] \\
&\qquad\qquad+  \sum_{g \in F_{r+d}(v_0)\setminus F_r(v_0)} \left[\sum_{v \in B_d(g)} |D_X(g,v)| + \sum_{f \in F_d(g)} |D_X(g,f)|\right] \\
&\leq  |B_{r+d}(v_0)\setminus B_r(v_0)| \cdot \left[ |B_d(v_0)| + |F_d(v_0)| \right] \cdot c + |F_{r+d}(v_0)\setminus F_r(v_0)| \cdot \left[|B_d(f_0)|+|F_d(f_0)|\right] \cdot c \\
&\leq c' \cdot \left[ |B_{r+d}(v_0) \setminus B_r(v_0)| + |F_{r+d}(v_0)\setminus F_r(v_0)| \right],
\end{align*}
where $c' = c \left[|B_d(v_0)| + |F_d(f_0)|\right]$.
Since $G$ is amenable and has finite maximum degree $\Delta(G)$, 
\[
\limsup_{r \to \infty} \frac{ |B_{r+d}(v_0) \setminus B_r(v_0)| + |F_{r+d}(v_0)\setminus F_r(v_0)| }{|B_{r}(v_0)|} \leq 
\limsup_{r \to \infty} \frac{ (\Delta(G)+1)|B_{r+d}(v_0) \setminus B_r(v_0)| }{|B_{r}(v_0)|} = 0,
\] and therefore
\[\limsup_{r\to\infty} \frac{\left|\sum_{v\in B_r(v_0)} \left[\mu'(v)-\mu(v)\right] + \sum_{f \in F_r(v_0)} \left[\nu'(f)-\nu(f)\right]\right|}{|B_{r}(v_0)|} = 0,\]
proving the claim.
\end{proof}

Theorem~\ref{thm:dischargingworks} demonstrates that $(c,d)$-local discharging functions provide a way to bound the density of a set $X$.
However, as defined, the function $D_X$ depends on the entire (possibly infinite) set $X$.
This is not an effective strategy for us to prove anything about a discharging function.
In order to build \emph{effective} discharging functions, we will assemble one using \emph{discharging rules}.

Informally, a discharging rule is a way to examine the local situation around a chargeable object, and then decide to exchange a certain amount of charge among nearby chargeable objects.
Such a rule could, for instance, consider which elements in $B_2(v)$ are in $X$, and use that information to exchange charge between $v$ and the vertices adjacent to $v$, or between $v$ and the faces incident to $v$.
If the amount of charge exchanged depends only on the isomorphism class of the configuration $(B_2(v),B_2(v) \setminus X, B_2(v) \cap X, F_1(v))$, then this rule has a finite description, even though it is applied an infinite number of times.
When several discharging rules are applied simultaneously, the discharging function $D_X$ is defined by collecting all of the charge exchanges from all instances of the discharging rules.

Formally, a \emph{discharging rule} is a tuple $R = (C, z, y_1, \dots, y_t, \sigma)$ where $C$ is a configuration, $z$, $y_1,\dots, y_t$ are chargeable objects in $C$, and $\sigma$ is a function $\sigma : \{0,1\}^{V(C)} \times \{y_1,\dots,y_t\} \to \R$.
We will consider the first parameter of $\sigma$ to be the incidence vector corresponding to the set $X \cap V(C)$.
Thus $\sigma(X\cap V(C), y_i)$ determines how much charge to exchange from $y_i$ to $z$, given the realization of $C$
For an embedding $\pi$ of $C$ into $G$, the rule considers $\pi^{-1}(X) \cap V(C)$ and the function $\sigma$ defines that some amount of charge is exchanged from each $\pi(y_i)$ to $\pi(z)$.
Thus, the rule defines a discharging function $D_X^R$ as
\[
	D_X^R(a,b) = \sum_{\stackanchor{\pi : \pi(z) = b}{i : \pi(y_i) = a}} \sigma( \pi^{-1}(X) \cap V(C), y_i ) - \sum_{\stackanchor{\pi : \pi(z) = a}{i : \pi(y_i) = b}}\sigma( \pi^{-1}(X) \cap V(C), y_i ).
\]
The above definition states that the amount of charge sent from $a$ to $b$ is the combination of the charge sent from $a$ to $b$ via all embeddings of the rule where $\pi(z) = b$ and $\pi(y_i) = a$ for some $i$, minus the charge sent from $b$ to $a$ via all embeddings of the rule where $\pi(z) = a$ and $\pi(y_i) = b$ for some $i$.

%Note that if the configuration $C$ in a rule $R$ has non-trivial symmetry (i.e. there exist non-trivial embeddings of $C$ into $C$ that point-wise stabilize $z$ and set-wise stabilizes $\{y_1,\dots,y_t\}$), then the function $\sigma$ should be aware of this symmetry.

If $R_1, \dots, R_m$ is a list of rules, then the discharging function $D_X$ resulting from using these rules simultaneously is defined as $D_X(a,b) = \sum_{i=1}^m D_X^{R_i}(a,b)$.

%\subsection{Discharging Rules}\label{sec:rules}

We can very quickly describe the configuration $C$ and chargeable objects $z, y_1,\dots, y_t$ of a discharging rule.
The function $\sigma$ is more complicated, and in fact we do not specify it at all.
For each possible element of the domain of $\sigma$, we create a variable.
In the next section, we will describe how to create a linear program to assign value to these variables, thereby completely defining the discharging rules.

We now describe a few discharging rules in the hexagonal grid.
The following list carefully defines each rule, but these rules can be simply described visually by Figure~\ref{fig:hexrulesex}.

\begin{cit}
\item $V_i$ : Let $z$ be a vertex, and consider the configuration on $B_i(z)$ and $F_1(z)$ and let $\{y_1,\dots,y_t\} = F_1(z)$.
	Thus, this rule uses the information from $X \cap B_i(z)$ and uses that to specify how charge is exchanged to $z$ from the faces incident to $z$.

\item $N$ : Let $z$ be a vertex and consider the configuration on $F_1(z)$ that contains all vertices incident to the faces in $F_1(z)$. In the hexagonal grid, $N$ is larger than $V_2$ but smaller than $V_3$.

%\item $J_2$ : Let $z$ be a vertex and $f$ a face at distance 2 from $z$. Let $x$ be the vertex incident to $f$ that is adjacent to $z$. 
%	Consider the configuration on vertices incident to faces in $F_1(x)$ and let $y_1 = f$.
%	Thus, charge is exchanged from $f$ to $z$ depending on the information on the faces surrounding $x$.

\item $F_{1,3}$ : Let $z$ be a face and let $f_1,f_2,f_3$ be three consecutive faces adjacent to $z$.
	Let $y_1, y_2$ be the vertices incident to both $z$ and $f_2$.
	Consider the configuration of all vertices incident to $z, f_1, f_2, f_3$.

\item $C_1$ : Let $y_1$ and $z$ be adjacent vertices and consider the configuration of vertices incident to the two faces that are incident to $y_1$ and $z$.

\item $C_2$ : Let $z$ be a vertex, $f$ be one of the faces incident to $z$, and let $y_1$ and $y_2$ be the vertices incident to $f$ that are at distance 2 from $z$.
	Consider the configuration of vertices that are incident to $f$ or adjacent to a vertex incident to $f$.

%\item $C_3$ : Let $z$ be a vertex, $f$ be one of the faces incident to $z$, and let $y_1$ be the vertex incident to $f$ that is at distance 3 from $z$.
%	Consider the configuration of vertices that are incident to $f$ or adjacent to a vertex incident to $f$.

\item $E_1$ : Let $z$ and $y_1$ adjacent faces, and consider the configuration of vertices incident to $z$ or $y_1$.

%\item $E_{1,3}$ : Let $z$ be a face, $y_1$ be an adjacent face, and $f, g$ the two faces adjacent to both $z$ and $y_1$.
%	Consider the configuration of vertices incident to $z, y_1, f$, or $g$.

\item $E_6$ : Let $z$ be a face, $y_1,\dots, y_6$ be the faces adjacent to $z$, and consider the configuration of vertices incident to $z$ or adjacent to a vertex incident to $z$.
\end{cit}

\begin{figure}[tp]
\centering
\def\cellwidth{0.33in}
\def\rulewidth{0.9in}
\begin{tabular}[h]{|m{\rulewidth}|m{\rulewidth}|m{\rulewidth}||m{\rulewidth}|m{\rulewidth}|m{\rulewidth}|}
\hline
\multicolumn{1}{|c|}{$V_2$} & 
\multicolumn{1}{c|}{$N$} & 
\multicolumn{1}{c||}{$F_{1,3}$} &
\multicolumn{1}{c|}{$C_1$} & 
\multicolumn{1}{c|}{$C_2$} & 
\multicolumn{1}{c|}{$E_1$}
\\
\includegraphics[width=\rulewidth]{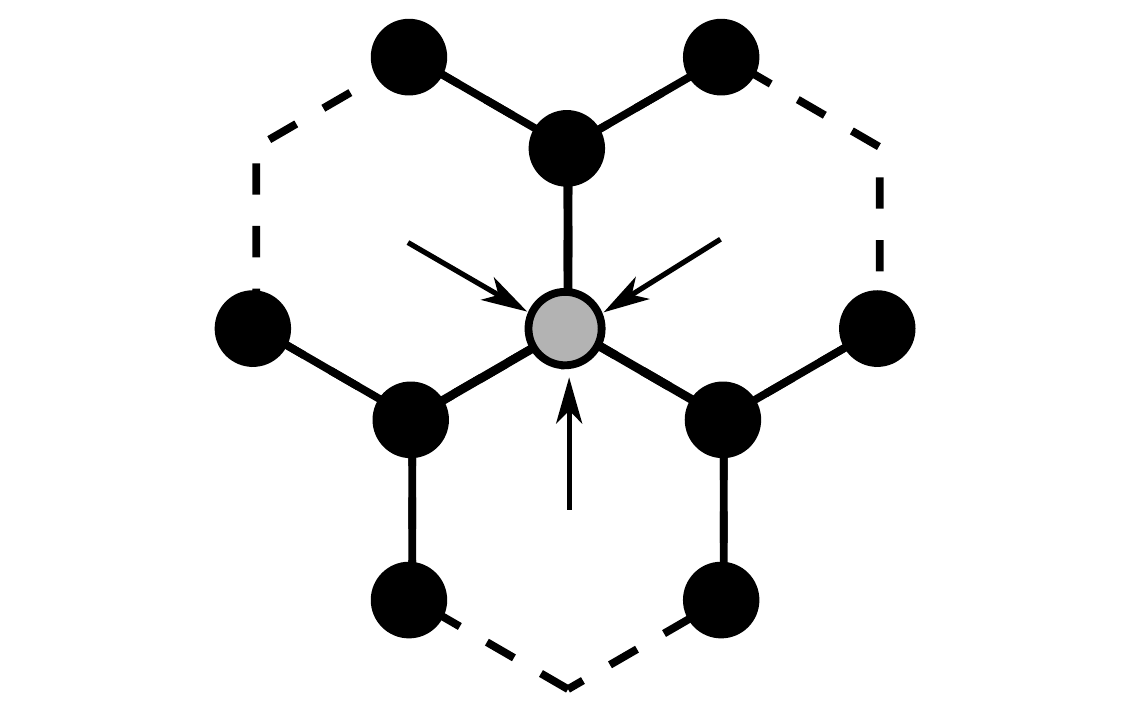} & 
\includegraphics[width=\rulewidth]{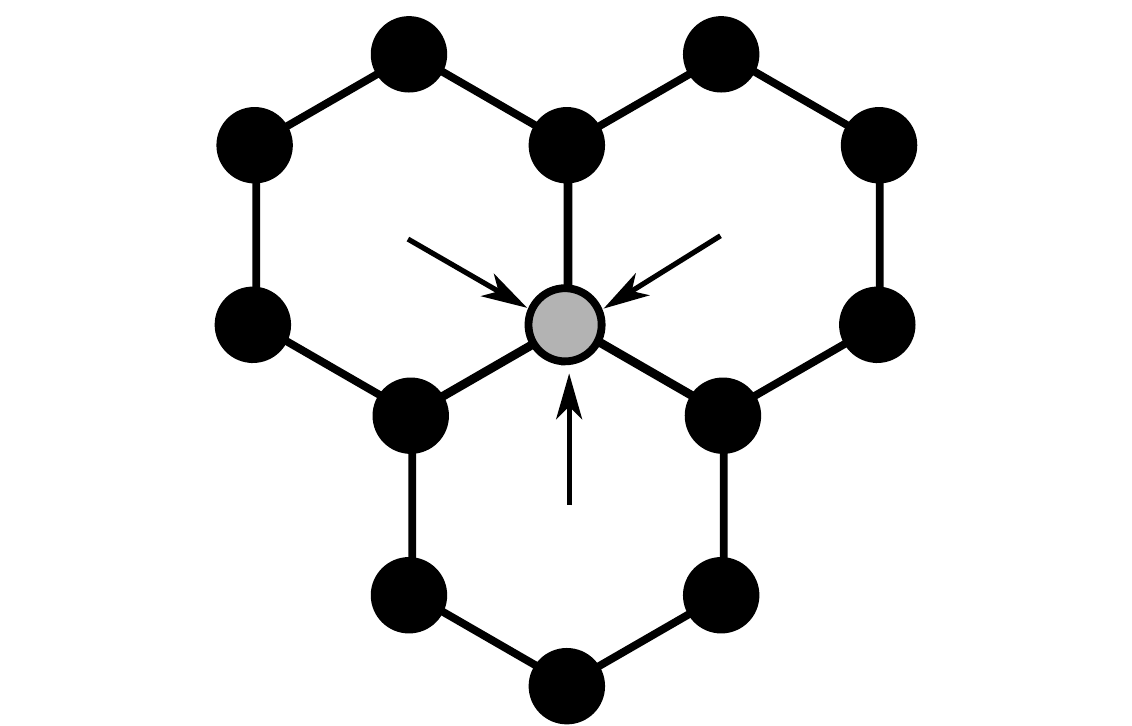} & 
\includegraphics[width=\rulewidth]{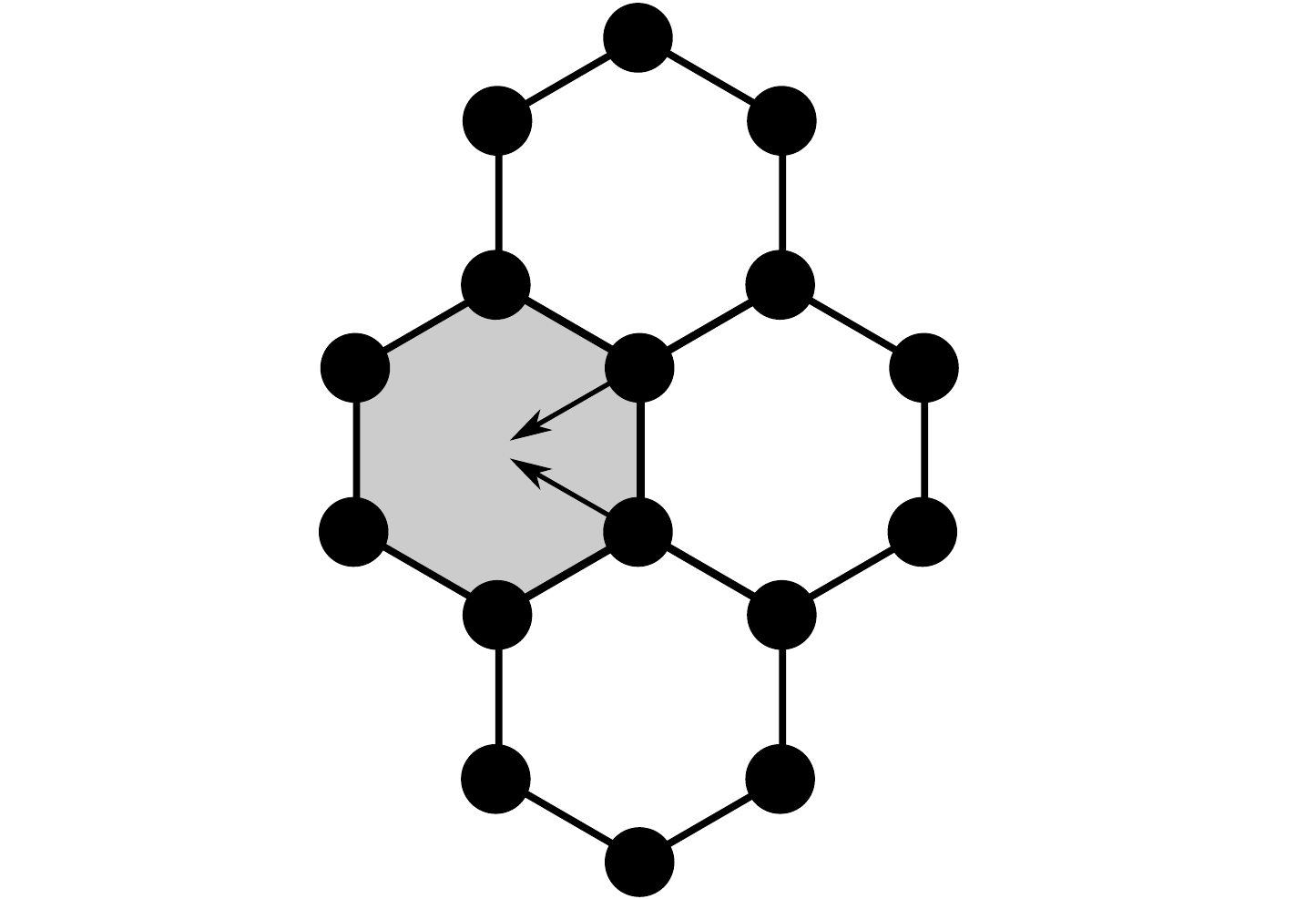}
 &
\includegraphics[width=\rulewidth]{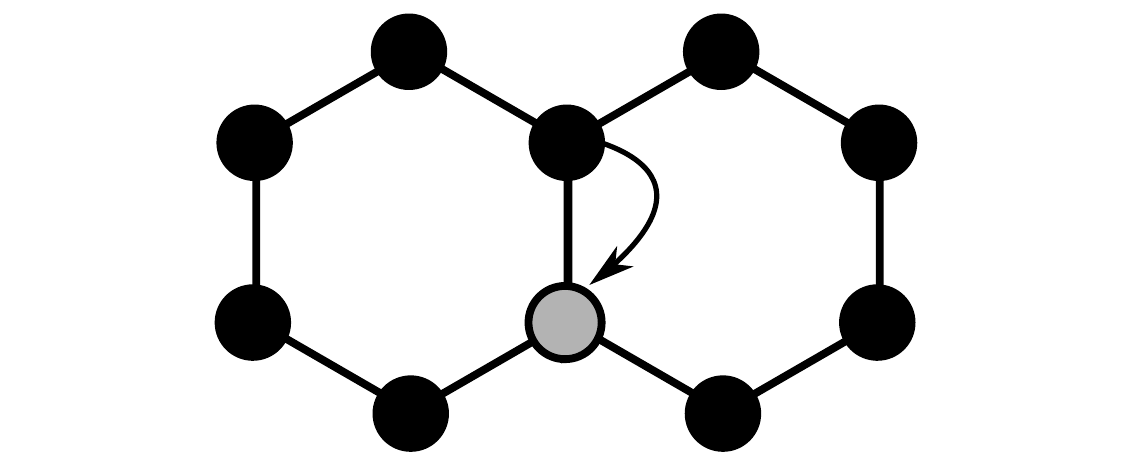} & 
\includegraphics[width=\rulewidth]{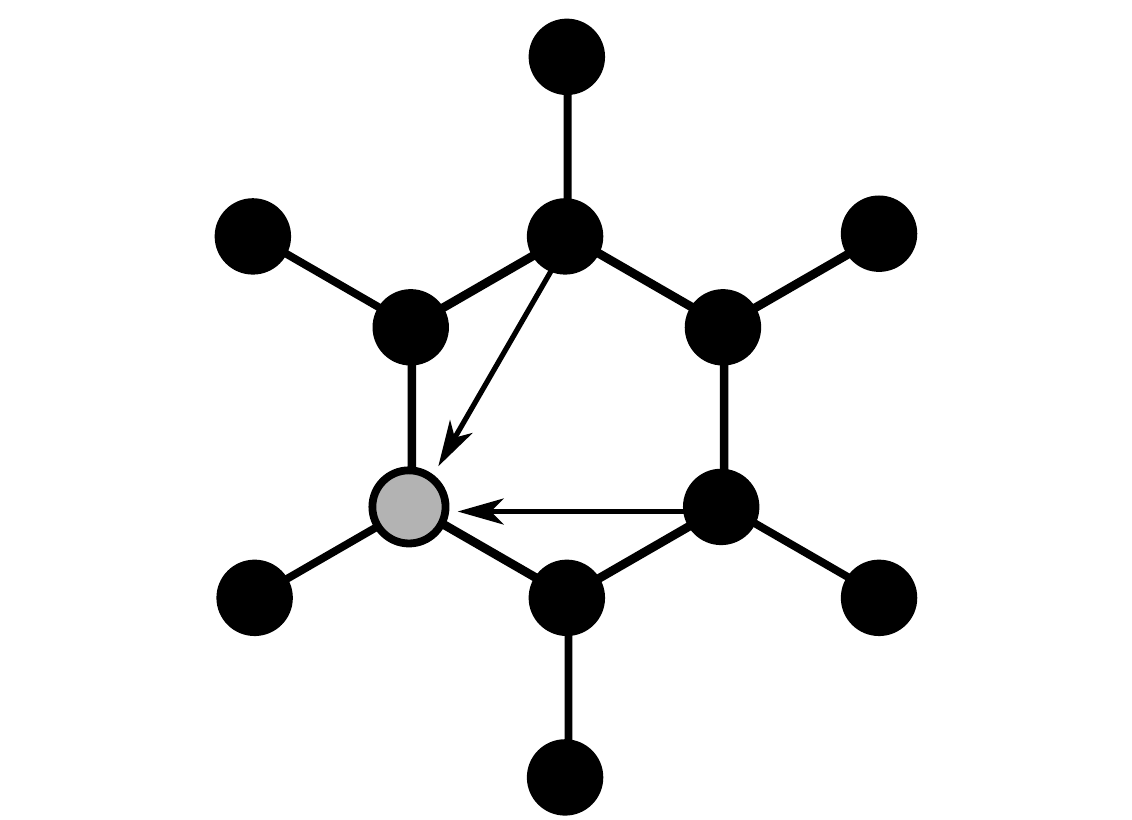} &
\includegraphics[width=\rulewidth]{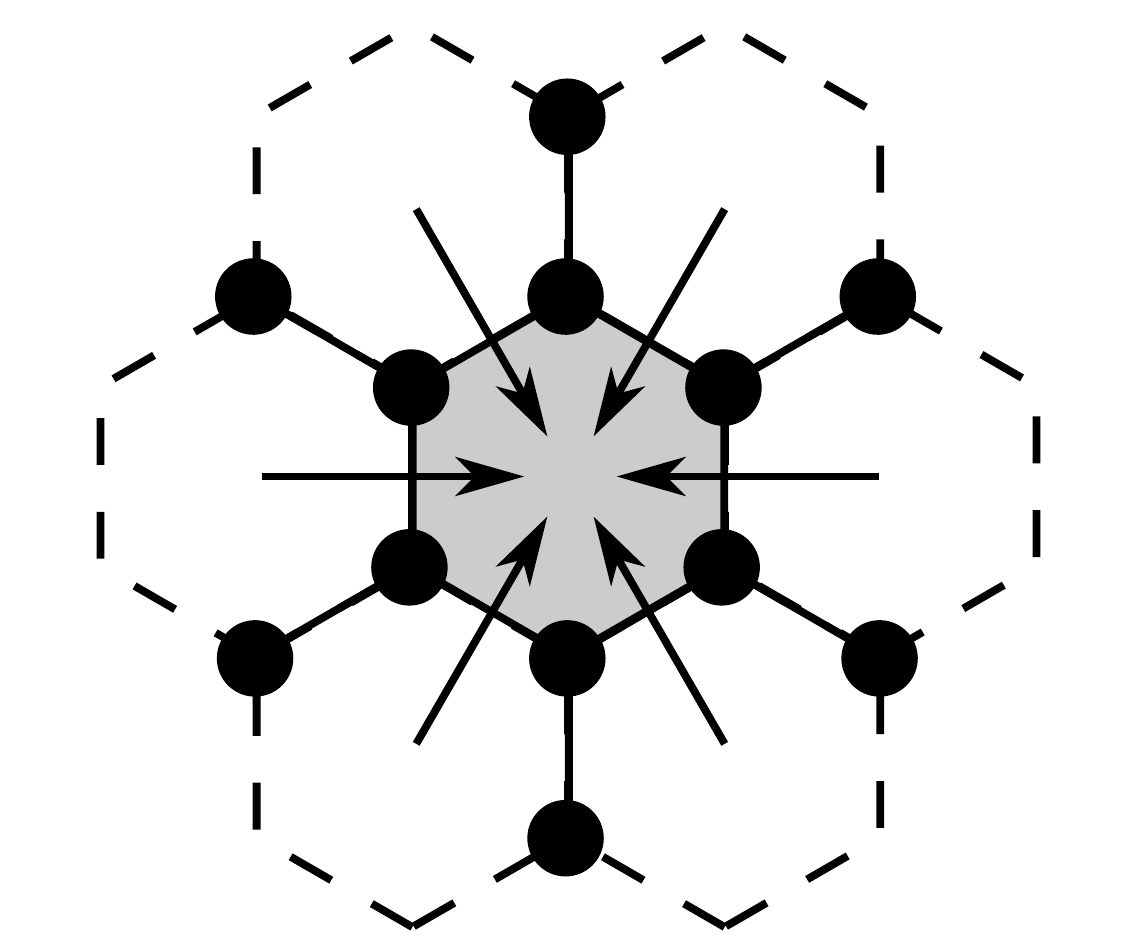} \\ 
\hline
%\hline
%\multicolumn{3}{|c||}{{Vertex-to-Vertex Rules}} & 
%\multicolumn{3}{c|}{{Face-to-Face Rules}}\\
%\hline
%\multicolumn{1}{|c|}{$C_1$} & 
%\multicolumn{1}{c|}{$C_2$} & 
%\multicolumn{1}{c||}{$C_3$} &
%\multicolumn{1}{c|}{$E_1$} & 
%\multicolumn{1}{c|}{$E_{1,3}$} & 
%\multicolumn{1}{c|}{$E_6$} \\
%\includegraphics[width=\rulewidth]{hexgrid-rule-c1.pdf} & 
%\includegraphics[width=\rulewidth]{hexgrid-rule-c2.pdf} & 
%\includegraphics[width=\rulewidth]{hexgrid-rule-c3.pdf} &
%\includegraphics[width=\rulewidth]{hexgrid-rule-f1.pdf} & 
%\includegraphics[width=\rulewidth]{hexgrid-rule-f13.pdf}& 
%\includegraphics[width=\rulewidth]{hexgrid-rule-6face.pdf} \\
%\hline
\end{tabular}
\caption{\label{fig:hexrulesex}Examples of Rules in the Hexagonal Grid.}
\end{figure}

It is possible to define an infinite number of discharging rules.
Note that some rules are inherently \emph{more complex} than others, so a partial ordering can be defined on the rules.
For instance, $V_1 \subset V_2 \subset N \subset V_3$, so $N$ is at least as effective as $V_2$.
We call attention to a few features of the discharging rules that should be balanced carefully in order to create the most effective rules.

\emph{Scope of Information.} The larger the configuration used for the discharging rule, the more information is known about the local environment of the chargeable object receiving charge. However, as the configuration $C$ grows, the number of realizations of the rule grows  approximately as $2^{|V(C)|}$, with some loss for symmetry and for avoiding forbidden configurations.

\emph{Range of Exchange.} Depending on the distance between $z$ and the $y_i$'s, charge can be exchanged across several distances. It may be beneficial to allow for charge to move longer distances, especially if it is possible to have large regions in $G$ where $X$ is much less dense than in other areas. 

\emph{Dependence.} For nearby chargeable objects, the configurations for different discharging rules overlap. Thus, some information is shared between the chargeable objects and that information can be used to assign value to the rule.

This step of creating a discharging argument is the step that requires the most amount of creativity and human intervention.
Creating interesting and effective rules is really where the proof author has most control, and this step is absolutely crucial in determining whether a discharging proof will provide a strong lower bound.
The strength of the rules must be balanced with the computational challenge of verifying their correctness, which is the topic of the next section.
Armed with a set of discharging rules, we can now define the algorithmic process for assigning value to the discharging rules.

\section{The Linear Program}\label{sec:lp}

The most difficult part of assigning value to discharging rules is verifying that objects of low charge receive enough charge to meet the goal charge while guaranteeing that objects of high charge do lot lose so much charge they drop below the goal charge.
In the contrapositive, it must be impossible to construct a configuration around a chargeable object where every discharging rule is evaluated and the resulting charge violates the goal charge.
Thus, we shall create a configuration $C$ around each chargeable object (up to isomorphism) such that $C$ contains the shape of each rule that can exchange charge to or from that object, and then generate each $\cF$-realization of $C$.
Every such realization determines which realizations of the discharging rules to use, and these values are combined to form a constraint in a linear program.

Recall that our goal requirement for vertices and faces are the following inequalities:
\begin{align*}
\mu'(v) &= \mu(v) + \sum_{u \in V(G)} D_X(u,v) + \sum_{g \in F(G)} D_X(g,v) \geq w\\
\nu'(f) &= \nu(f) + \sum_{u \in V(G)} D_X(u,f) + \sum_{g \in F(G)} D_X(g,f) \geq 0\\
\end{align*}

\vspace{-2em}
Since we are using a finite list of finite-sized discharging rules, this inequality will in fact use a finite number of nonzero terms.
Also, the amount of exchanged charge depends on a finite-sized local region about each chargeable object.
Given a grid $G$ and a list $R_1,\dots, R_m$ of discharging rules, the \emph{constraint configuration} about a chargeable object $x$ is defined as the set of faces and vertices that appear in an embedding $\pi(C)$ of a configuration $C=C(R_j)$ such that $\pi(z(R_j)) = x$ or $\pi(y_i(R_j)) = x$ for some $i$.
Observe that the constraint configurations about two chargeable objects, $x$ and $x'$, are isomorphic if and only if $x$ and $x'$ are in orbit within $G$.

For example, consider the rules $N$ and $J_2$ in the hexagonal grid.
Since the hexagonal grid is vertex-transitive and face-transitive, we only need to consider the constraint configurations for $v_0$ and $f_0$.
About $v_0$, there are three embeddings of $N$ and three embeddings of $J_2$ such that the vertices $z(N)$ and $z(J_2)$ are mapped to $v_0$. 
Together, these embeddings form a constraint configuration about $v_0$ consisting of all faces in $F_2(v_0)$, and all vertices incident to a face in $F_2(v_0)$.
About $f_0$, there are 18 embeddings of $N$ and six embeddings of $J_2$ such that one of the faces $y_1(N), y_2(N), y_3(N)$ or the face $y_1(J_2)$ are mapped to $f_0$.
Together, these embeddings form a constraint configuration about $f_0$ consisting of all faces in $F_1(f_0)$ and all vertices incident to a face in $F_1(f_0)$.
These constraint configurations are shown in Figure~\ref{fig:constraintex}.

\begin{figure}[htp]
\centering
\def\rulewidth{1in}
\begin{tabular}[h]{|m{1in}m{1in}|m{1in}m{1in}|}
\hline
\multicolumn{2}{|c|}{Rules} & 
\multicolumn{2}{c|}{Constraint Configurations} \\
\hline
\multicolumn{1}{|c}{$N$} &
\multicolumn{1}{c|}{$J_2$} & 
\multicolumn{1}{c}{Vertex} & 
\multicolumn{1}{c|}{Face} \\ 
\includegraphics[width=\rulewidth]{hexgrid-rule-n.pdf} & 
\includegraphics[width=\rulewidth]{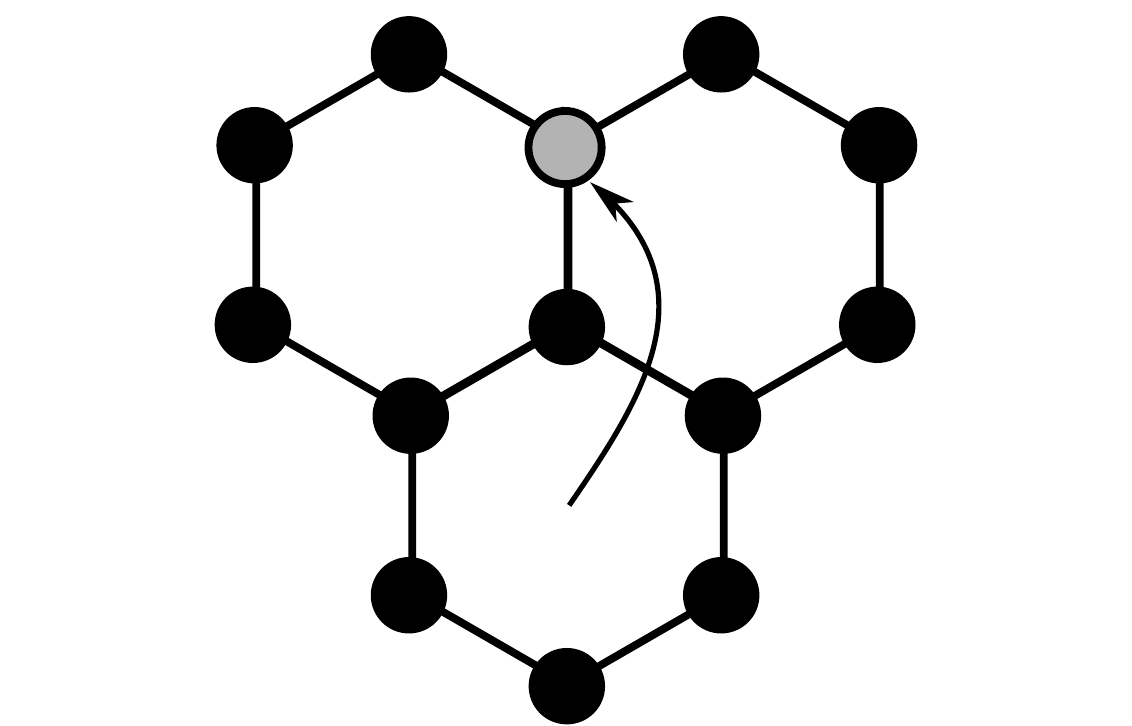} &
\includegraphics[width=\rulewidth]{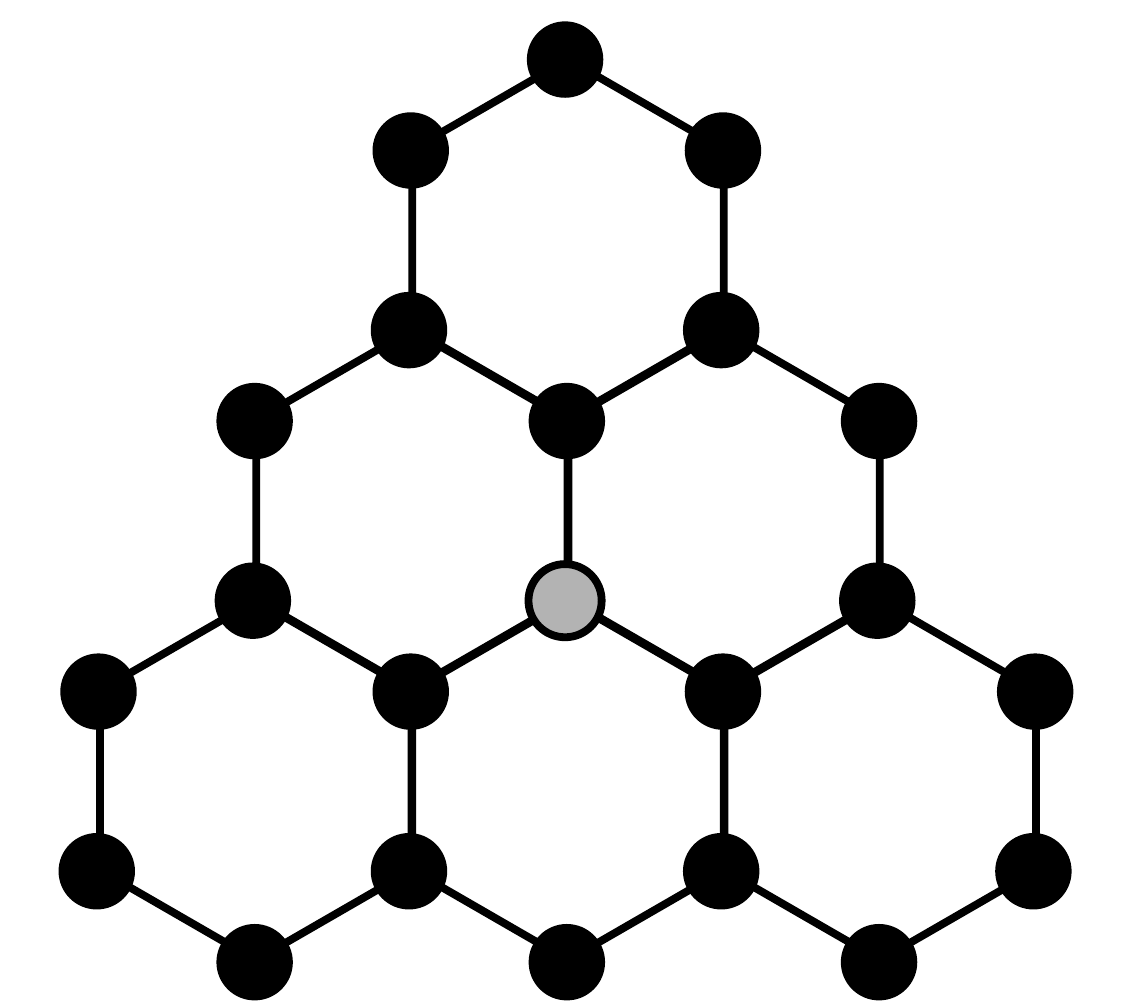} &
\includegraphics[width=\rulewidth]{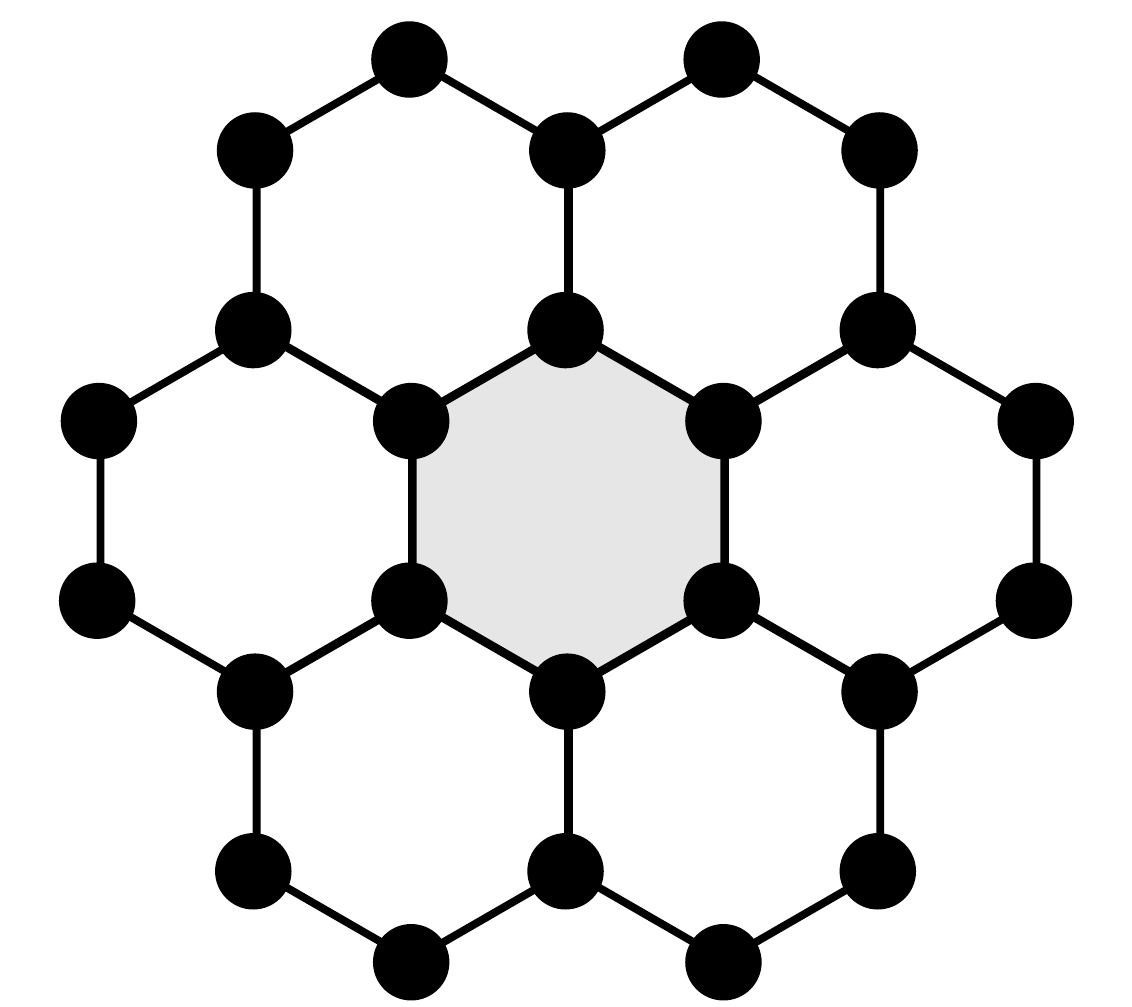} \\
\hline
\end{tabular}
\caption{\label{fig:constraintex}The rules $N$ and $J_2$ and the resulting constraint configurations.}
\end{figure}

Given a constraint configuration, the way the discharging rules assign value to the charge exchange is determined exactly by the way that $X$ intersects the vertices of this configuration.
Therefore, we generate all $\cF$-realizations of the constraint configuration.
Given an $\cF$-realization $C'$ of our constraint configuration $C$, we add a constraint to our linear program.

Suppose we are using the rules $R_1,\dots,R_k$.
If $C$ is a constraint configuration centered on a vertex $v$ and $C'$ is an $\cF$-realization of $C$, then we enforce that $\mu'(v) \geq w$ after the discharging process is complete by adding the constraint
\[
	\mu(v) + \sum_{j=1}^k \sum_{\pi : \pi(z(R_j)) = v} \sum_{i} \sigma_j(S_1(C'), y_i) -  \sum_{j=1}^k\sum_{i, \pi : \pi(y_i(R_j)) = v} \sigma_j(S_1(C'), y_i)  \geq w
\]
to the linear program.
If $C$ is a constraint configuration centered on a face $f$ and $C'$ is an $\cF$-realization of $C$, then we enforce that $\nu'(f) \geq w$ after the discharging process is complete by adding the constraint
\[
	\nu(f) + \sum_{j=1}^k \sum_{\pi : \pi(z(R_j)) = f} \sum_{i} \sigma_j(S_1(C'), y_i) - \sum_{j=1}^k \sum_{i, \pi : \pi(y_i(R_j)) = f} \sigma_j(S_1(C'), y_i)  \geq 0
\]
to the linear program.

Observe that whenever these constraints are satisfied, the discharging argument demonstrates a lower bound of $\delta(X) \geq w$ for any $X \in \forb(\cF)$.
In order to produce the largest lower bound, use $\max w$ as the optimization function of the linear program.

Thus, we have a complete description of an adage proof.
In summary, the three main steps are:
(1) Define a set of rules $R_1,\dots,R_k$, and generate all $\cF$-realizations $C'$ of their configurations, mapping the values $\sigma_j(S_1(C'),y_i)$ to a list of variables;
(2) For each constraint configuration (up to isomorphism), generate all $\cF$-realizations and add the resulting constraint to the linear program;
(3) Solve the linear program to determine the values $\sigma_j(S_1(C'),y_i)$ and the lower bound $w$.

These steps were implemented and executed for several sets of rules, which are shown along with their constraint configurations in Appendix~\ref{sec:ruletables}.
All software and data are available online\footnote{See {\scriptsize\url{http://www.math.iastate.edu/dstolee/r/adage.htm}} for all software and data.}.
The following theorem implies Theorem~\ref{thm:main}.

\begin{theorem}
Let $X$ be an identifying code in the hexagonal grid.
The adage proof using rule $N$ demonstrates $\delta(X) \geq \frac{23}{55} \approx 0.4181818$.
\end{theorem}

The ADAGE framework as described is not tied to any specific grid or family of forbidden configurations $\cF$.
In the next section, we discuss variations on identifying codes and summarize the adage proofs of sharp lower bounds for those variations.

\section{Variations}
\label{sec:variations}

Due to the modular development of the ADAGE framework for grids, the components for the grid and the forbidden configurations can be interchanged. 
This allows for adage proofs to be constructed for the hexagonal, square, and triangular grids.
Several discharging rules and corresponding constraint configurations for these grids are demonstrated in Tables~\ref{tab:hexagonalrules}, \ref{tab:squarerules}, and \ref{tab:triangularrules}. 
More planar grids could be implemented and used, including those that are not vertex- or face-transitive, such as the hexagon-triangle grid.

There are several variations of an identifying code, each with its own application to fault-detection in computer networks.
A set $X \subset V(G)$ matches these variations if  the following constraints are satisfied:

\begin{cit}
\item Dominating Set: $N[v] \cap X \neq \varnothing$ for all $v \in V(G)$.

\item Identifying Code: $N[v] \cap X \neq \varnothing$ and $(N[v] \cap X) \neq (N[u] \cap X)$ for all distinct $u, v \in V(G)$.

\item Strong Identifying Code: $N[v] \cap X \neq \varnothing$ and $\{ N[v] \cap X, N(v) \cap X\} \cap \{ N[u] \cap X, N(u)\cap X\} = \varnothing$ for all distinct $u, v \in V(G)$ (see~\cite{honkala2010optimal,honkala2002strongly}).

\item Locating-Dominating Code: $N(v) \cap X \neq \varnothing$ for $v \notin X$, and $N(v) \cap X \neq N(u) \cap X$ for all distinct $u, v \in V(G) \setminus X$ (see~\cite{caceres2013locating,honkala2006optimal,honkala2006locating,slater1995locating}).

\item Open-Locating-Dominating (OLD) Code: $N(v) \cap X \neq \varnothing$ and $N(u) \cap X \neq N(v) \cap X$ for all distinct $u, v \in V(G)$ (see~\cite{kincaid2014optimal,seo2010open}).

%\item Neighbor-Identifying Code: $N[u] \cap X \neq N[v] \cap X$ for $uv \in E(G)$.
\end{cit}

All of these variations are implemented in the current version of ADAGE on grids.
Several collections of discharging rules were used to find adage proofs of lower bounds on these variations, and the results can be found in Table~\ref{tab:bounds}.
We summarize the sharp bounds below, with attribution to the first authors to find such bounds.
See Appenix~\ref{sec:bounds} for lower bounds demonstrated by other rule sets.

\begin{theorem}[Ben-Haim and Litsyn~\cite{ben2005exact}]
Let $X$ be an identifying code in the square grid.
The adage proof using the rule $V_2$ demonstrates $\delta(X) \geq \frac{7}{20}$.
\end{theorem}

\begin{theorem}[Karpovsky, Chakrabarty, and Levitin~\cite{karpovsky1998new}]
Let $X$ be an identifying code in the triangular grid.
The adage proof using the rule $V_1$ demonstrates $\delta(X) \geq \frac{1}{4}$.
\end{theorem}

\begin{theorem}[Honkala~\cite{honkala2006optimal}]
Let $X$ be a locating-dominating code in the hexagonal grid.
The adage proof using the rule $V_2$ demonstrates $\delta(X) \geq \frac{1}{3}$.
\end{theorem}

\begin{theorem}[Slater~\cite{slater2002fault}]
Let $X$ be a locating-dominating code in the square grid.
The adage proof using the rule $C_1$ demonstrates $\delta(X) \geq \frac{3}{10}$.
\end{theorem}

\begin{theorem}[Seo and Slater~\cite{seo2010open}]
Let $X$ be an open-locating dominating code in the hexagonal grid.
The adage proof using the rule $V_2$ demonstrates $\delta(X) \geq \frac{1}{2}$.
\end{theorem}

\begin{theorem}[Seo and Slater~\cite{seo2010open}]
Let $X$ be an open-locating dominating code in the square grid.
The adage proof using the rule $C_1^+$ demonstrates $\delta(X) \geq \frac{2}{5}$.
\end{theorem}

\begin{theorem}[Kincaid, Oldham, and Yu~\cite{kincaid2014optimal}]
Let $X$ be an open-locating dominating code in the triangular grid.
The adage proof using the rule $C_1^+$ demonstrates $\delta(X) \geq \frac{4}{13}$.
\end{theorem}

Observe that among all variations on all three grids, the only variations that failed to find a sharp lower bound were identifying codes on the hexagonal grid, and strong identifying codes on all three grids.
% (\textbf{TODO} Locating-dominating sets in triangular grid???).
Likely, the strong identifying codes are more challenging because a strong identifying code is both an identifying code and an open-locating dominating code, so the optimal density is highest among all of these variations.
Also, there are more forbidden configurations and this leads to fewer realizations of the discharging rules (and hence fewer variables in the linear program).

There are also variations on identifying codes that are robust against edge changes~\cite{honkala2004optimal,honkala2006robust,honkala2007identifying,laihonen2006robust,slater2002fault}, or identify all sets of vertices of size at most $\ell$~\cite{foucaud2011improved,laihonen2005optimal,laihonen2006optimal}, or consider balls of larger radius~\cite{junnila2013new,junnila2012optimal,karpovsky1998new,martin2010lower,roberts2008locating,stanton2011improved,ville2011optimality}.
These variations are good candidates for future implementation.

\section{Conclusions and Future Work}

This first application of the ADAGE framework is successful in showing alternative proofs of existing sharp bounds~\cite{ben2005exact,karpovsky1998new,honkala2006locating,slater2002fault,seo2010open,kincaid2014optimal}, and surpassing the human-written proofs of lower bounds on identifying codes in the hexagonal grid~\cite{karpovsky1998new,cohen2000bounds,cranston2009new,cukierman2013new}.
The computer-automated portions of the method replace lengthy case analysis and can be more detailed than something within the reach of a human prover.
However, the simple description of the discharging rules can perhaps lead to a deeper understanding of the structure and success of a discharging argument.
By automating the process of assigning value to the discharging rules, a proof author can focus more on the creative process of designing rules. 
Thus, the most important part is to balance the strength of the rules against the size of the constraint configurations.

There are some features that will be added to the ADAGE framework in order to make the proofs more robust.
The rules used so far are based entirely on the realization of the rule configuration.
This leads to an exponential growth in the number of variables and constraints as the rules grow.
To lower the number of variables, the rules could be clustered by families of realizations.
For example, a rule could be based on the number of elements incident to a face instead of the exact arrangement of elements on the face.
This is equivalent to placing equality constraints among groups of variables coming from similar realizations.
Such a clustering of rules can also greatly decrease the number of distinct constraints, as several realizations of a constraint configuration will result in the same combination of variables.

Another feature is to use the discharging argument to characterize sharp examples.
If the discharging proof presents a sharp lower bound on the density of a set, then we can use this to generate a class of optimal examples. 
Among all optimal sets, the configurations of optimal density must be those where the discharging arguments are sharp, except at a density-zero portion of the chargeable objects. 
Thus, it must be possible to construct arbitrarily large configurations that do not contain a forbidden configuration and the discharging rules result in charge exactly $w$ on every internal vertex and exactly $0$ on every internal face.
A combinatorial generation algorithm could discover such configurations.
%once the values of the rules are set.

%The author has plans to modify this technique to prove results on planar graphs, sparse graphs, and finite grids.
%The ADAGE framework can be stated generally enough to encompass all of these applications, but each requires its own software implementation.

\section*{Acknowledgements}

Thanks to Michael Ferrara, Stephen G. Hartke, Bernard Lidick\'y, Ryan R. Martin, and Paul S. Wenger for several very helpful discussions about the discharging method and identifying codes.

\label{paper:end}

{
\small
\frenchspacing
\setlength{\parsep}{0em}
\setlength{\itemsep}{0em}

%\nocite{*}

\bibliographystyle{abbrv}
\bibliography{identcodes.bib}
}

\clearpage
\appendix
\section{Discharging Rules and Constraint Configurations}\label{apx:begin}\label{sec:ruletables}

\begin{table}[hp]\small
\centering
\def\rulefigwidth{0.85in}
\def\constraintwidth{0.85in}
\def\bigconstraintwidth{1.7in}
\def\medconstraintwidth{1.3in}
\mbox{
\begin{tabular}[h]{|m{0.1in}m{\rulefigwidth}|m{\constraintwidth}m{\constraintwidth}|}
\hline
\multicolumn{2}{|c|}{\textbf{Rules} }
&
\multicolumn{2}{c|}{\textbf{Constraint Configurations}}\\
\hline&&&\\[-3ex]
\hline&&&\\[-2ex]
$V_1$ & \includegraphics[width=\rulefigwidth]{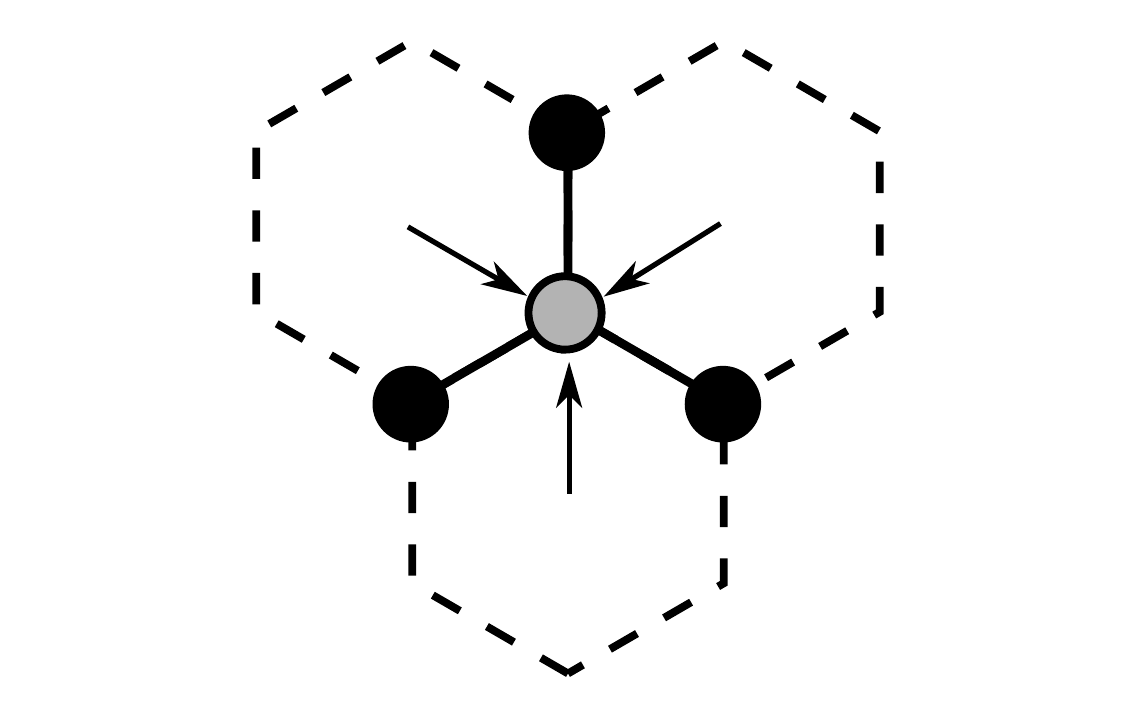} & 
\includegraphics[width=\constraintwidth]{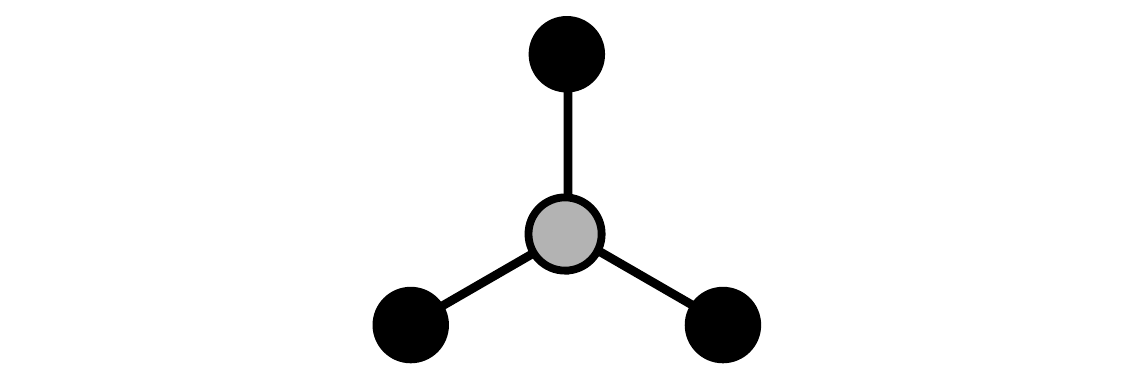} & 
\includegraphics[width=\constraintwidth]{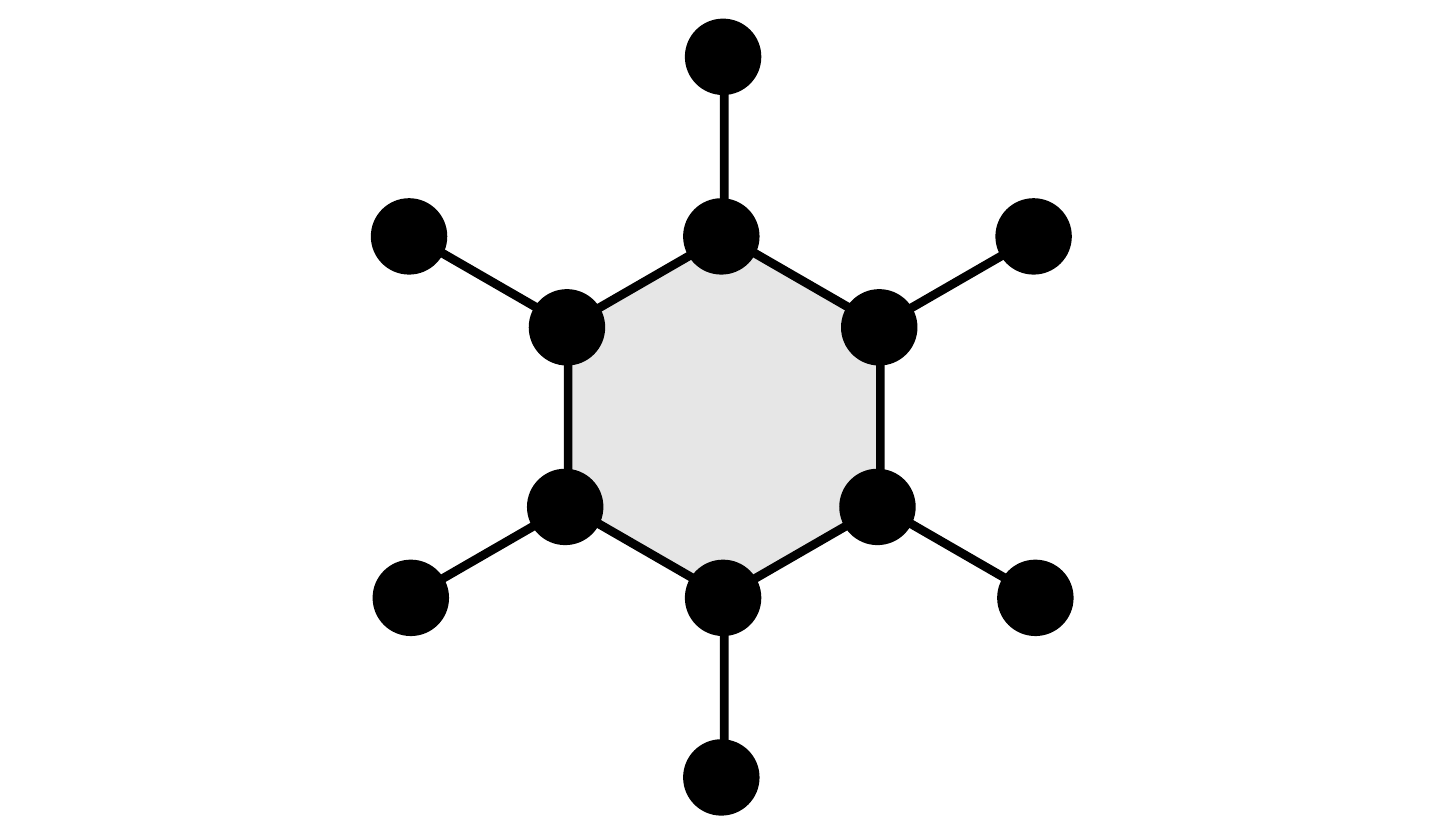} \\
\hline&&&\\[-2ex]
$V_2$ & \includegraphics[width=\rulefigwidth]{hexgrid-rule-v2.pdf} & 
\includegraphics[width=\constraintwidth]{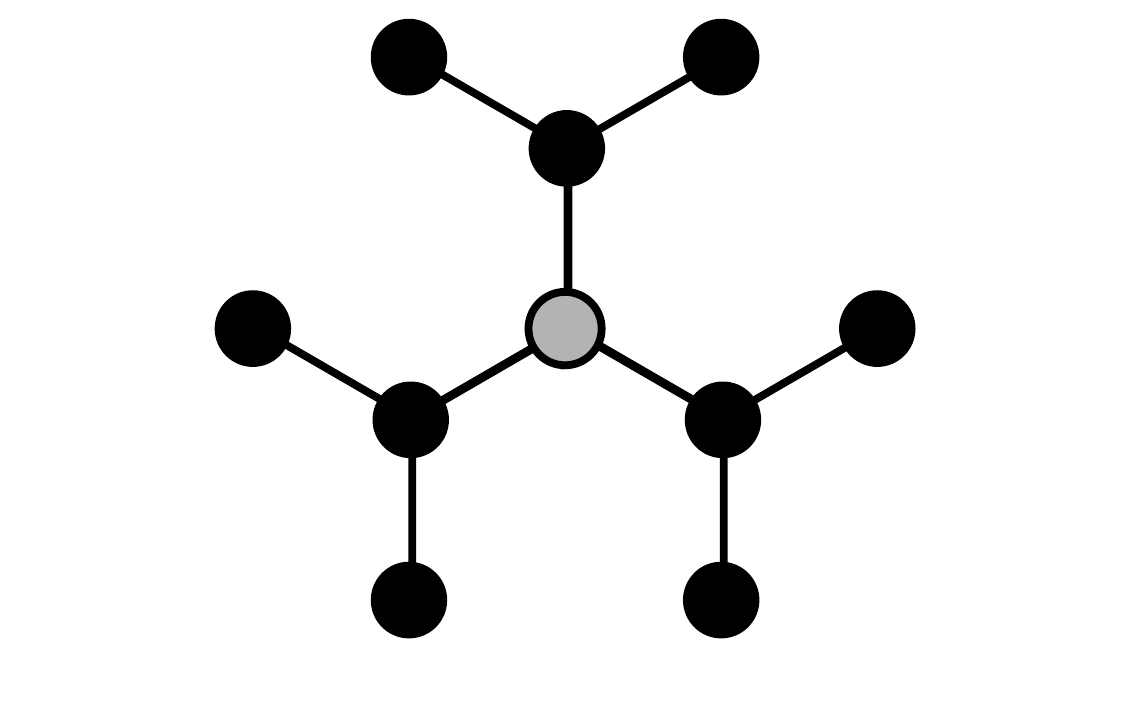} & 
\includegraphics[width=\constraintwidth]{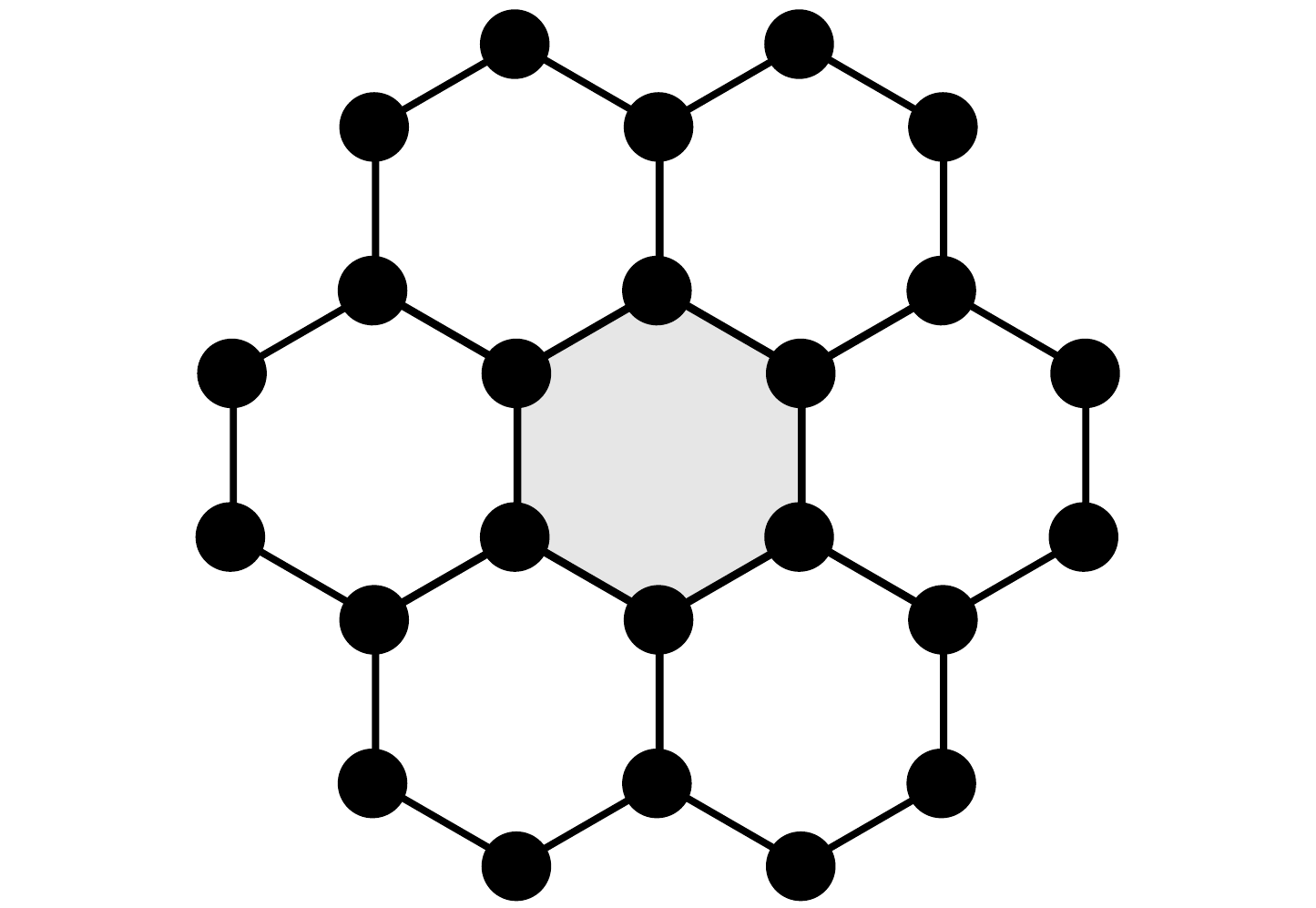} \\
\hline&&&\\[-2ex]
$N$ & \includegraphics[width=\rulefigwidth]{hexgrid-rule-n.pdf} & 
\includegraphics[width=\constraintwidth]{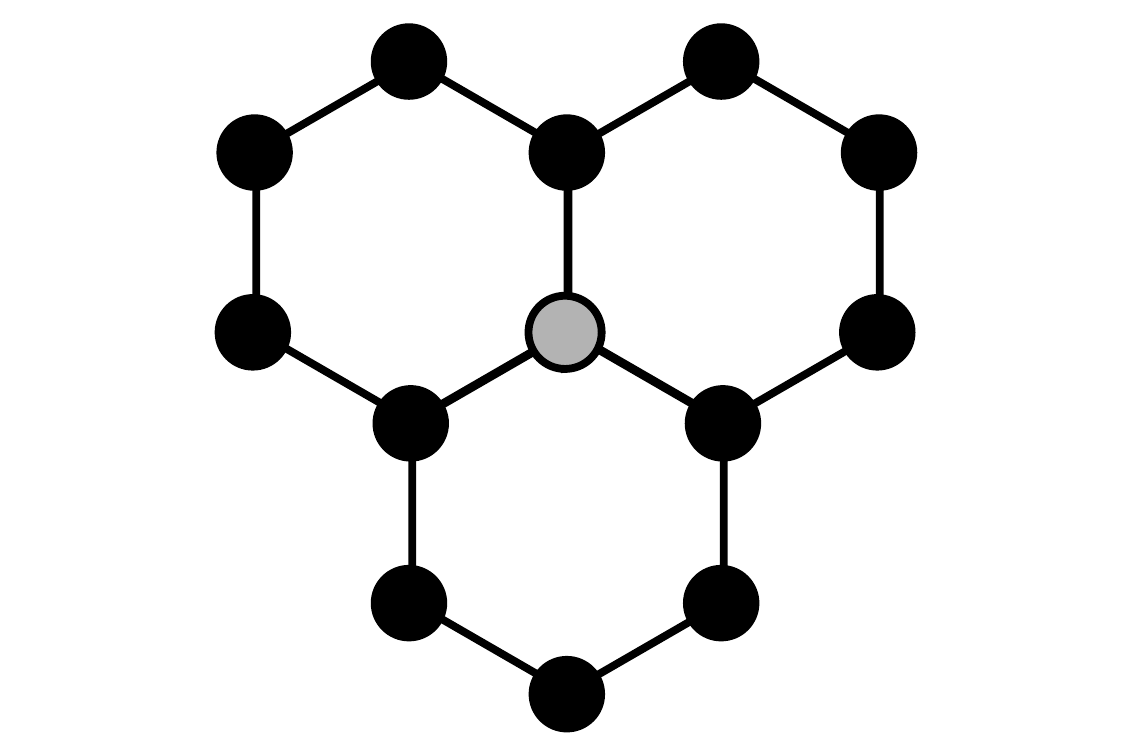} & 
\includegraphics[width=\constraintwidth]{hexgrid-constraint-face-n.pdf}  \\
\hline&&&\\[-2ex]
$N^+$ & \includegraphics[width=\rulefigwidth]{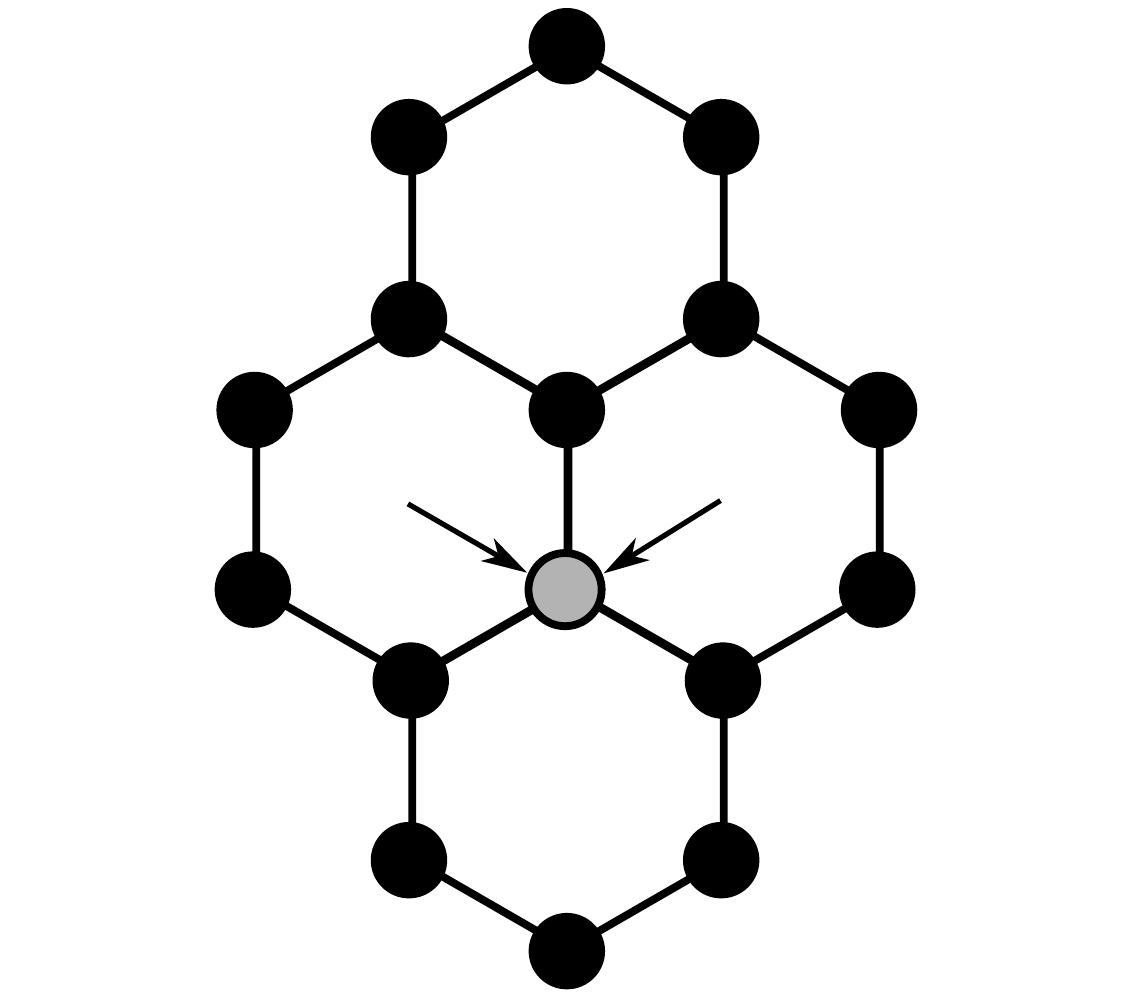} & 
\includegraphics[width=\constraintwidth]{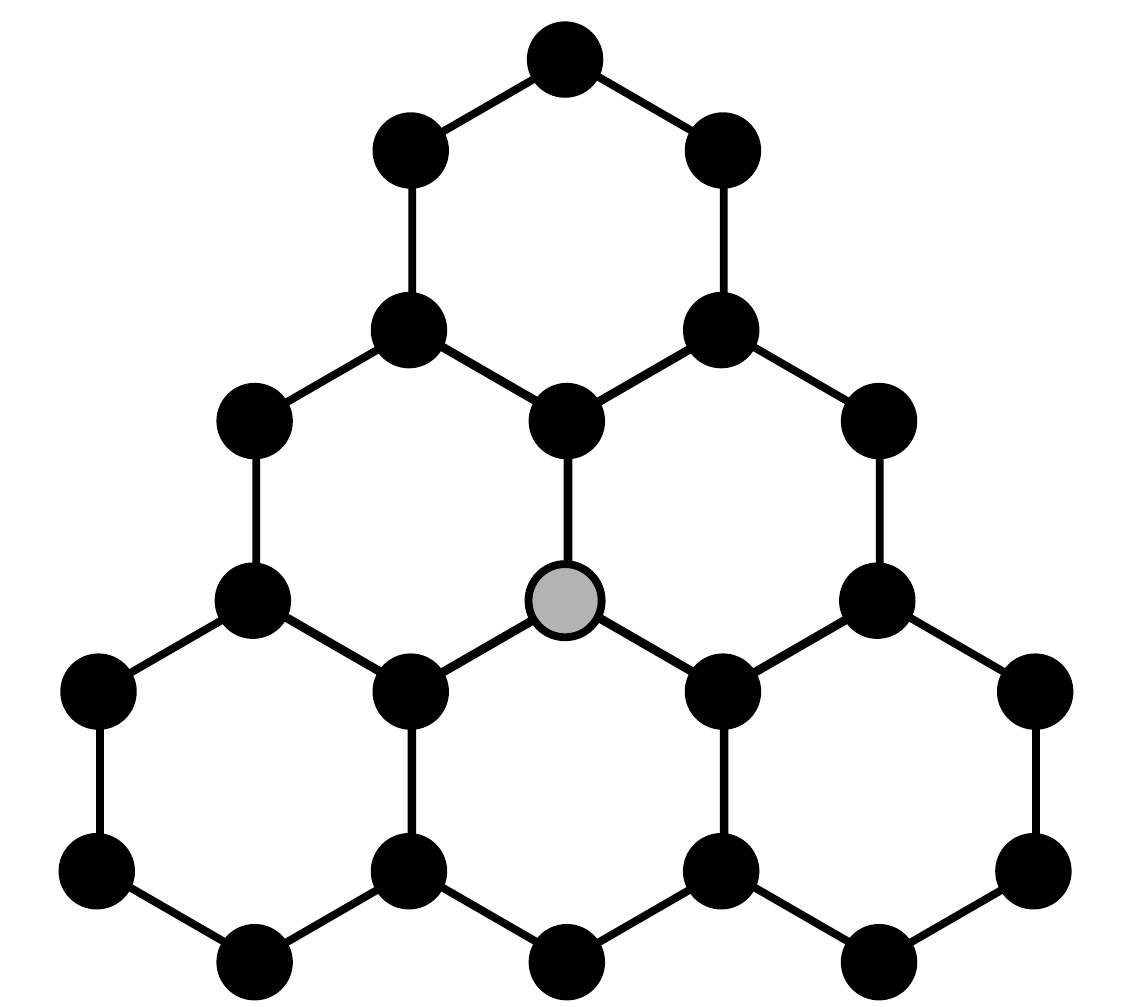} & 
\includegraphics[width=\constraintwidth]{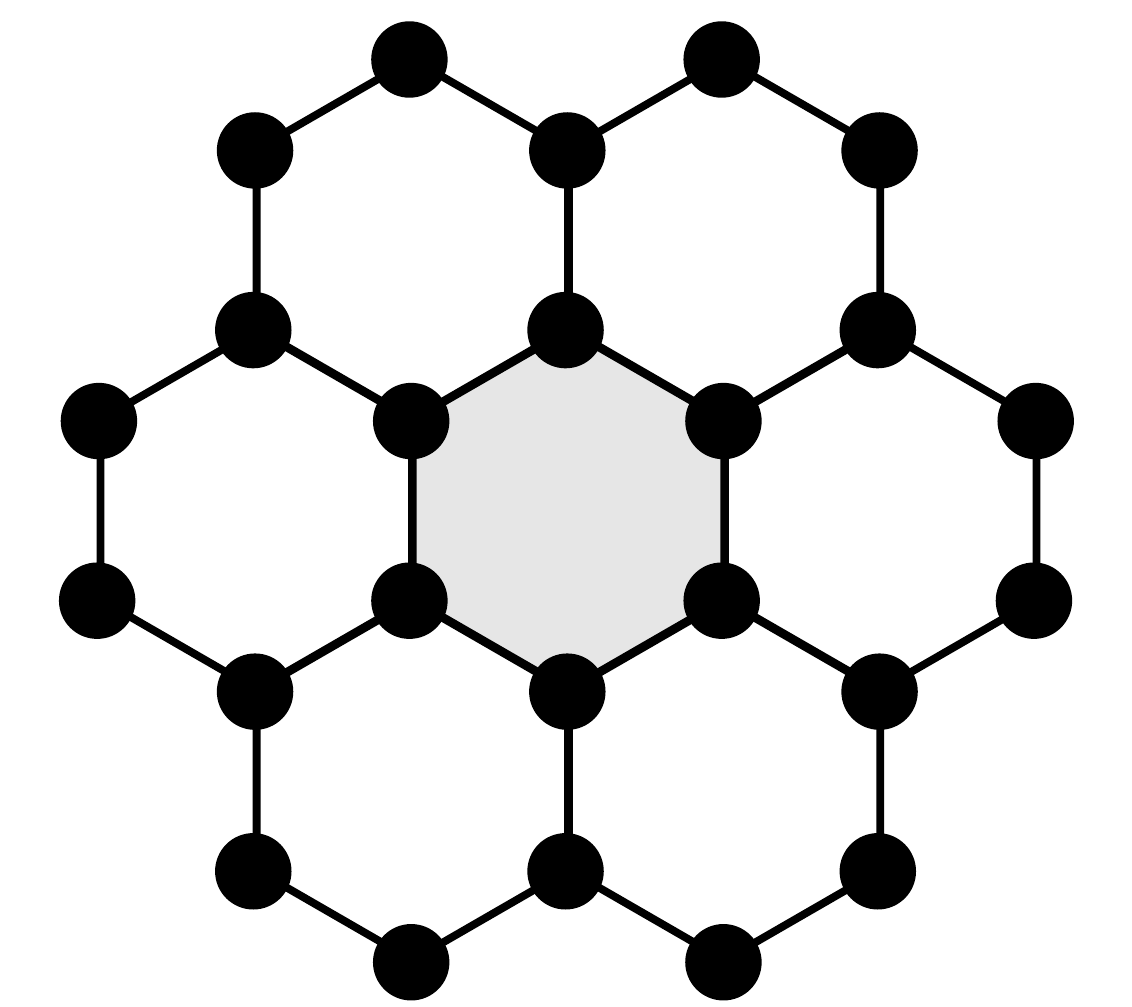}  \\
\hline&&&\\[-2ex]
$V_3$ & \includegraphics[width=\rulefigwidth]{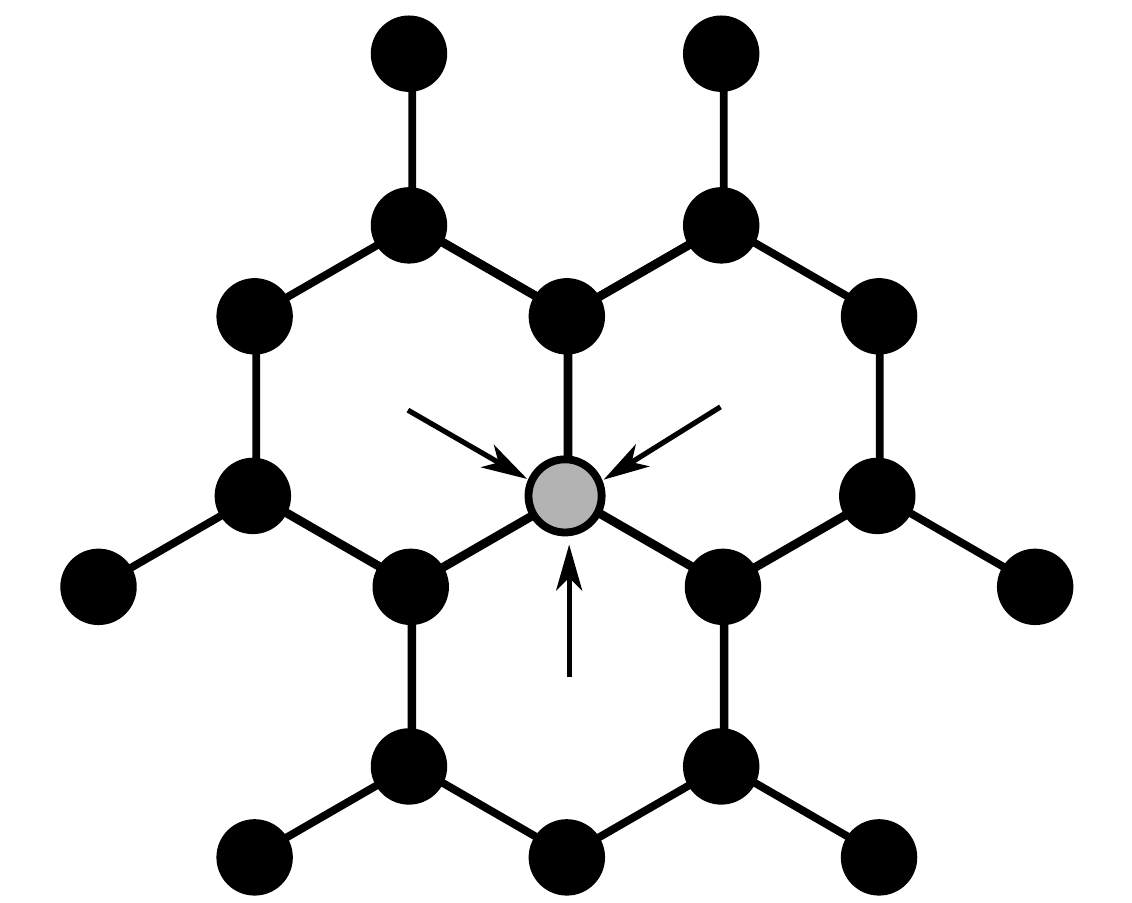} & 
\includegraphics[width=\constraintwidth]{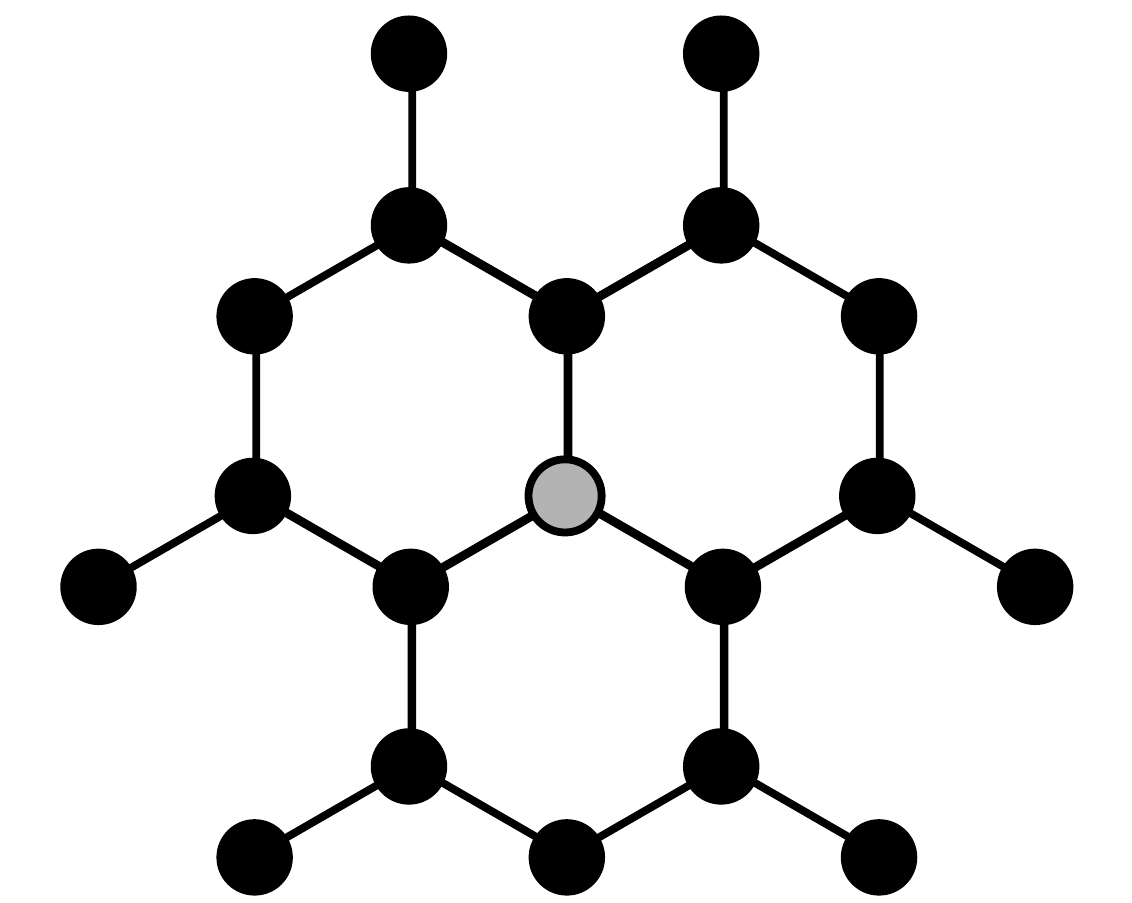} & 
\includegraphics[width=\constraintwidth]{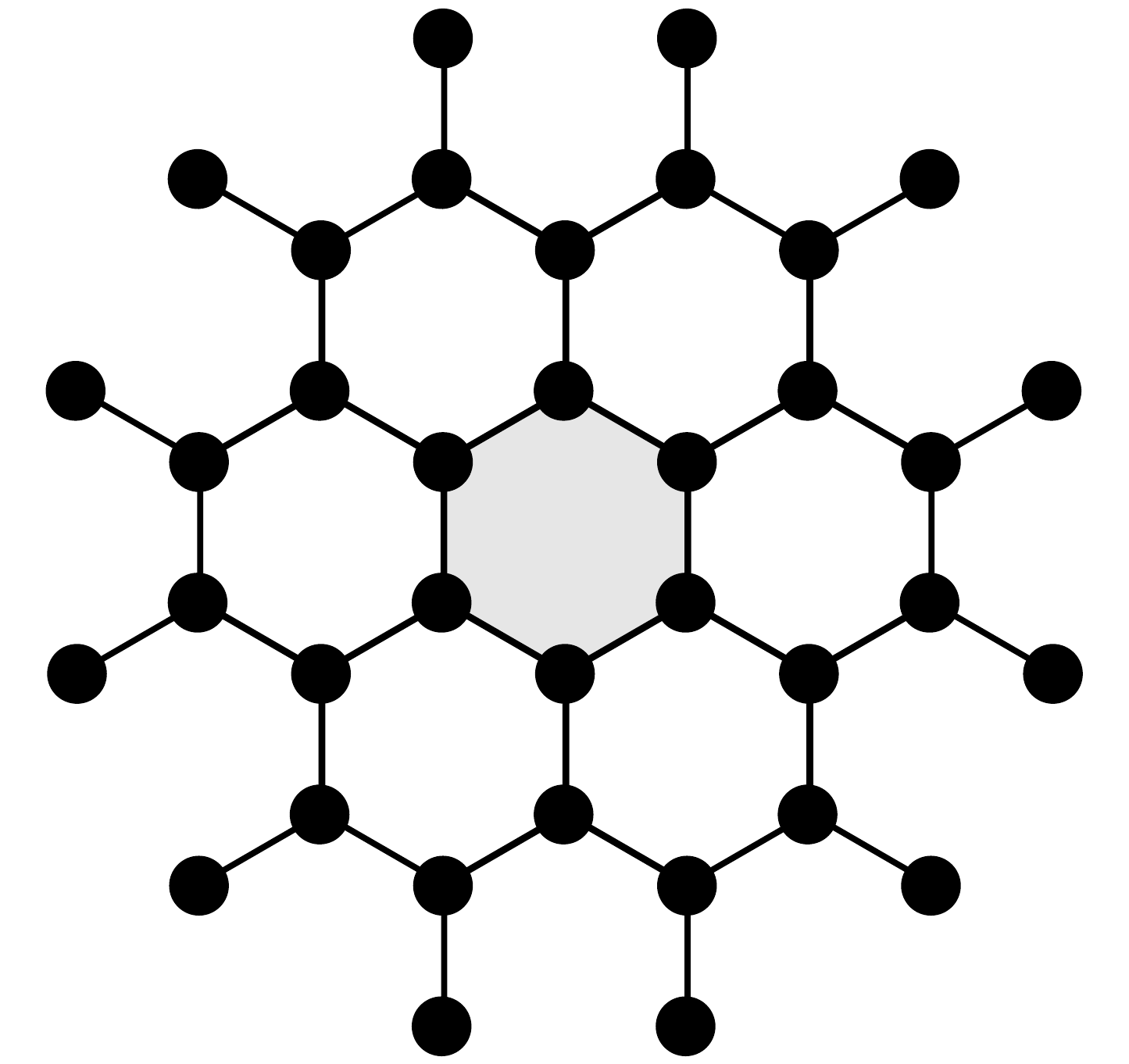} \\
\hline&&&\\[-2ex]
$C_1$ & \includegraphics[width=\rulefigwidth]{hexgrid-rule-c1.pdf} & 
\multicolumn{2}{c|}{\multirow{3}{\bigconstraintwidth}{\includegraphics[width=\bigconstraintwidth]{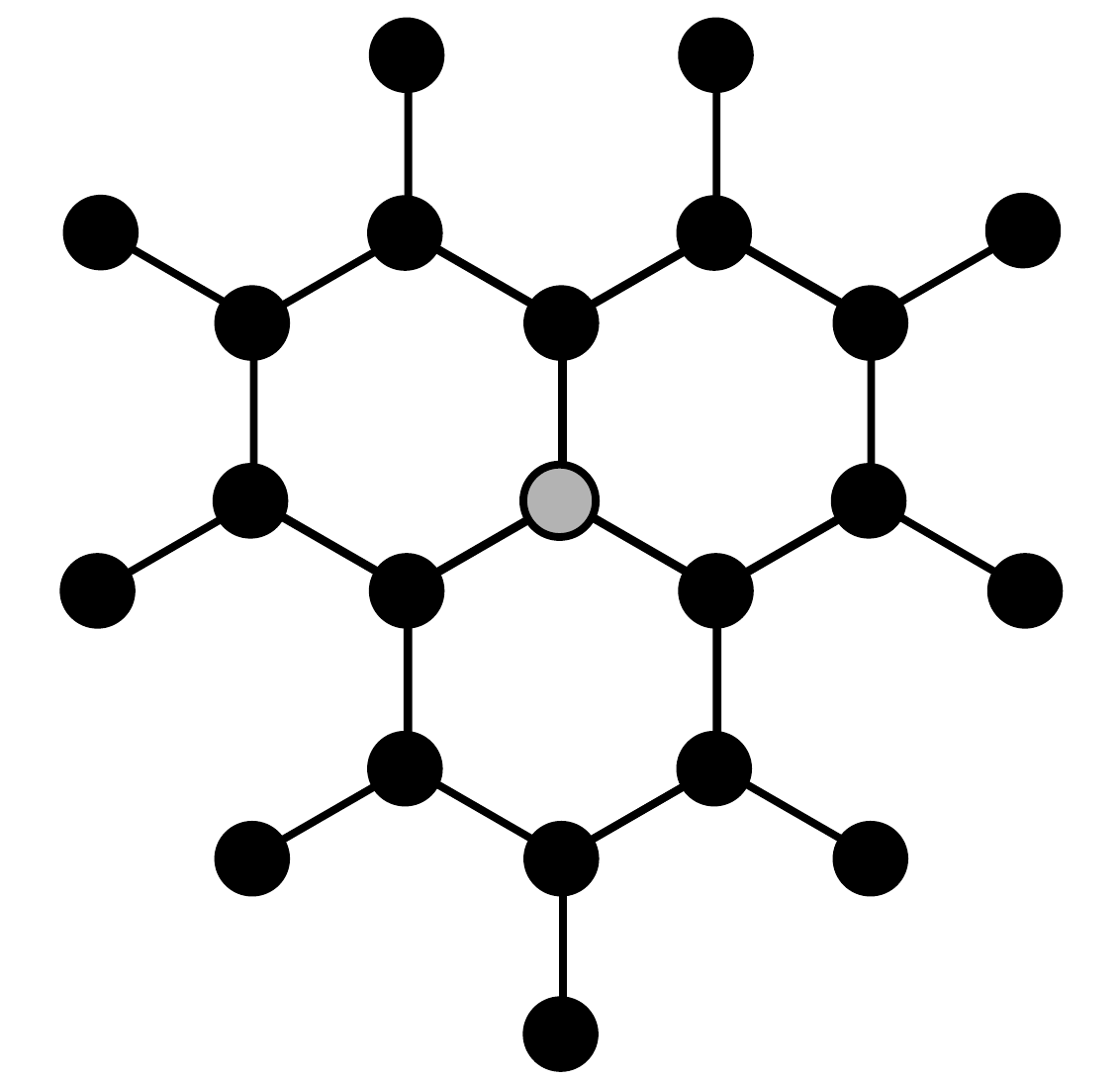}}} \\[4ex]
$C_2^+$ & \includegraphics[width=\rulefigwidth]{hexgrid-rule-c2.pdf} && \\[4ex]
$C_3$ & \includegraphics[width=\rulefigwidth]{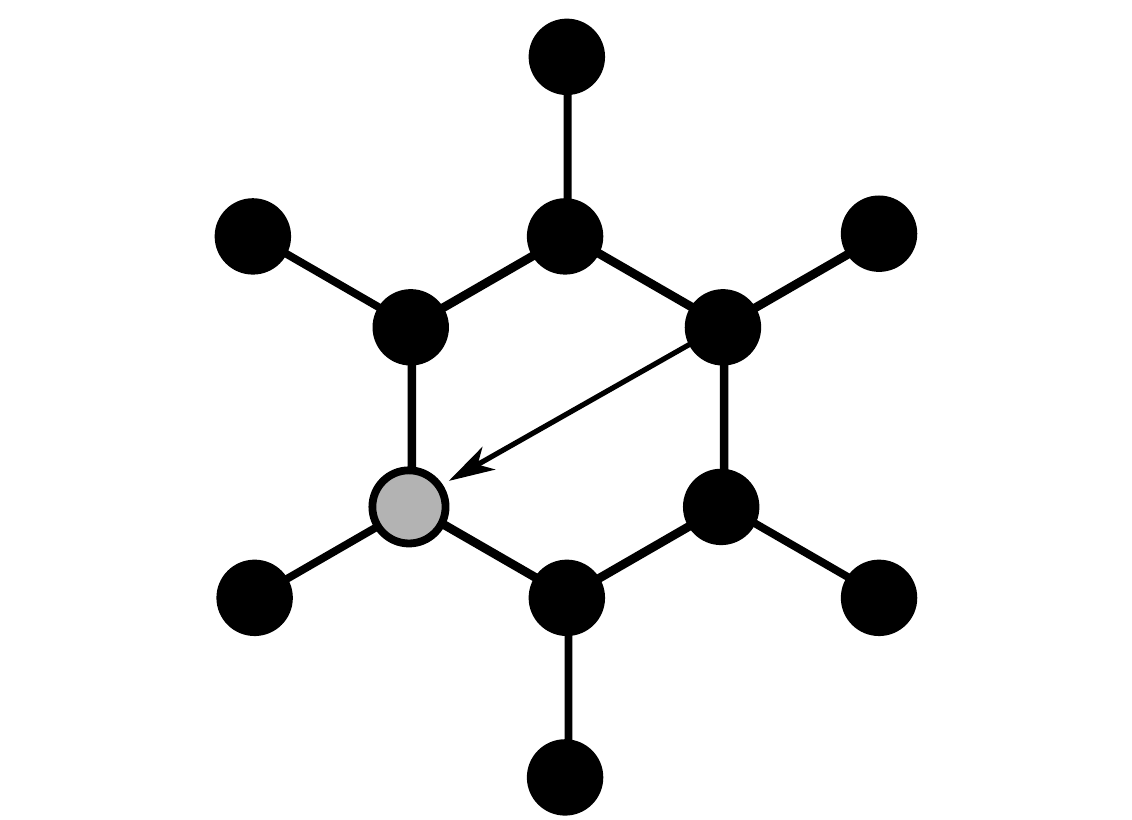} && \\[4ex]
\hline
\end{tabular}
\quad
\begin{tabular}[h]{|m{0.1in}m{\rulefigwidth}|m{\constraintwidth}m{\constraintwidth}|}
\hline
\multicolumn{2}{|c|}{\textbf{Rules} }
&
\multicolumn{2}{c|}{\textbf{Constraint Configurations}}\\
\hline&&&\\[-3ex]
\hline&&&\\[-2ex]
$N$ & \includegraphics[width=\rulefigwidth]{hexgrid-rule-n.pdf} & 
\multirow{2}{\constraintwidth}{\includegraphics[width=\constraintwidth]{hexgrid-constraint-vert-nj2.pdf}} &
\multirow{2}{\constraintwidth}{\includegraphics[width=\constraintwidth]{hexgrid-constraint-face-n.pdf}} \\
$J_2$ & \includegraphics[width=\rulefigwidth]{hexgrid-rule-j2.pdf} && \\
\hline&&&\\[-2ex]
$V_3$ & \includegraphics[width=\rulefigwidth]{hexgrid-rule-v3.pdf} & 
\multirow{2}{\constraintwidth}{\includegraphics[width=\constraintwidth]{hexgrid-constraint-vert-v3}} &
\multirow{2}{\constraintwidth}{\includegraphics[width=\constraintwidth]{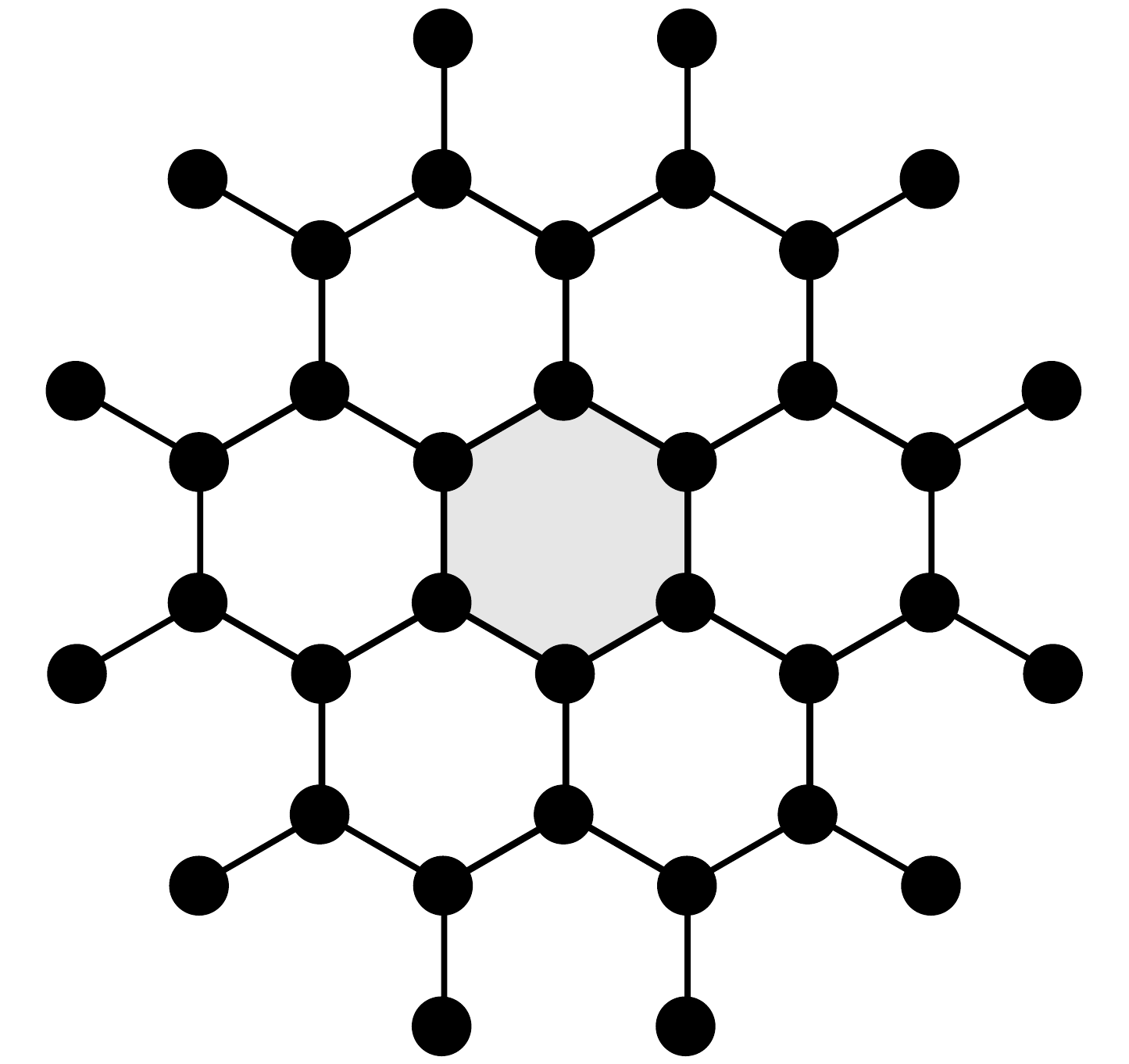}} \\
$E_6$ & \includegraphics[width=\rulefigwidth]{hexgrid-rule-6face.pdf} && \\
\hline&&&\\[-2ex]
&& \multicolumn{2}{c|}{\multirow{4}{\bigconstraintwidth}{\includegraphics[width=\bigconstraintwidth]{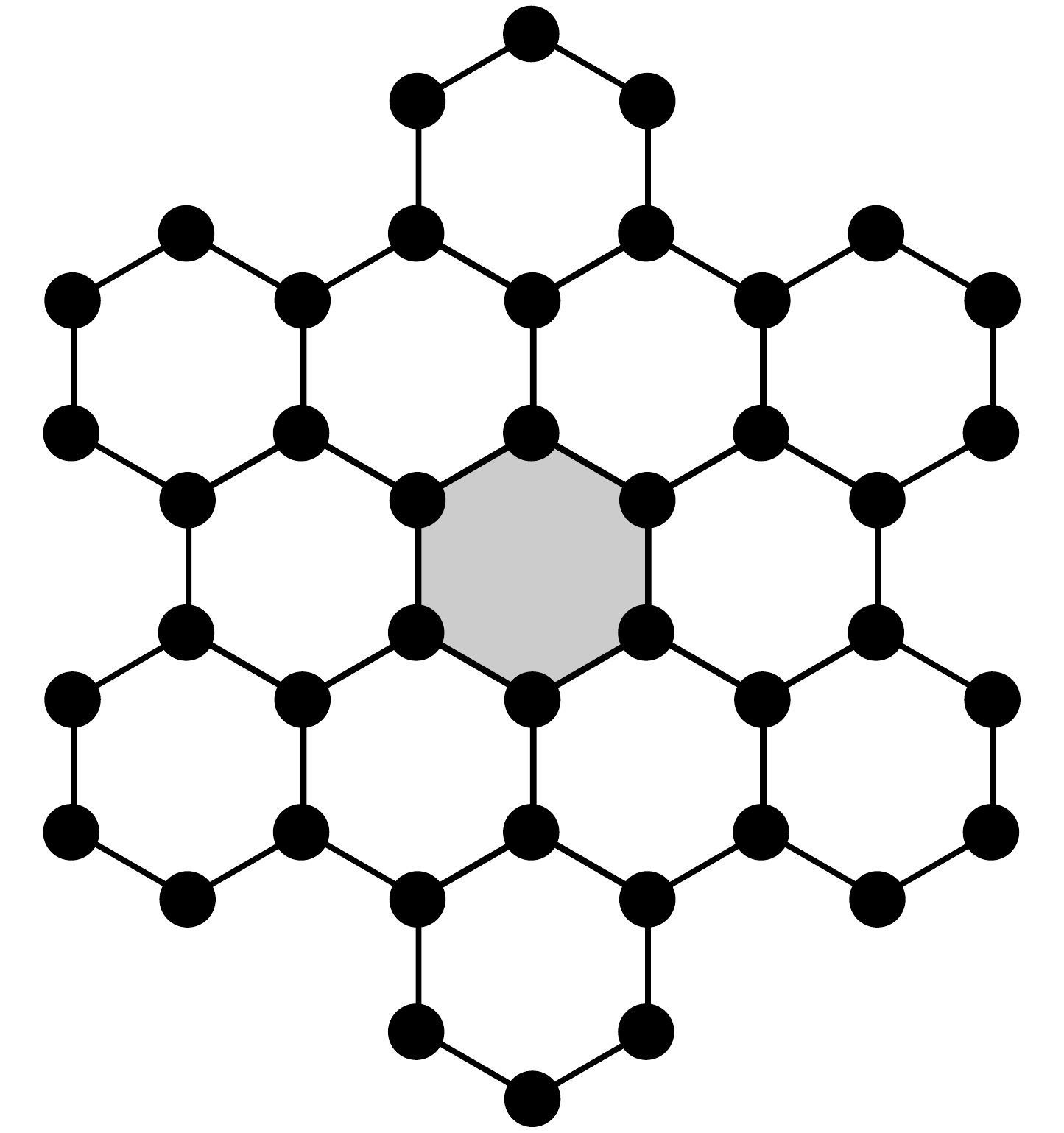}}}\\[-3ex]
$N$ & \includegraphics[width=\rulefigwidth]{hexgrid-rule-n.pdf} & &\\
$E_{1,3}$ & \includegraphics[width=\rulefigwidth]{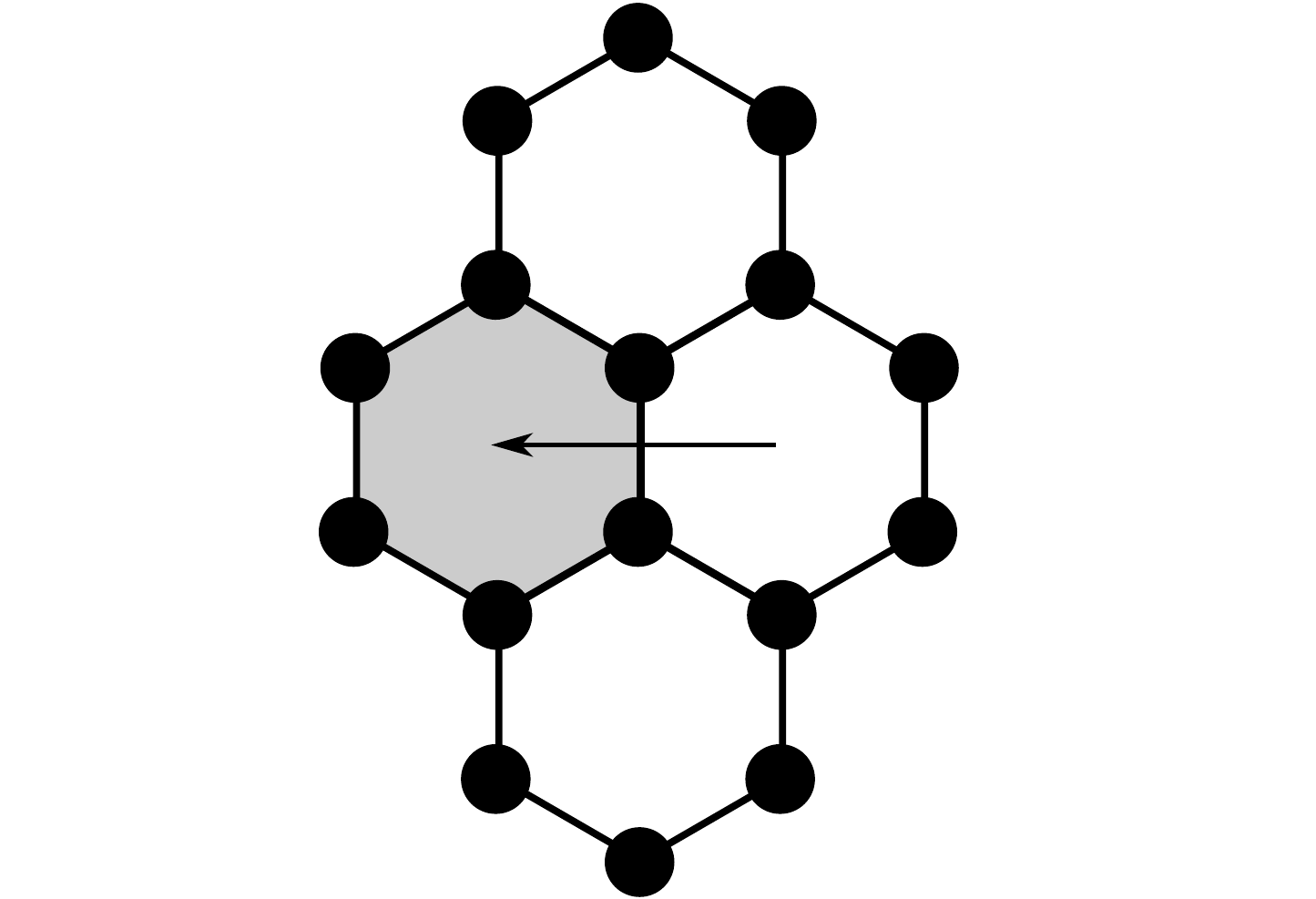} &&\\
$Y_{\operatorname{size}}$ & \includegraphics[width=\rulefigwidth]{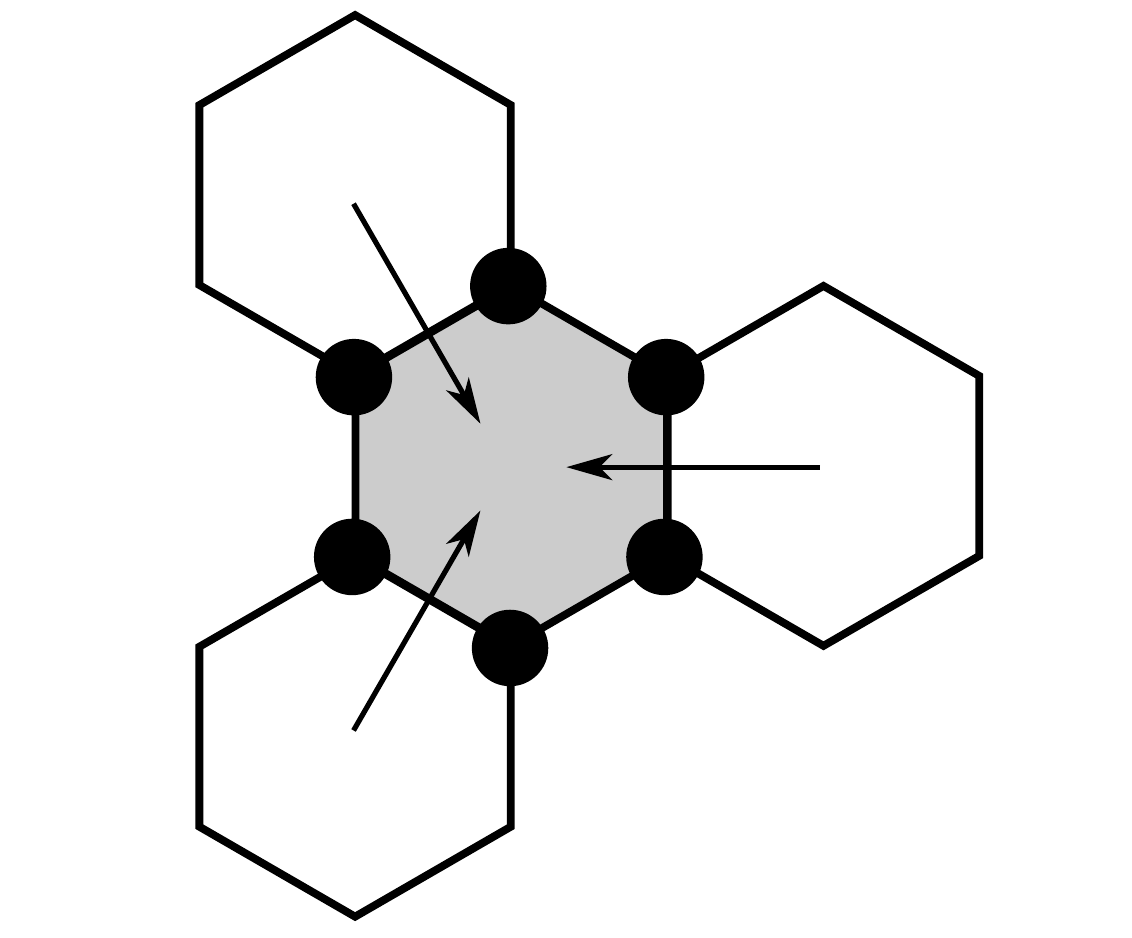} && \\
\hline
\end{tabular}
}

\caption{\label{tab:hexagonalrules}Various Rules and Constraint Configurations in the Hexagonal Grid.}
\end{table}

\begin{table}[p]\small

\centering
\def\rulefigwidth{0.85in}
\def\constraintwidth{0.85in}
\def\bigconstraintwidth{1.7in}
\def\medconstraintwidth{1.3in}
\mbox{
\begin{tabular}[h]{|m{0.1in}m{\rulefigwidth}|m{\constraintwidth}m{\constraintwidth}|}
\hline
\multicolumn{2}{|c|}{\textbf{Rules} }
&
\multicolumn{2}{c|}{\textbf{Constraint Configurations}}\\
\hline&&&\\[-3ex]
\hline&&&\\[-2ex]
$V_1$ &  \includegraphics[width=\rulefigwidth]{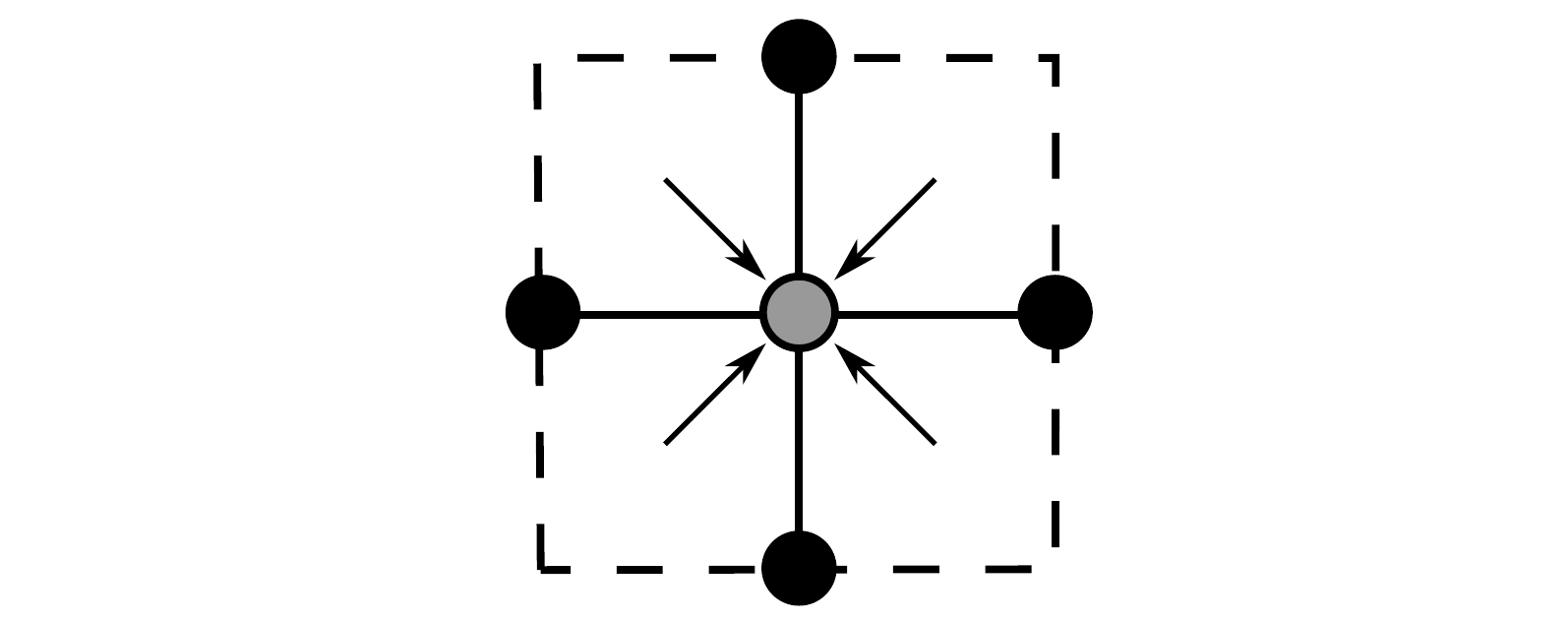} 
& 
\includegraphics[width=\constraintwidth]{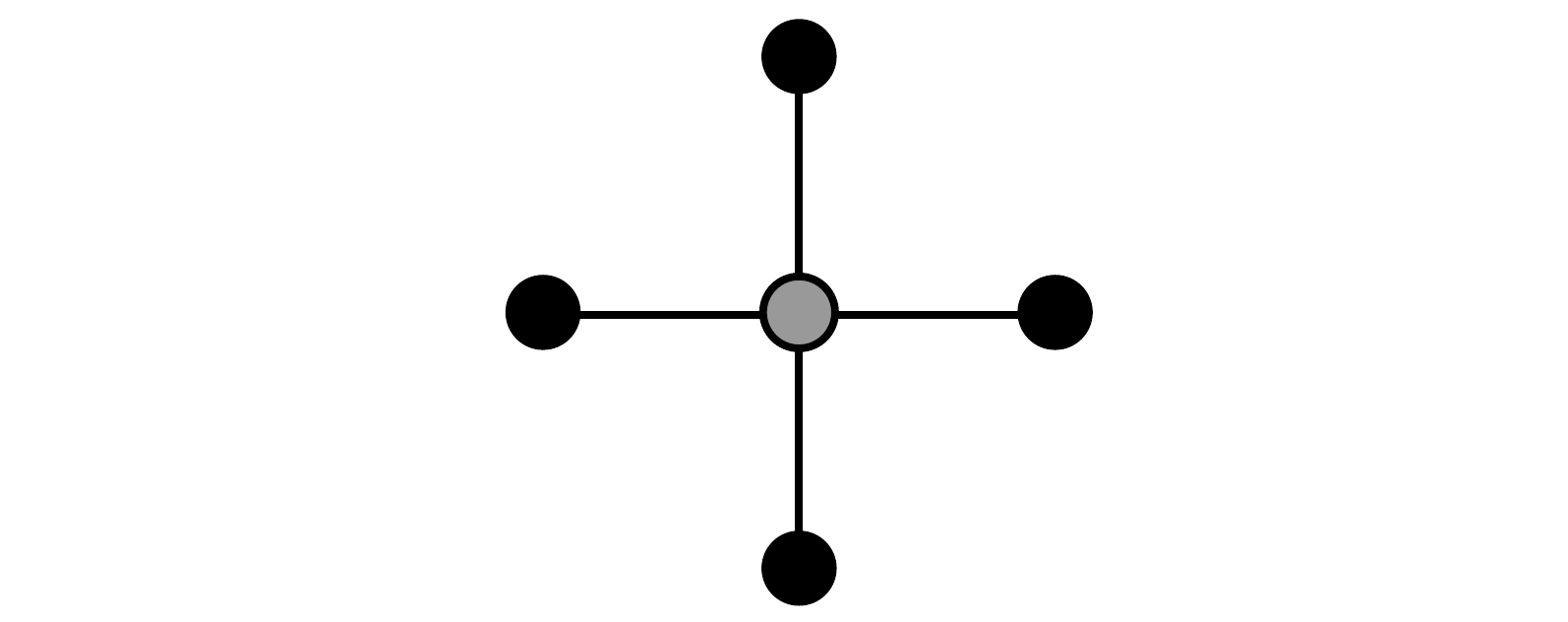} 
& 
\includegraphics[width=\constraintwidth]{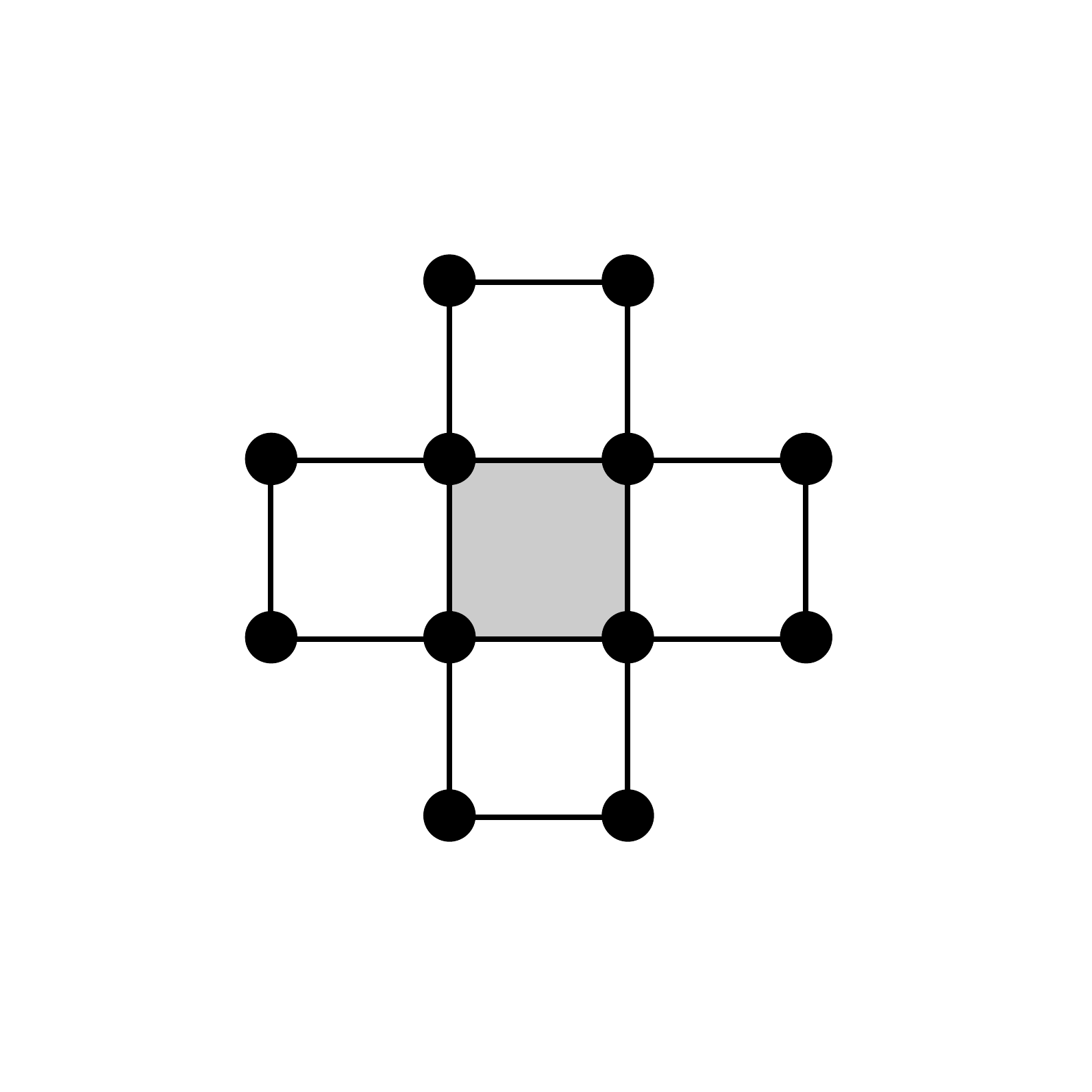} \\
\hline&&&\\[-2ex]
$N$ &  \includegraphics[width=\rulefigwidth]{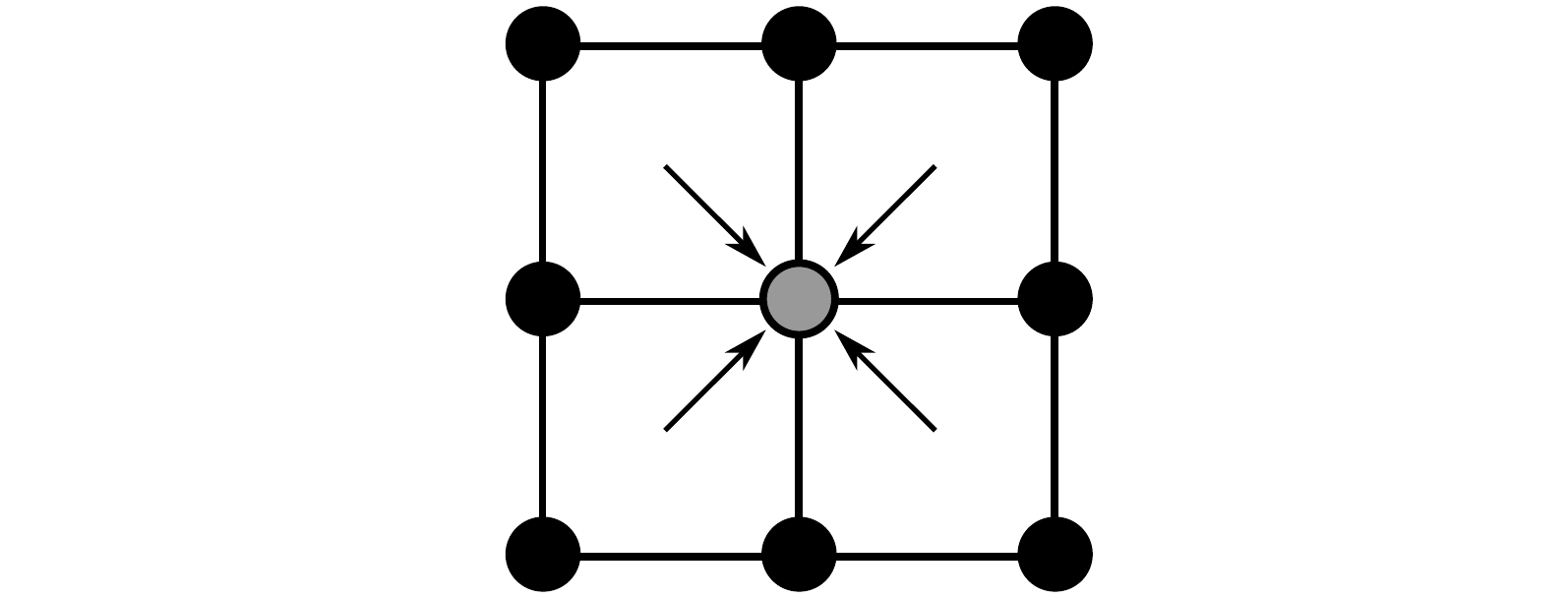} 
& 
\includegraphics[width=\constraintwidth]{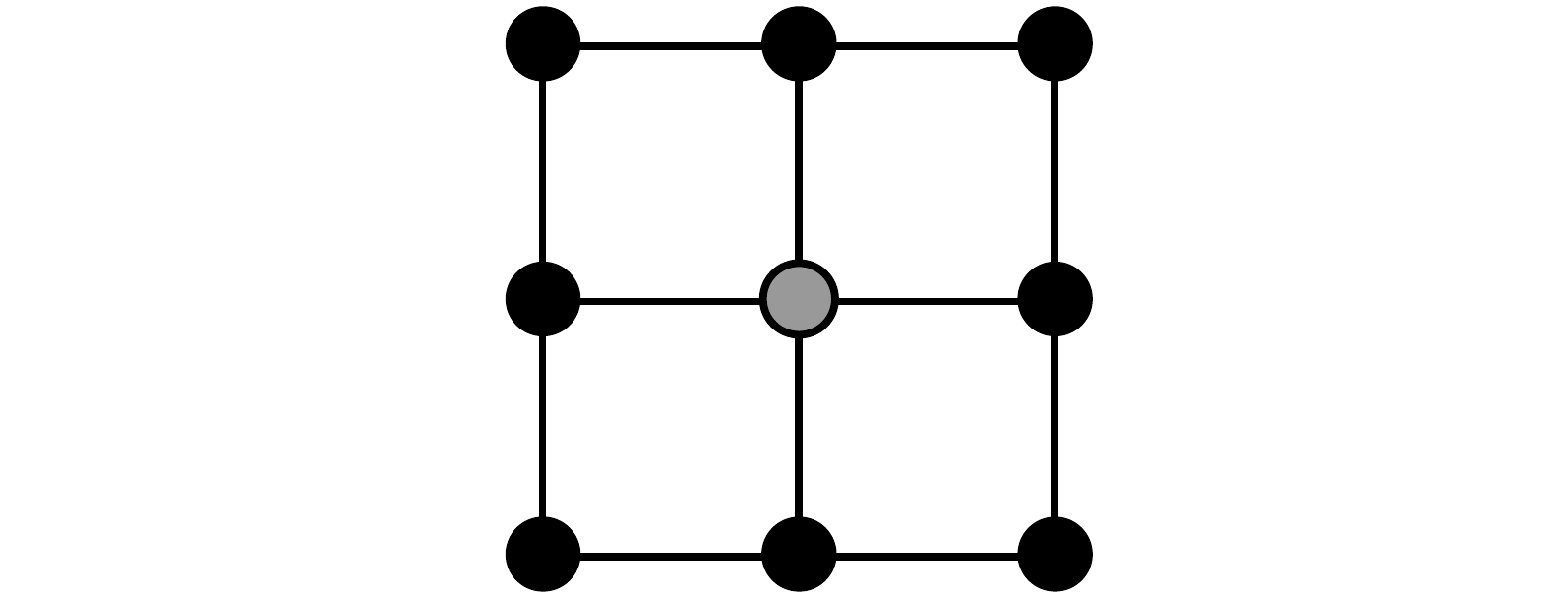} 
& 
\includegraphics[width=\constraintwidth]{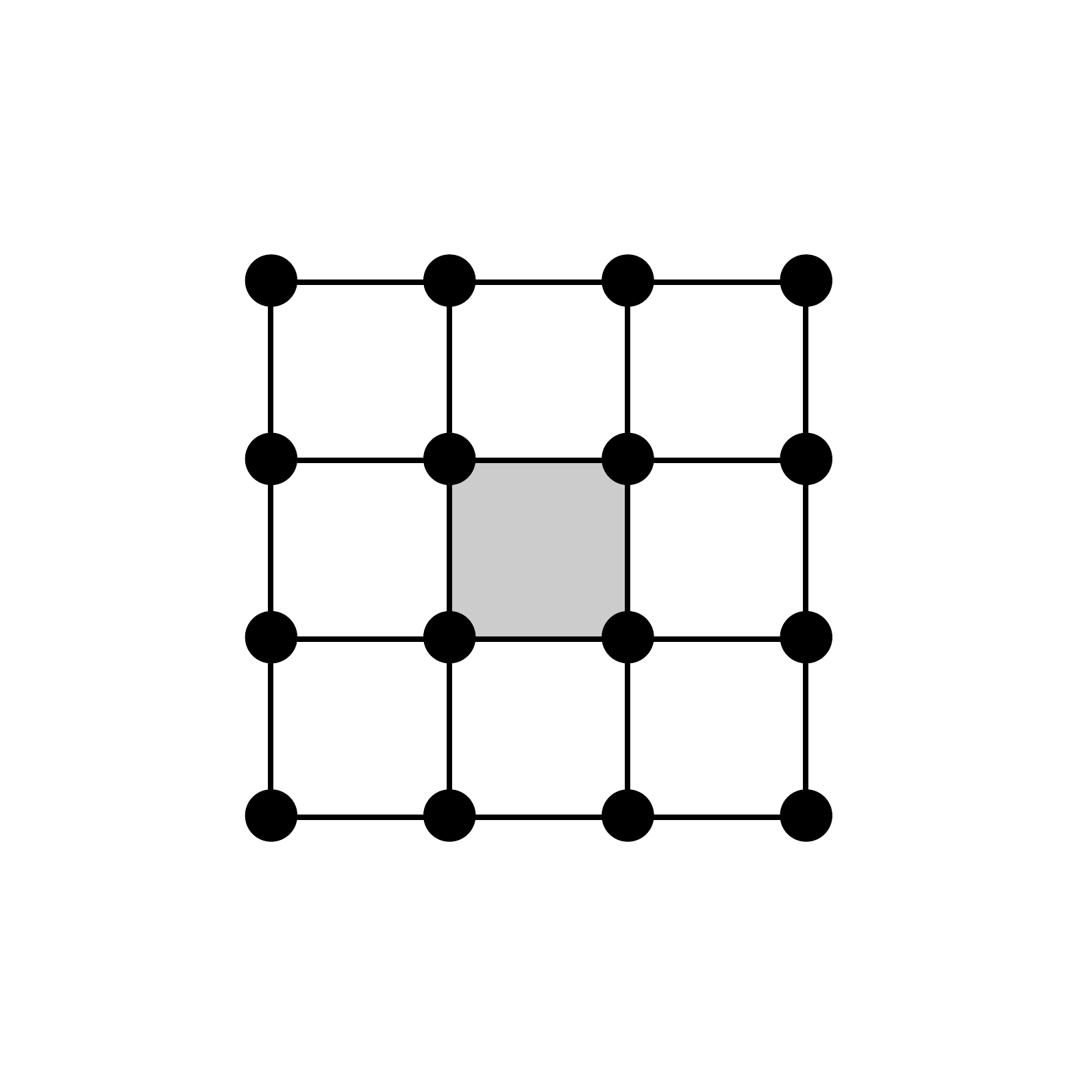} \\
\hline&&&\\[-2ex]
$V_2$ & 
\includegraphics[width=\rulefigwidth]{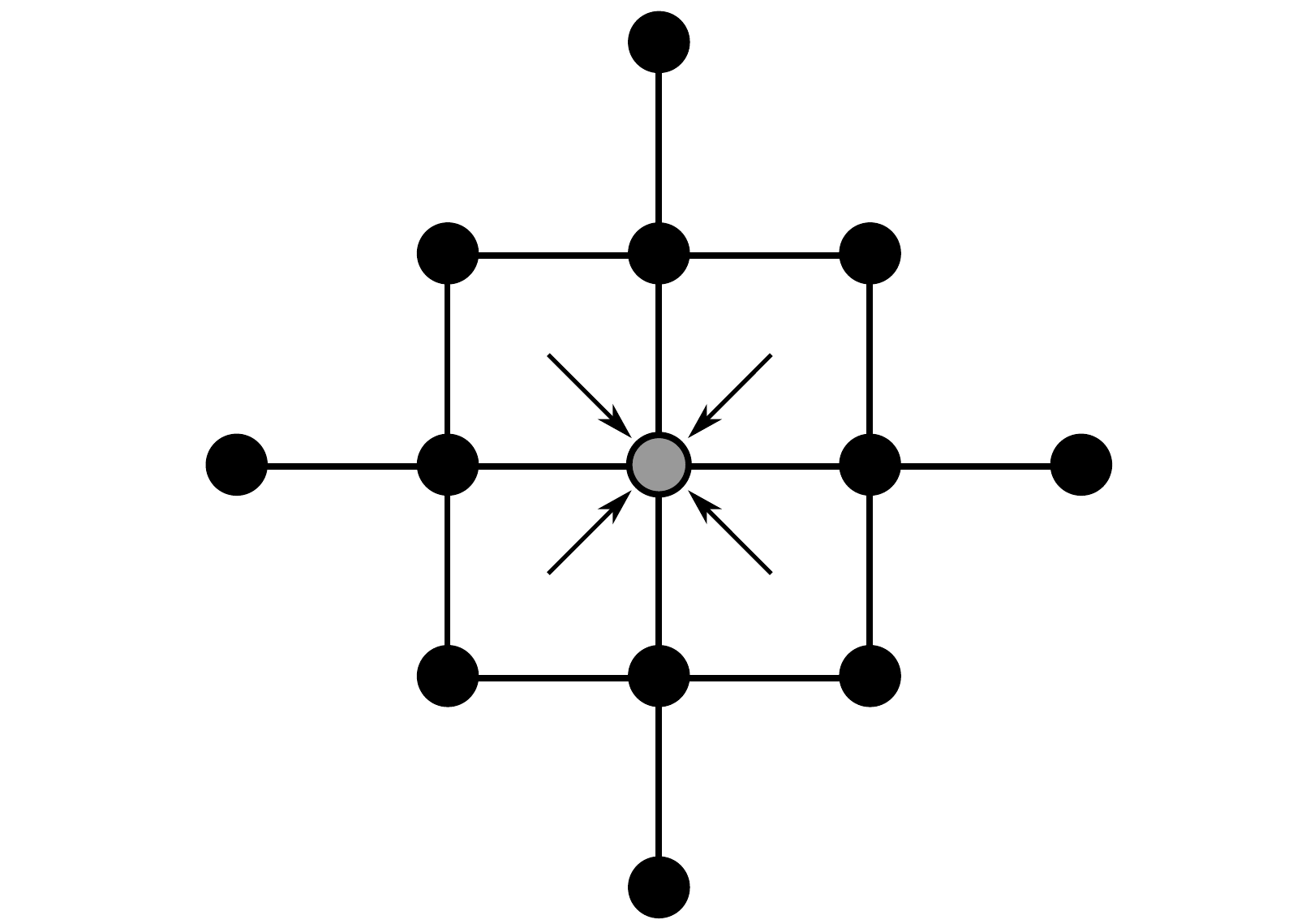} 
& 
\includegraphics[width=\constraintwidth]{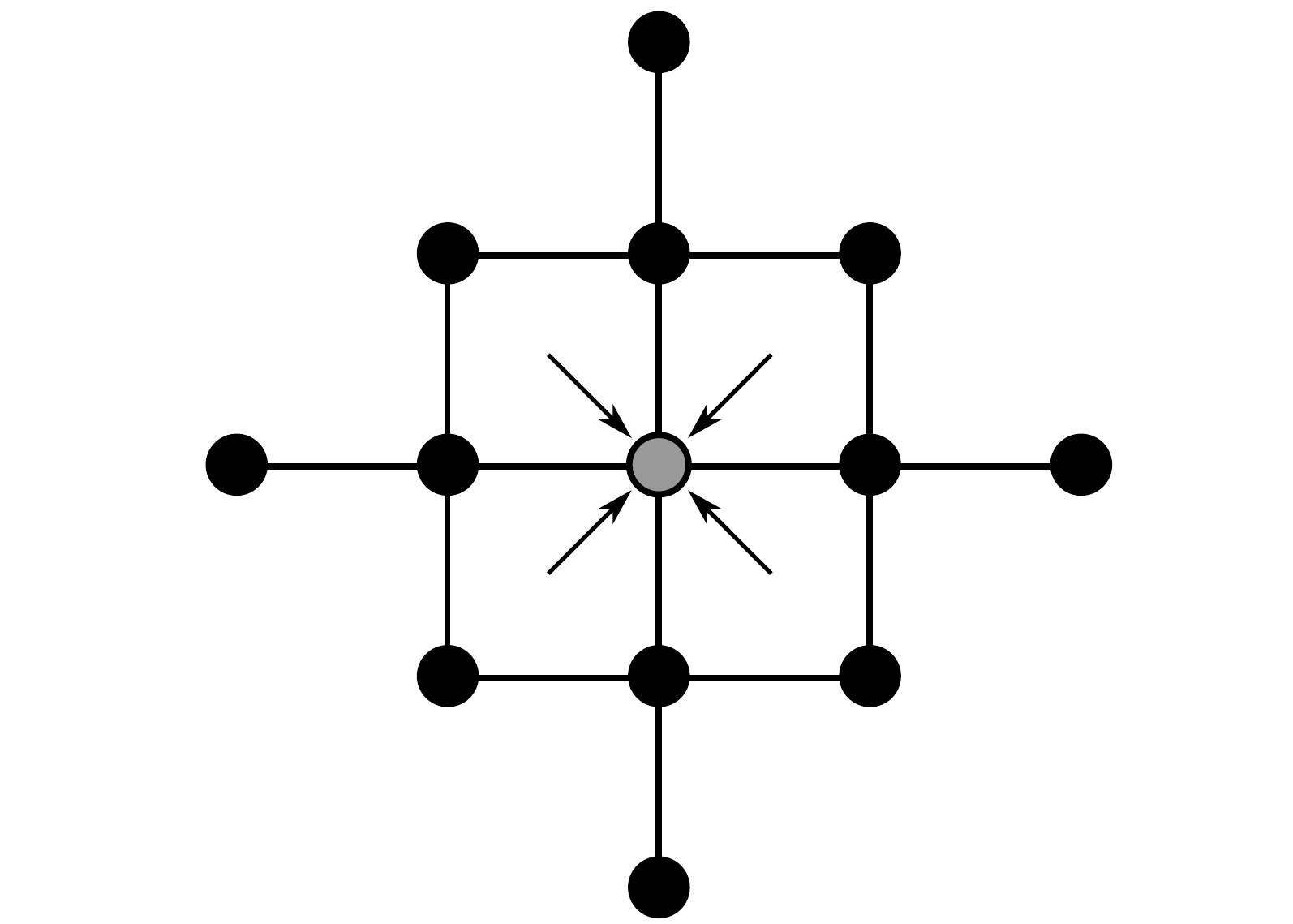} 
& 
\includegraphics[width=\constraintwidth]{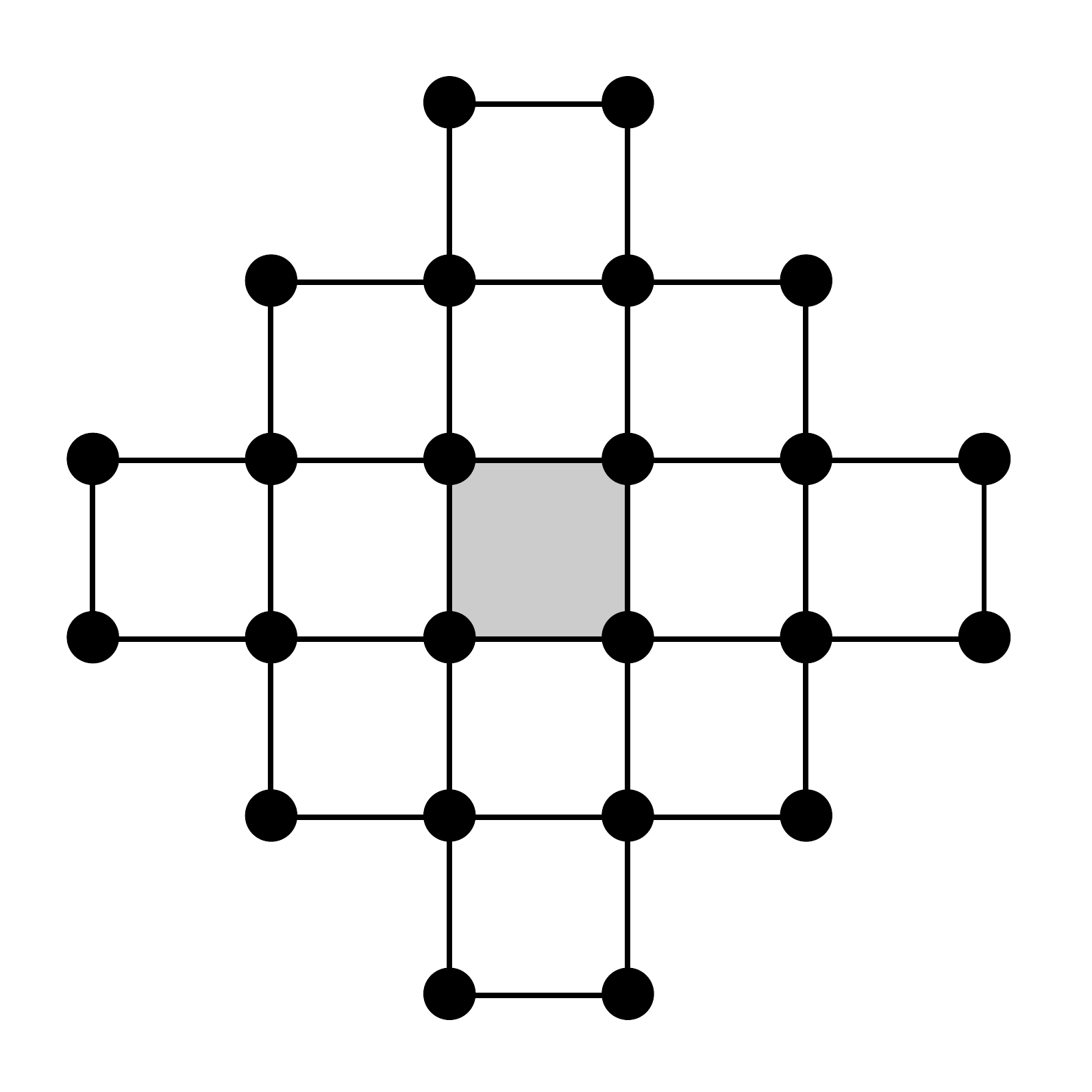} \\
\hline&&&\\[-2ex]
$C_1$ & 
\includegraphics[width=\rulefigwidth]{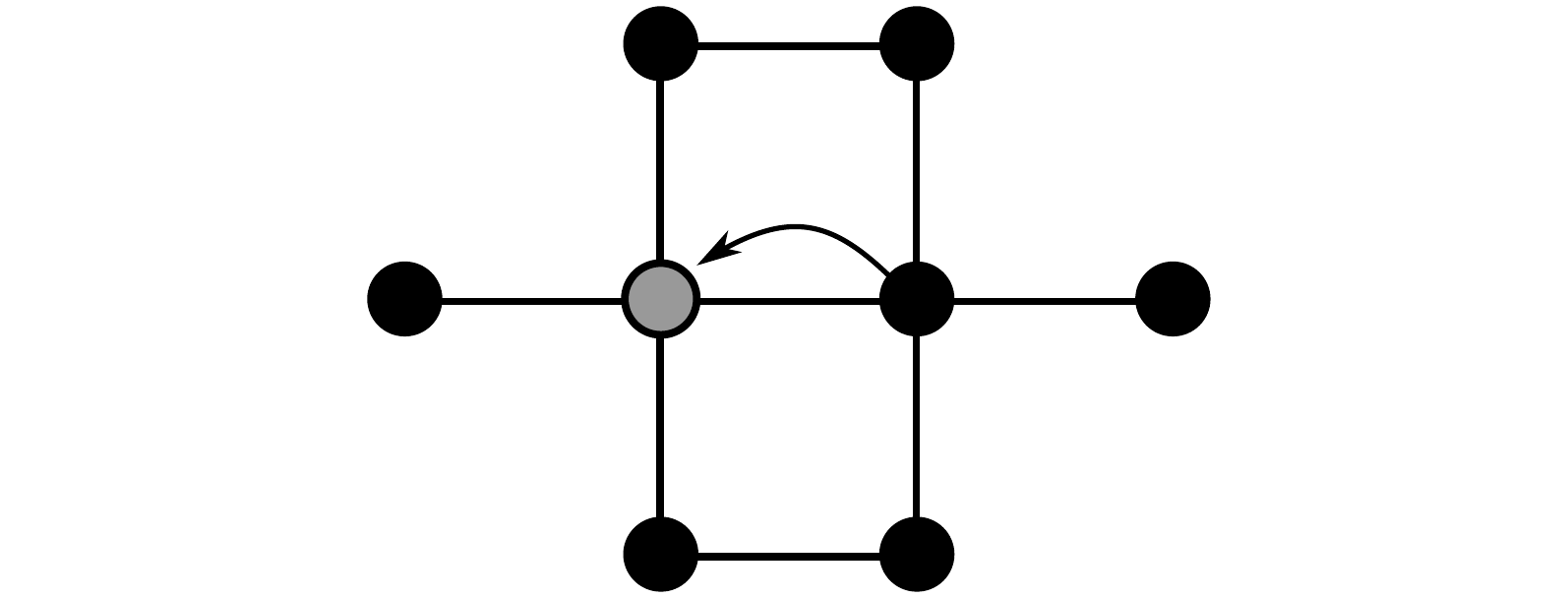} 
& \includegraphics[width=\constraintwidth]{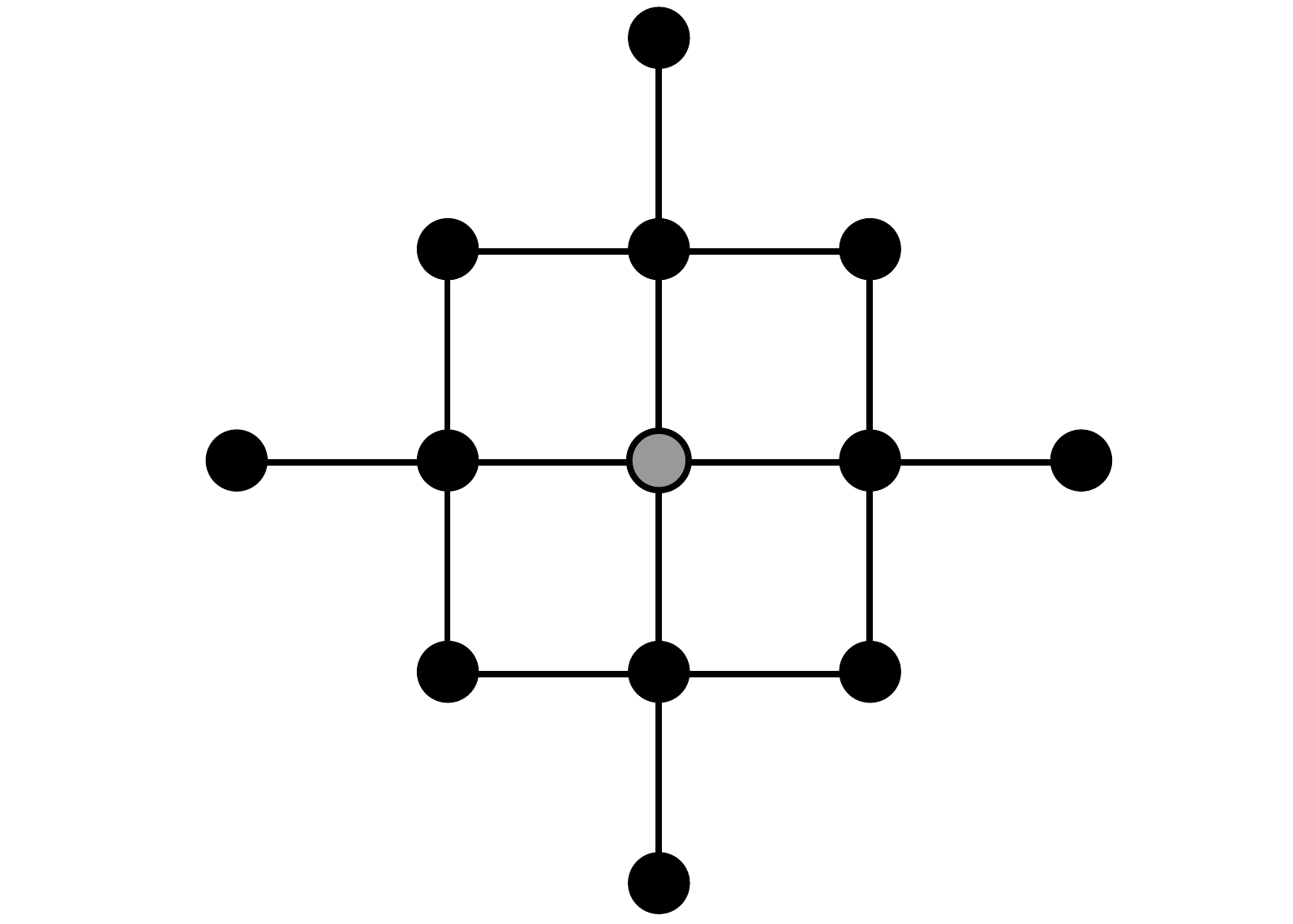} &
\\
\hline&&&\\[-2ex]
$C_{1}^+$ & 
\includegraphics[width=\rulefigwidth]{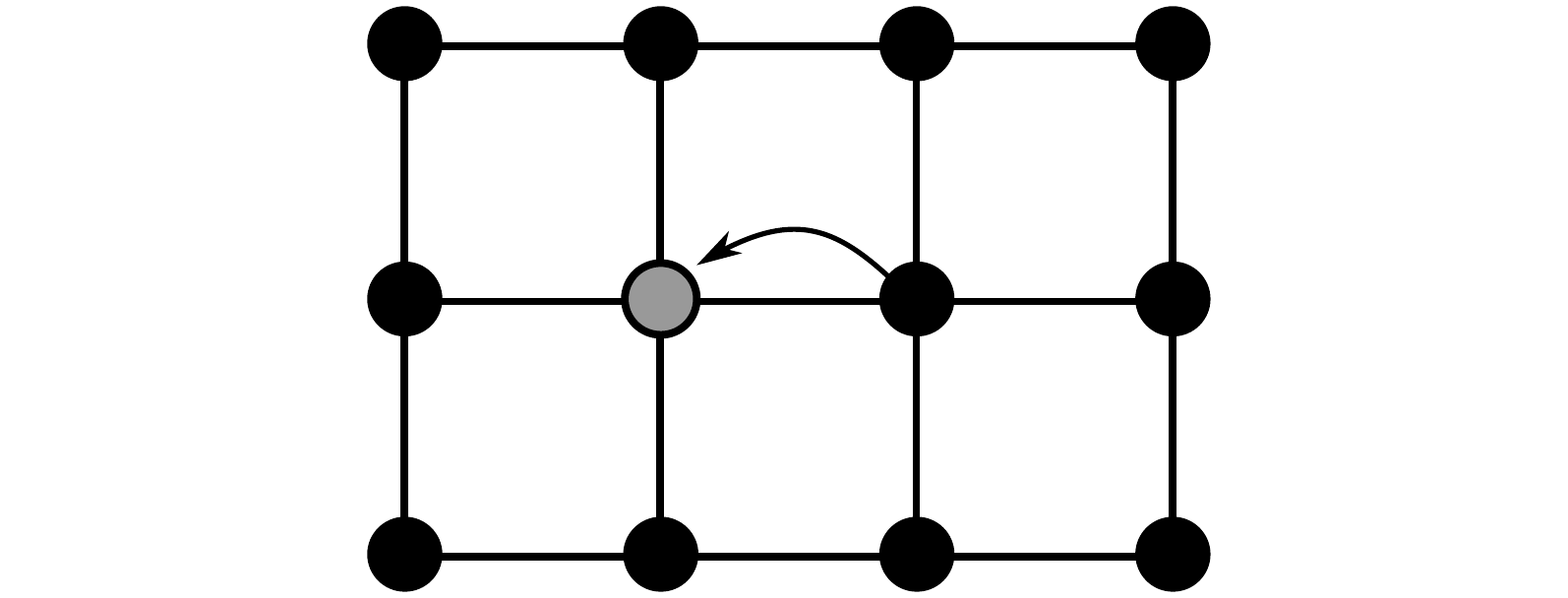} 
& \includegraphics[width=\constraintwidth]{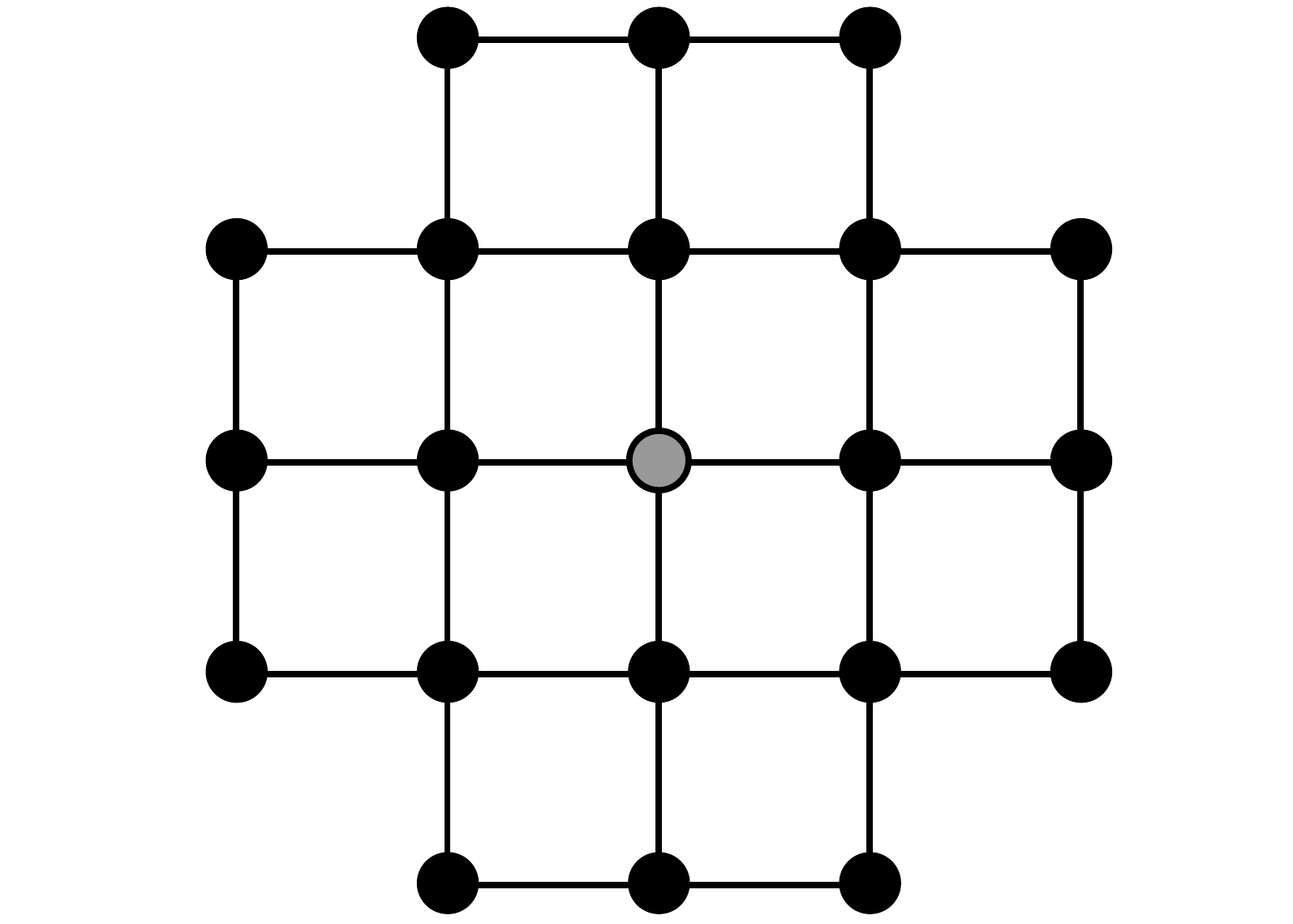} &
\\
\hline&&&\\[-2ex]
\multirow{2}{0.1in}{$C_{2}$} & 
\includegraphics[width=\rulefigwidth]{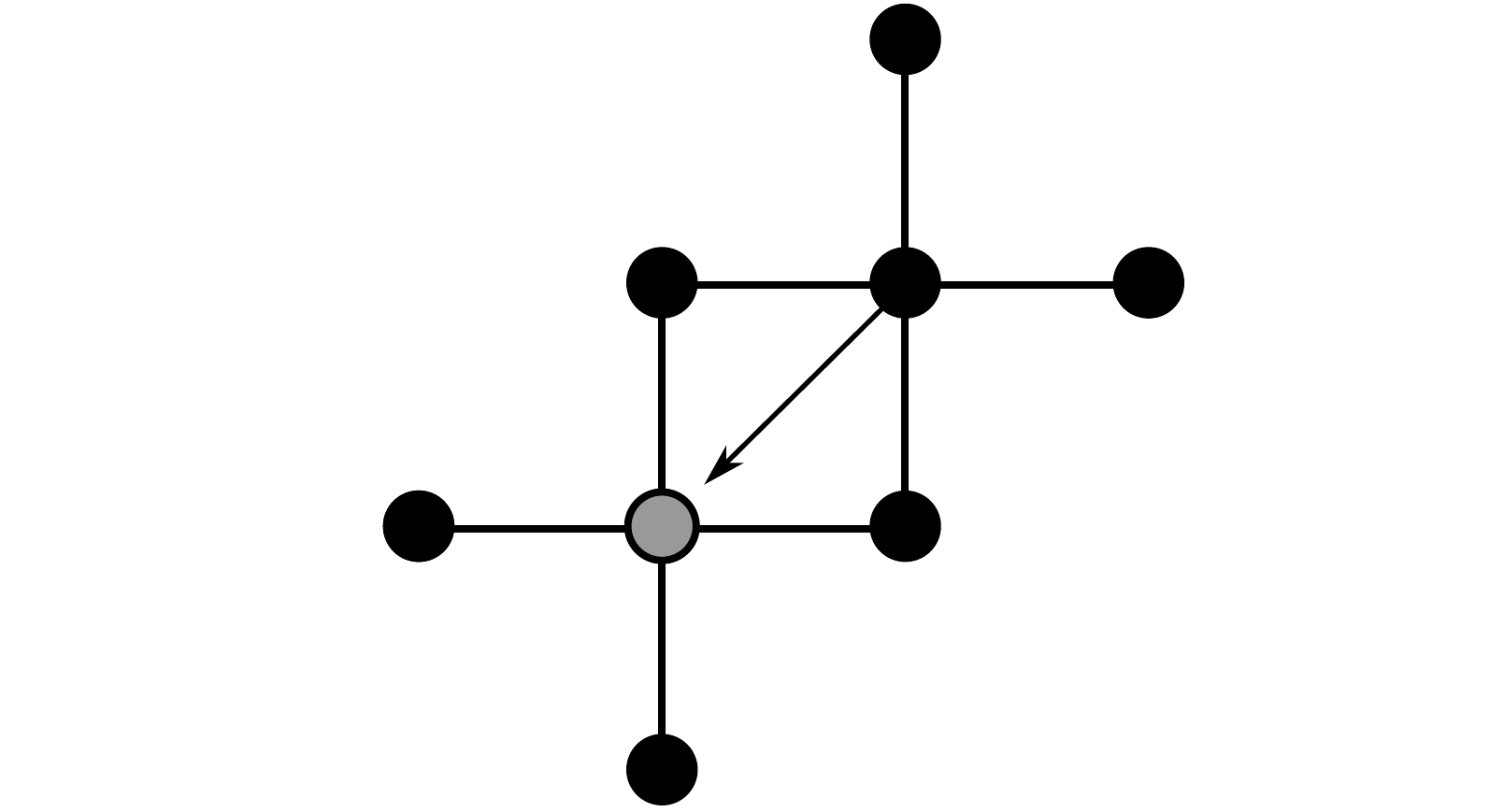} 
& 
\multirow{2}{\constraintwidth}{\includegraphics[width=\constraintwidth]{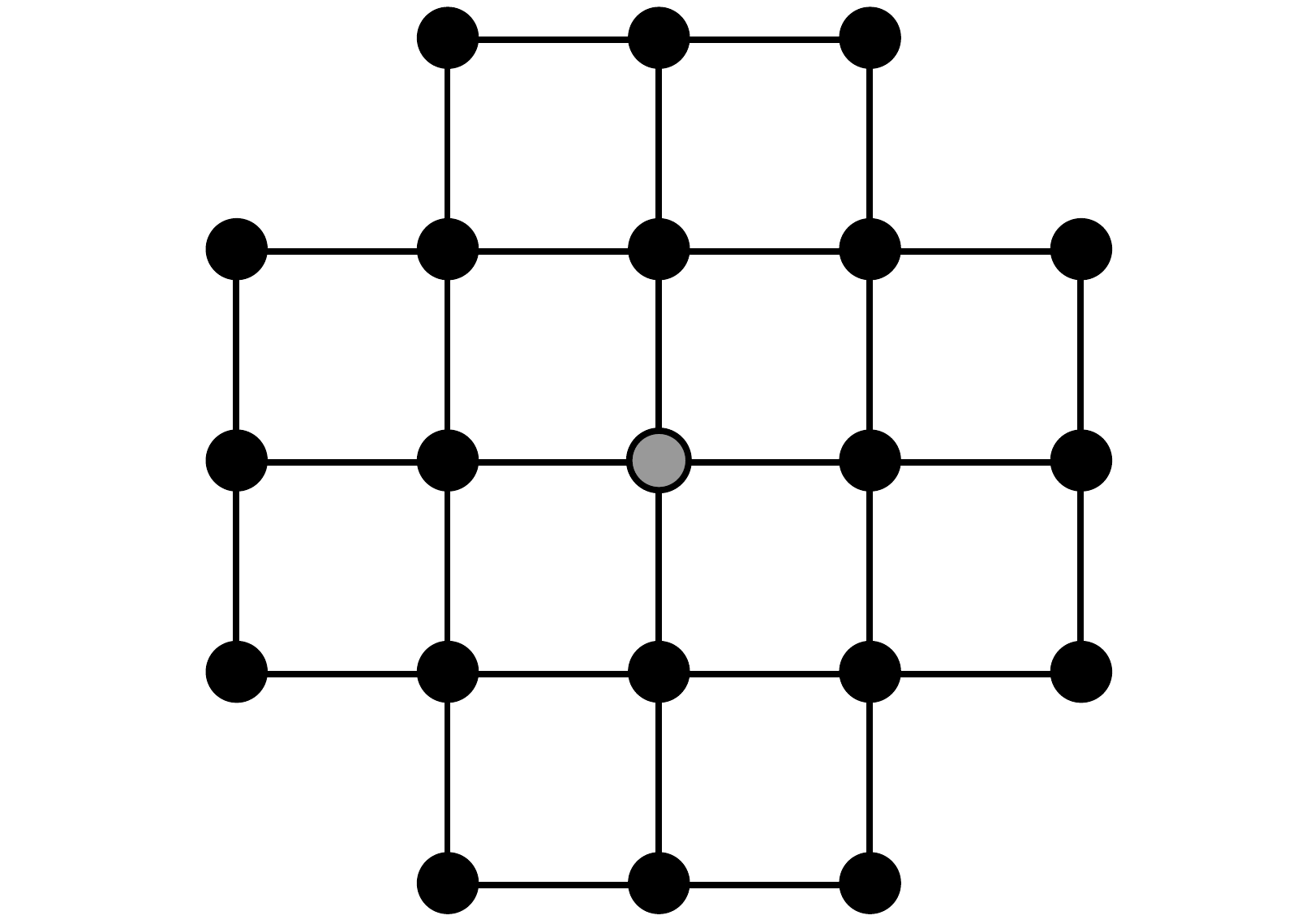}} &
\\
& 
\includegraphics[width=\rulefigwidth]{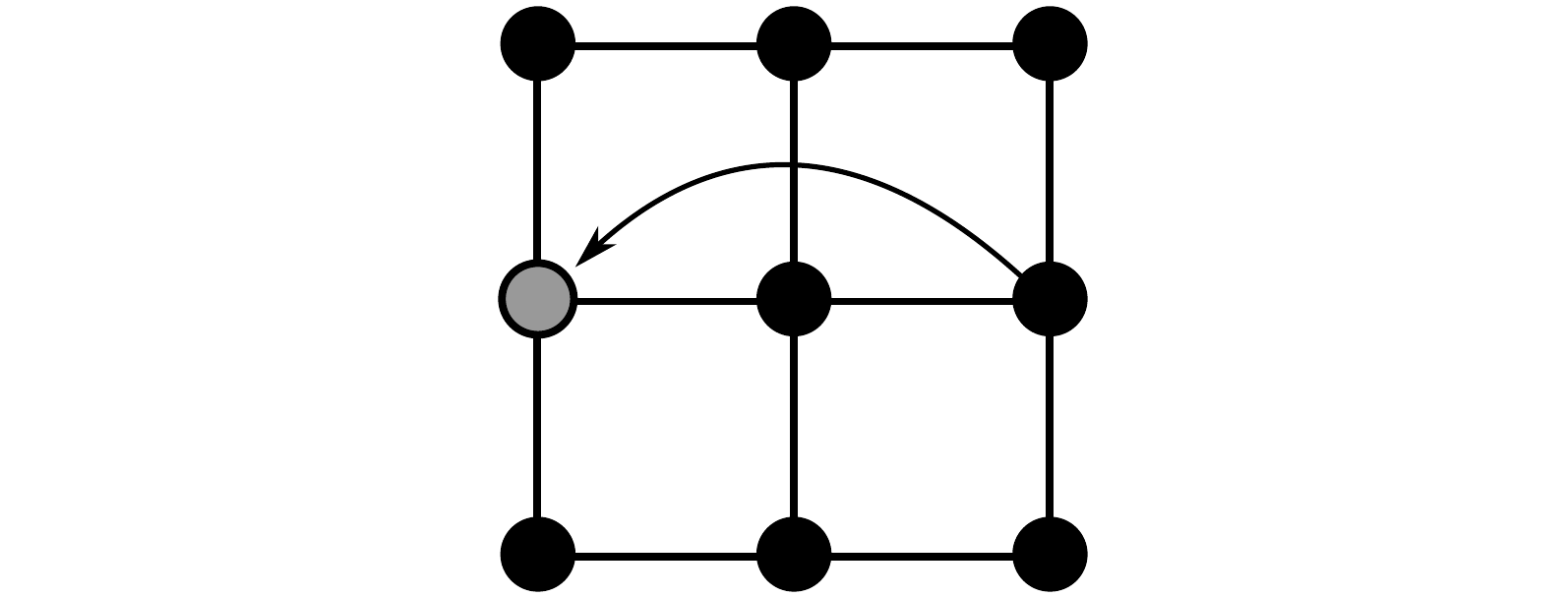} &&
\\
\hline&&&\\[-2ex]
& &
\multirow{3}{\constraintwidth}{\includegraphics[width=\constraintwidth]{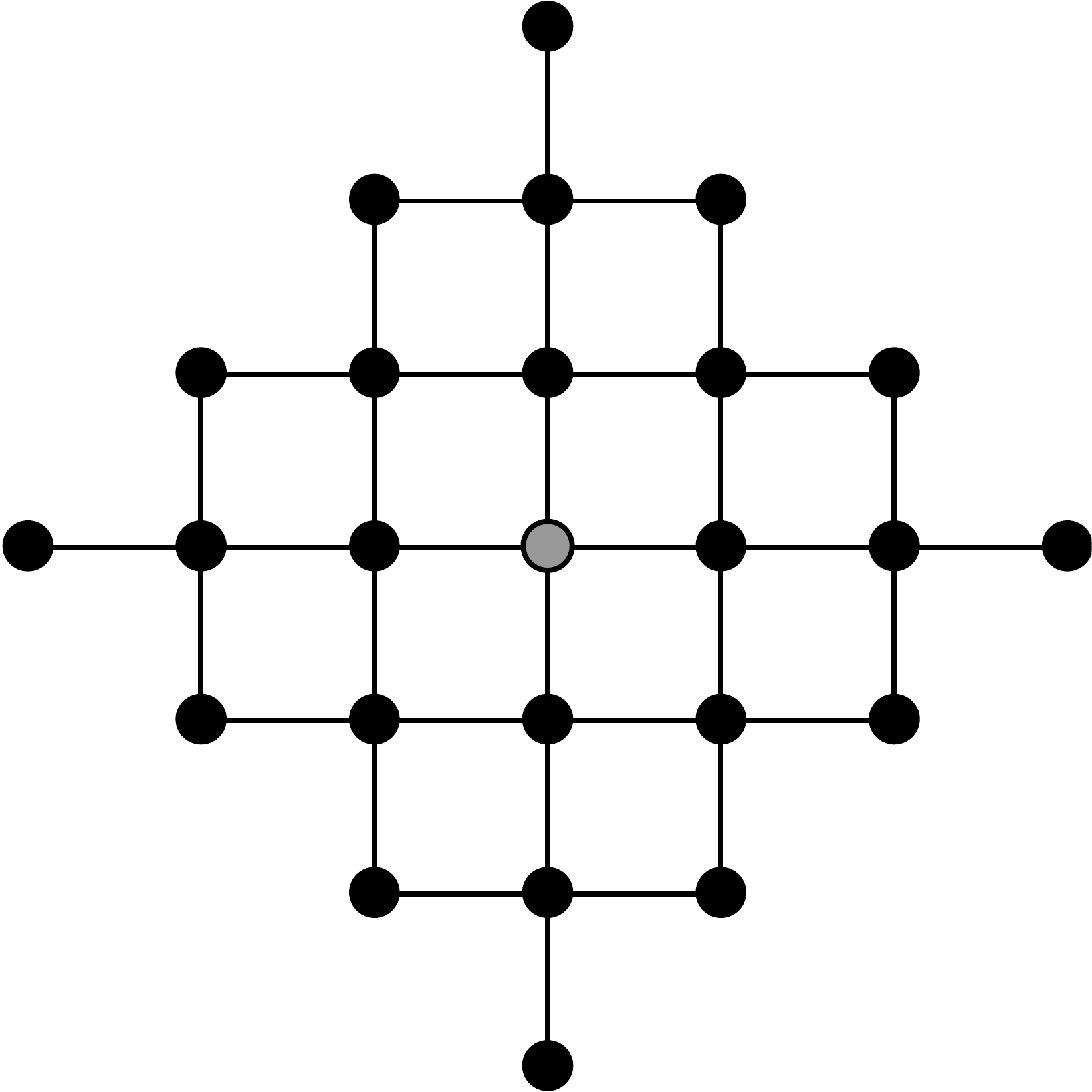}} &
\\[-2ex]
\multirow{2}{0.1in}{$C_{2}^+$} & 
\includegraphics[width=\rulefigwidth]{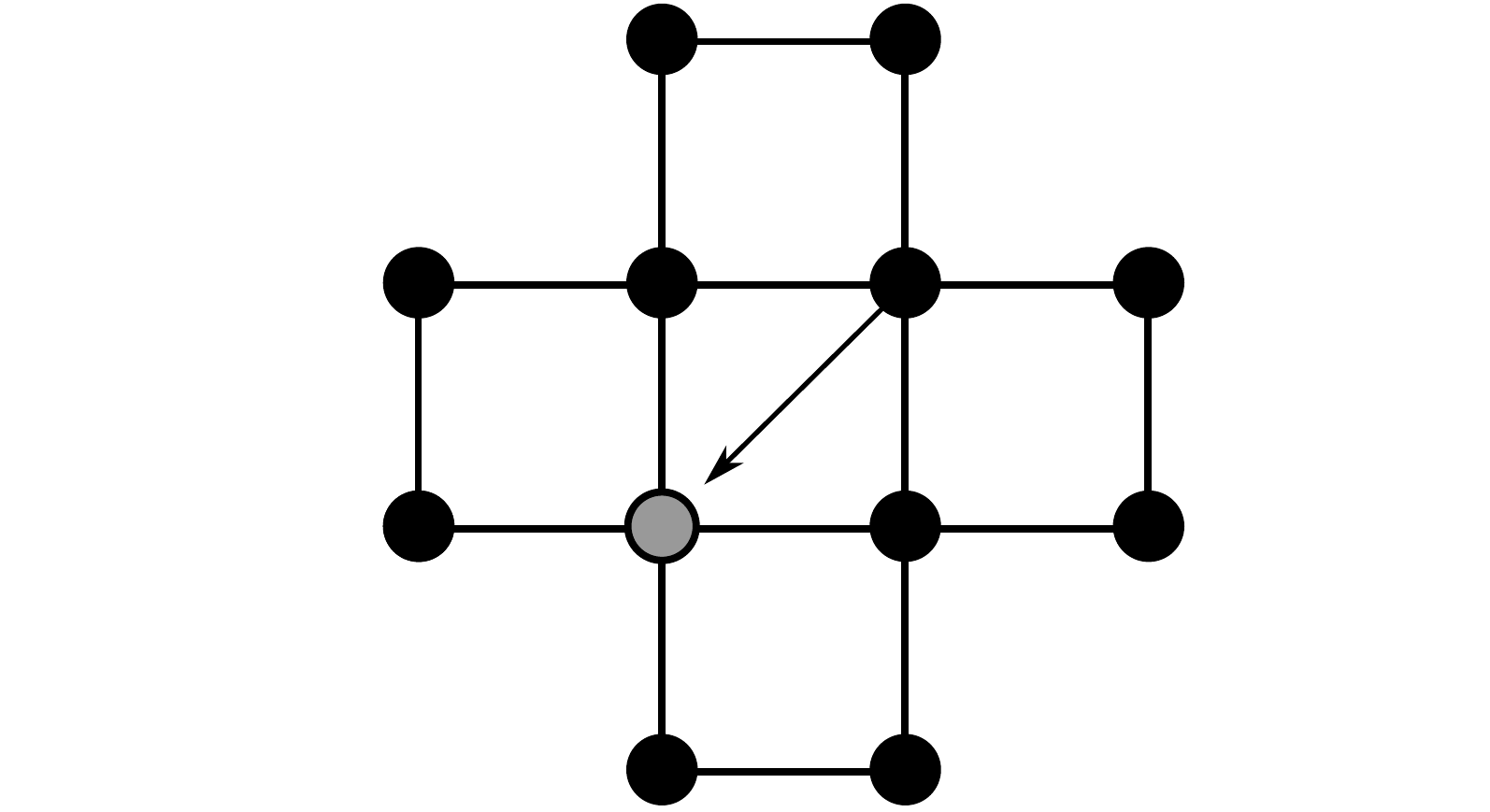} &&
\\
& 
\includegraphics[width=\rulefigwidth]{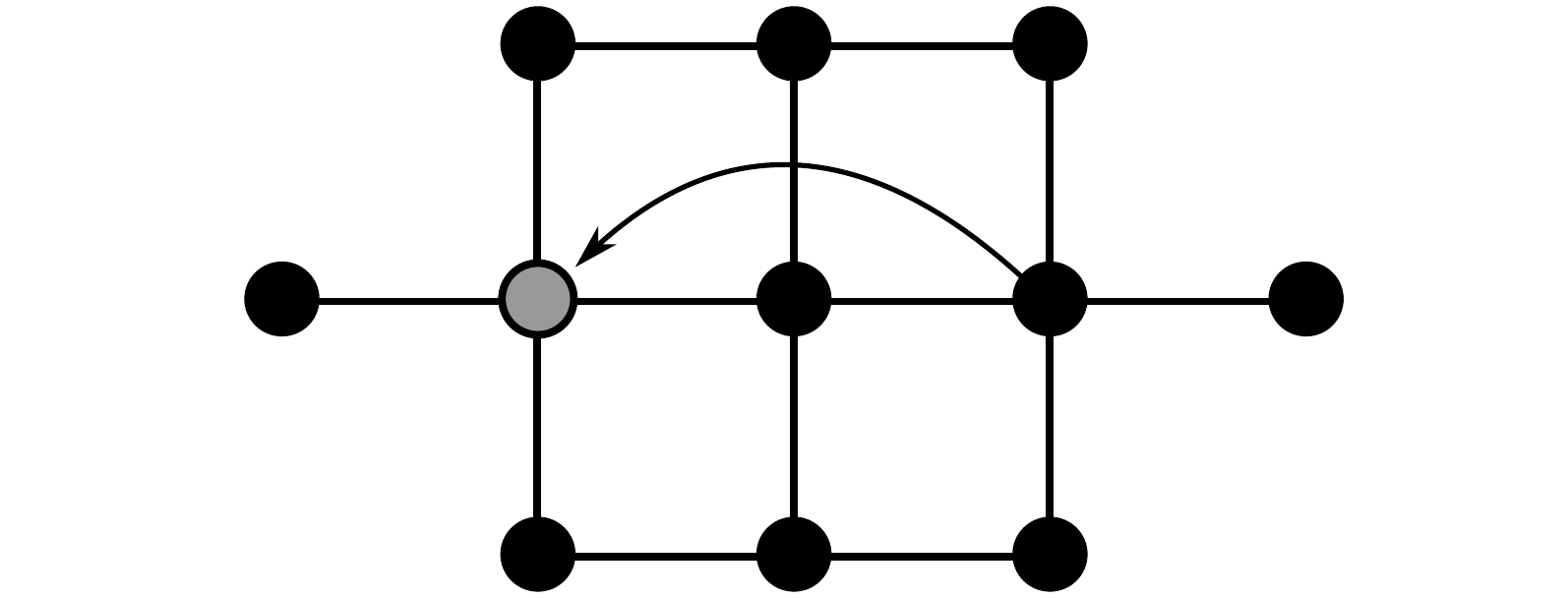} &&
\\[3ex]
\hline
\end{tabular}
\quad
\begin{tabular}[h]{|m{0.1in}m{\rulefigwidth}|m{\constraintwidth}m{\constraintwidth}|}
\hline
\multicolumn{2}{|c|}{\textbf{Rules} }
&
\multicolumn{2}{c|}{\textbf{Constraint Configurations}}\\
\hline&&&\\[-3ex]
\hline&&&\\[-2ex]
$C_{1}$ & 
\includegraphics[width=\rulefigwidth]{squaregrid-rule-c1.pdf} 
& 
\multicolumn{2}{c|}{\multirow{5}{\bigconstraintwidth}{\includegraphics[width=\bigconstraintwidth]{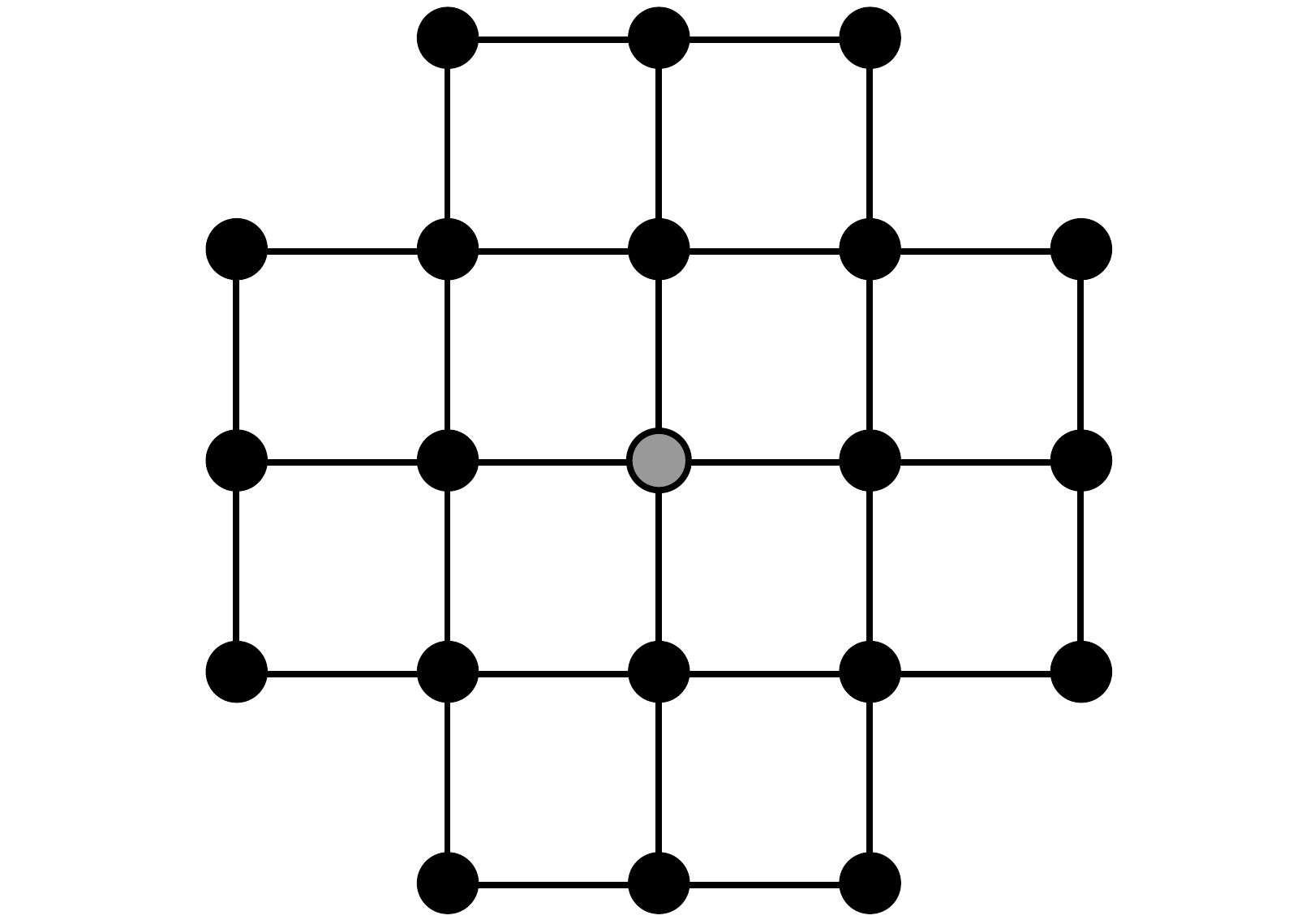}}}
\\
\multirow{2}{0.1in}{$C_{2}$} & 
\includegraphics[width=\rulefigwidth]{squaregrid-rule-c2a.pdf} 
& 
&
\\
& 
\includegraphics[width=\rulefigwidth]{squaregrid-rule-c2b.pdf} &&
\\
\multirow{2}{0.1in}{$C_{3}$} & 
\includegraphics[width=\rulefigwidth]{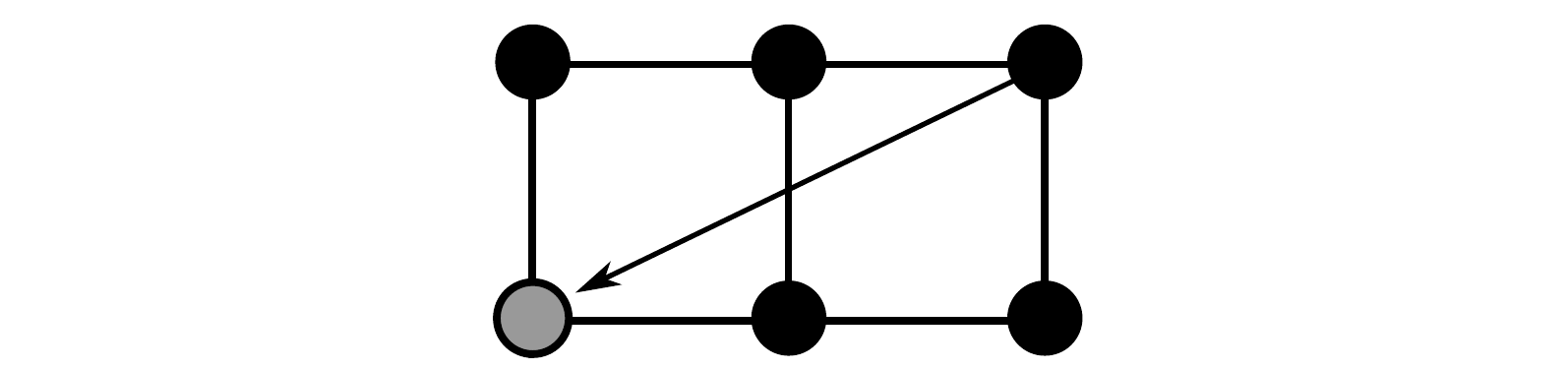} 
& 
&
\\
 & 
\includegraphics[width=\rulefigwidth]{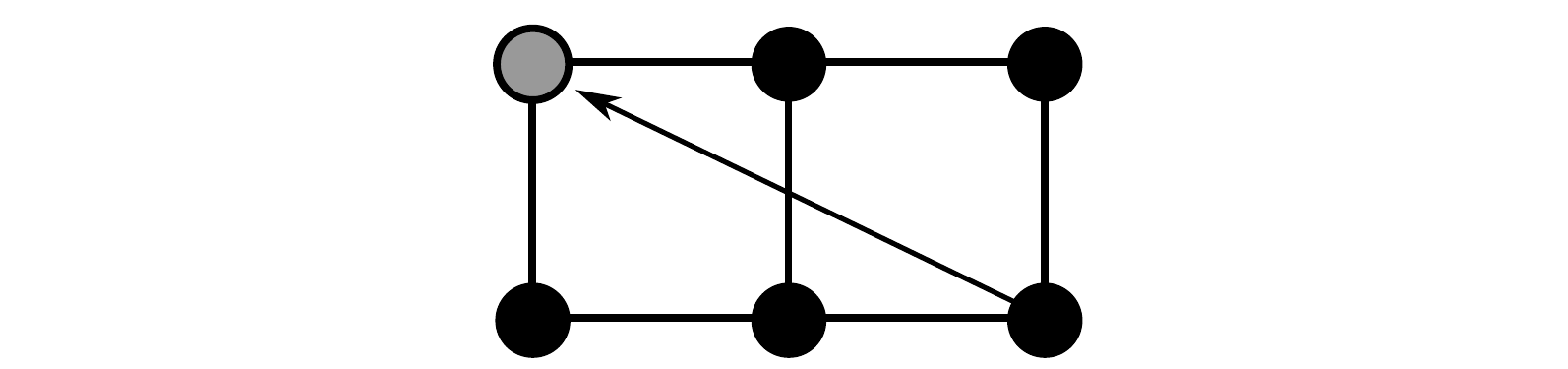} 
& 
&
\\
\hline&&&\\[-2ex]
$C_{1}^+$ & 
\includegraphics[width=\rulefigwidth]{squaregrid-rule-c1p.pdf} 
& 
\multicolumn{2}{c|}{\multirow{5}{\bigconstraintwidth}{\includegraphics[width=\bigconstraintwidth]{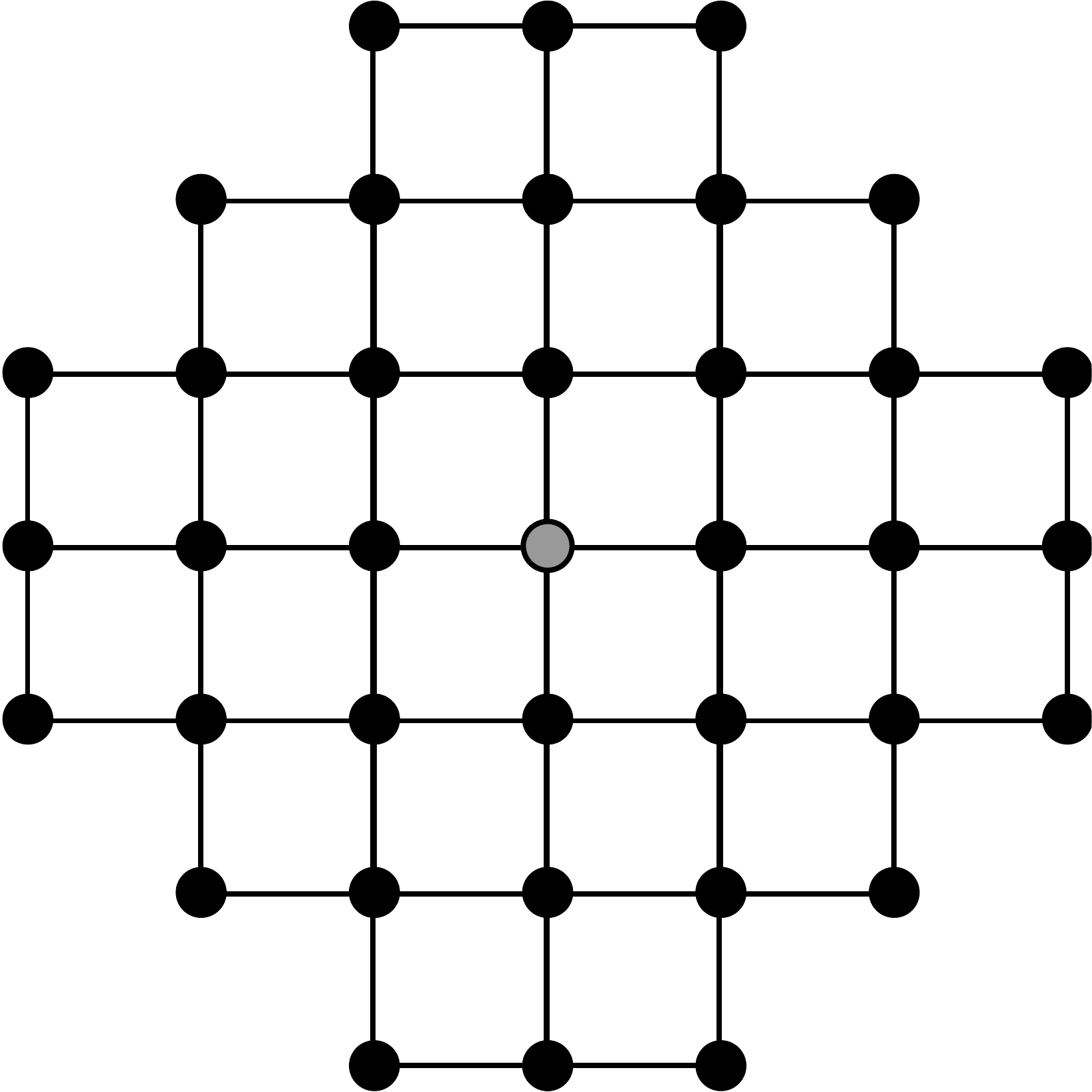}}}
\\
\multirow{2}{0.1in}{$C_{2}^+$} & 
\includegraphics[width=\rulefigwidth]{squaregrid-rule-c2ap.pdf} 
& 
&
\\
& 
\includegraphics[width=\rulefigwidth]{squaregrid-rule-c2bp.pdf} &&
\\
\multirow{2}{0.1in}{$C_{3}^+$} & 
\includegraphics[width=\rulefigwidth]{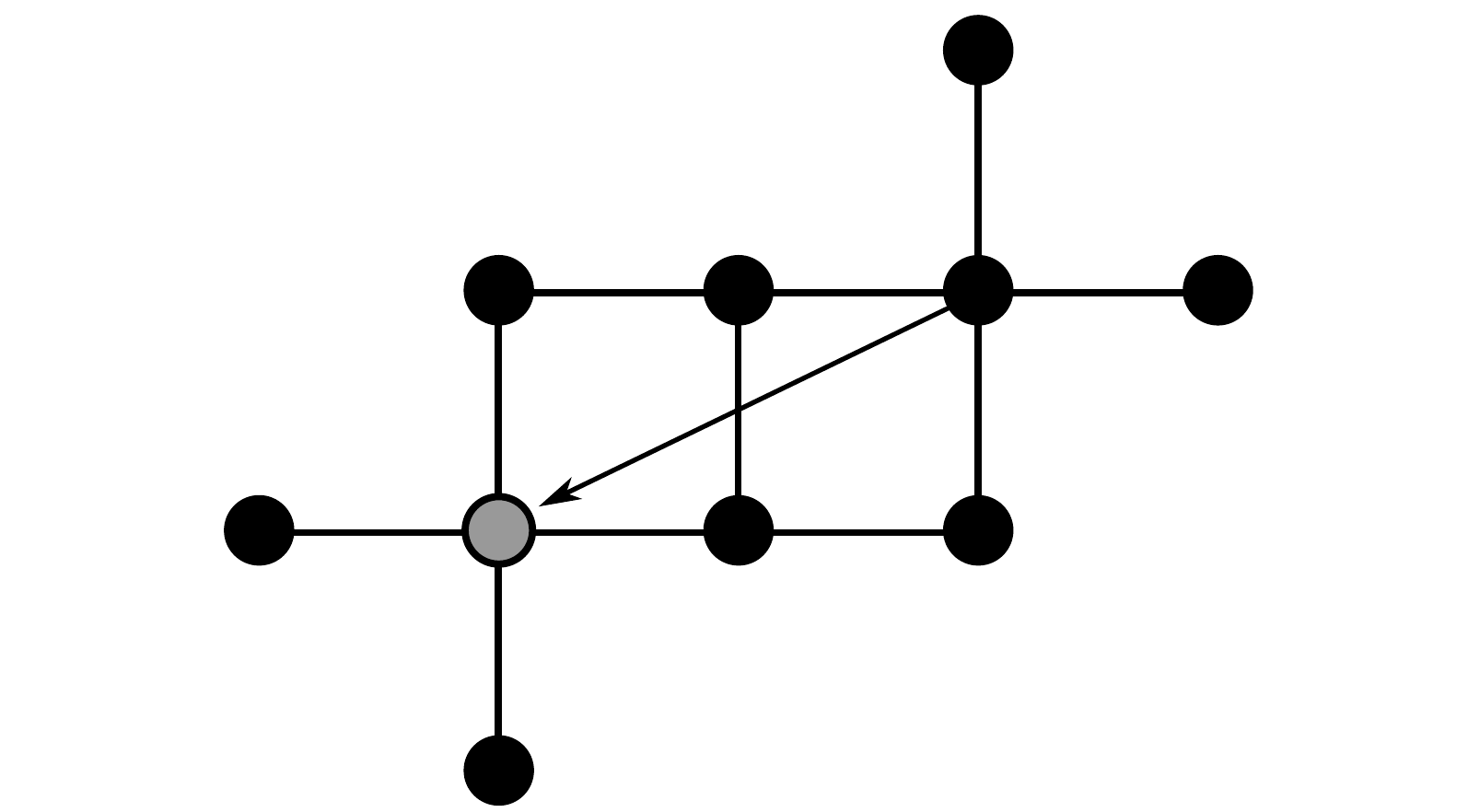} 
& 
&
\\
 & 
\includegraphics[width=\rulefigwidth]{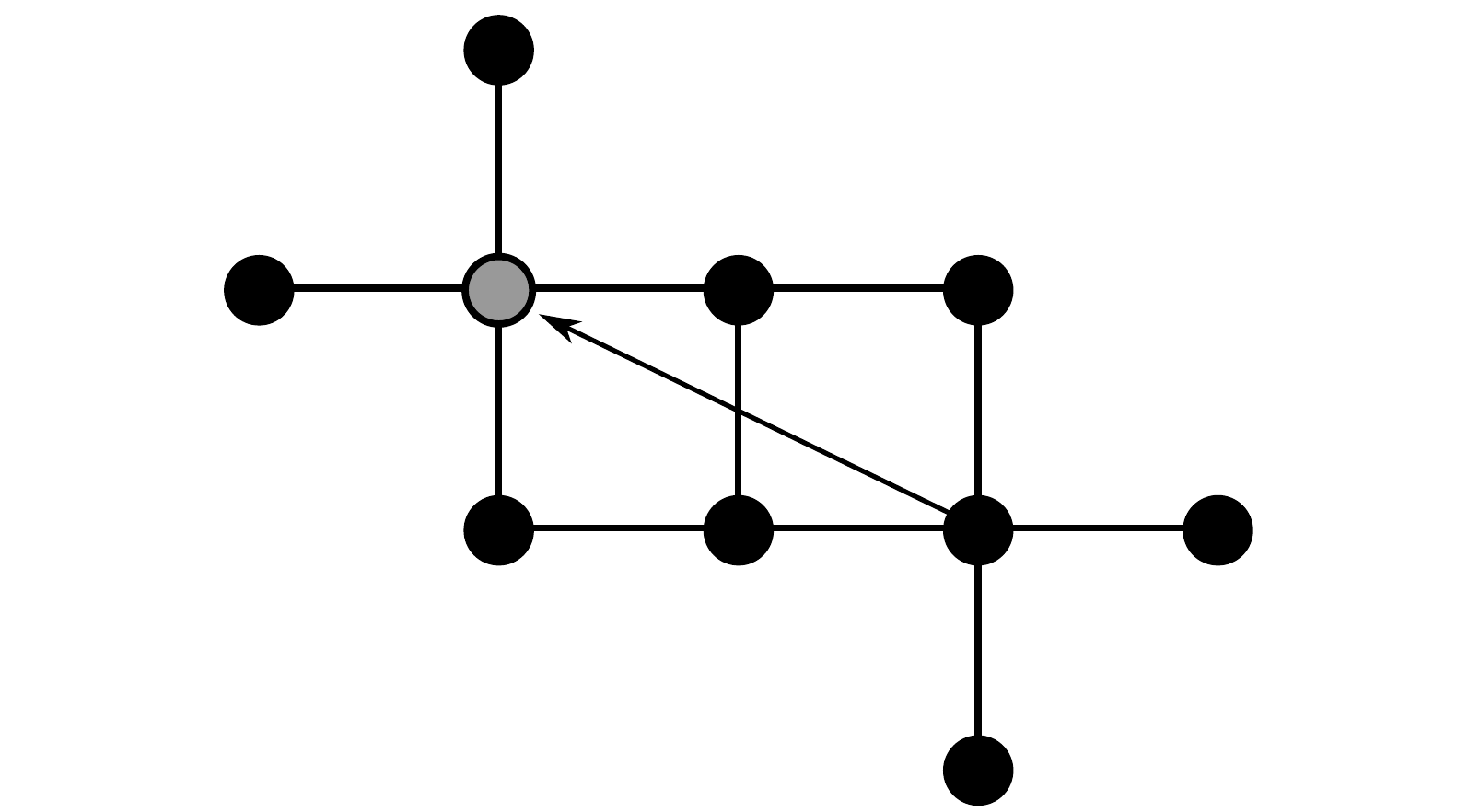} 
& 
&
\\
\hline&&&\\[-2ex]
&&\multirow{3}{\constraintwidth}{\includegraphics[width=\constraintwidth]{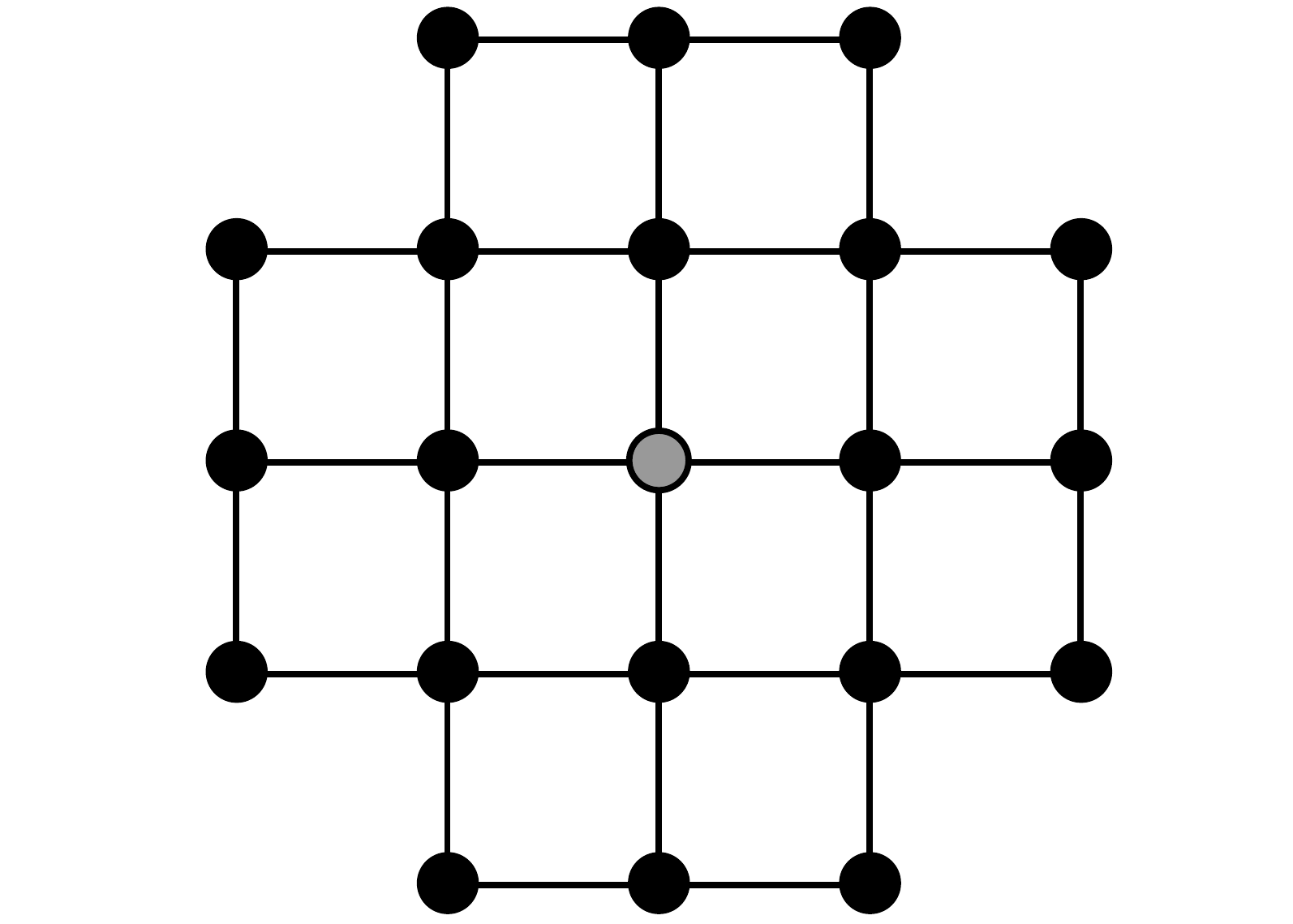}} &
\multirow{3}{\constraintwidth}{\includegraphics[width=\constraintwidth]{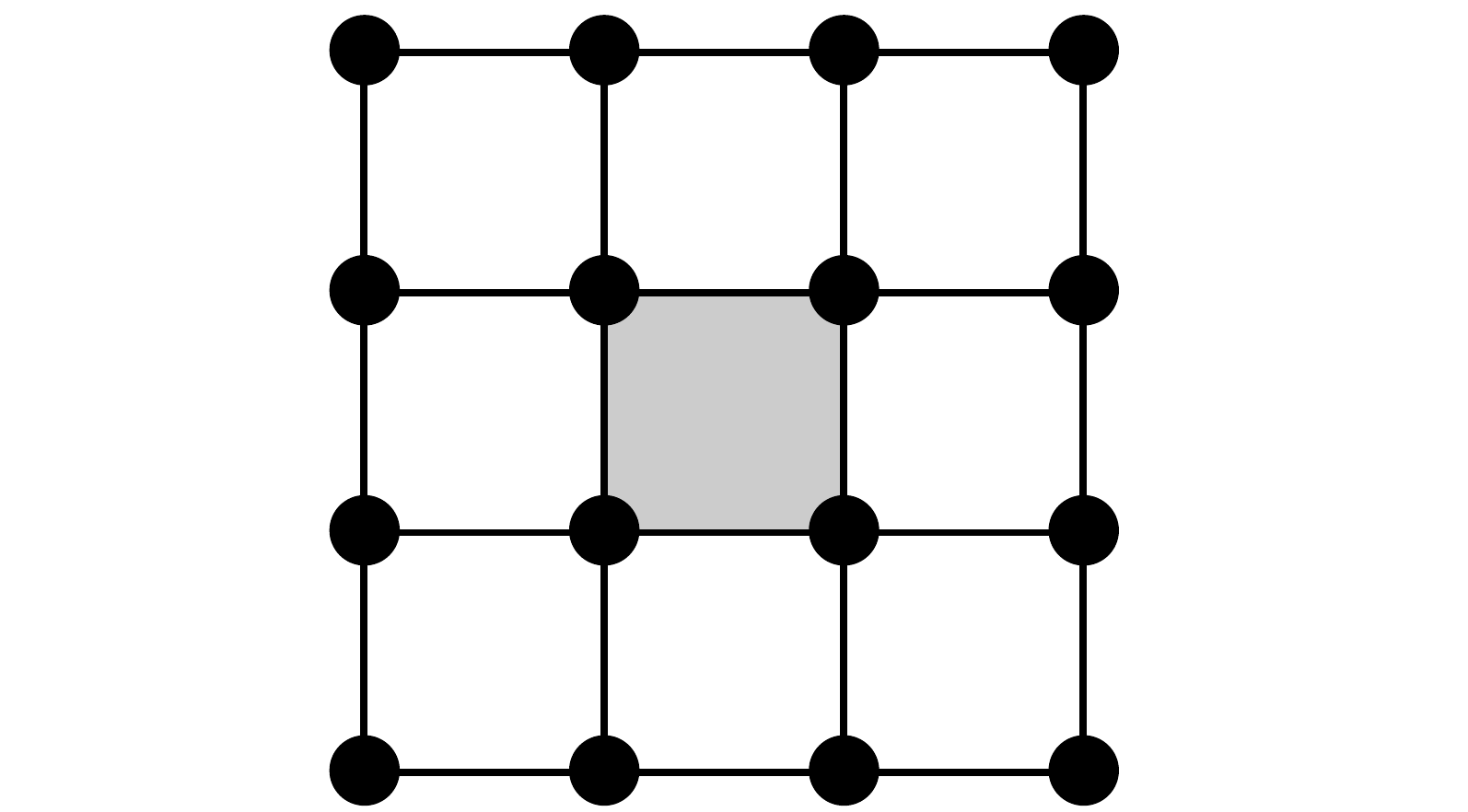}} \\[-3ex]
$N$ & \includegraphics[width=\rulefigwidth]{squaregrid-rule-n.pdf} & & \\
$J_2$ & \includegraphics[width=\rulefigwidth]{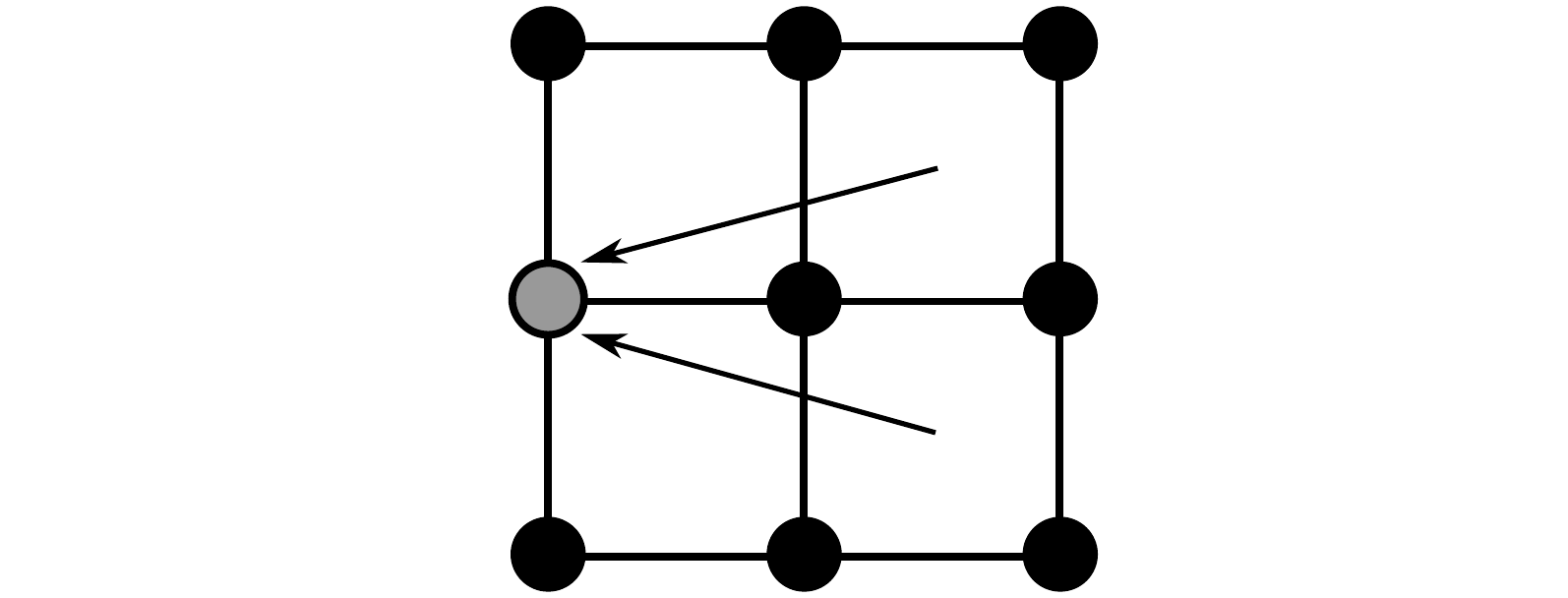} & & \\[5ex]
\hline&&&\\[-2ex]
&&\multirow{3}{\constraintwidth}{\includegraphics[width=\constraintwidth]{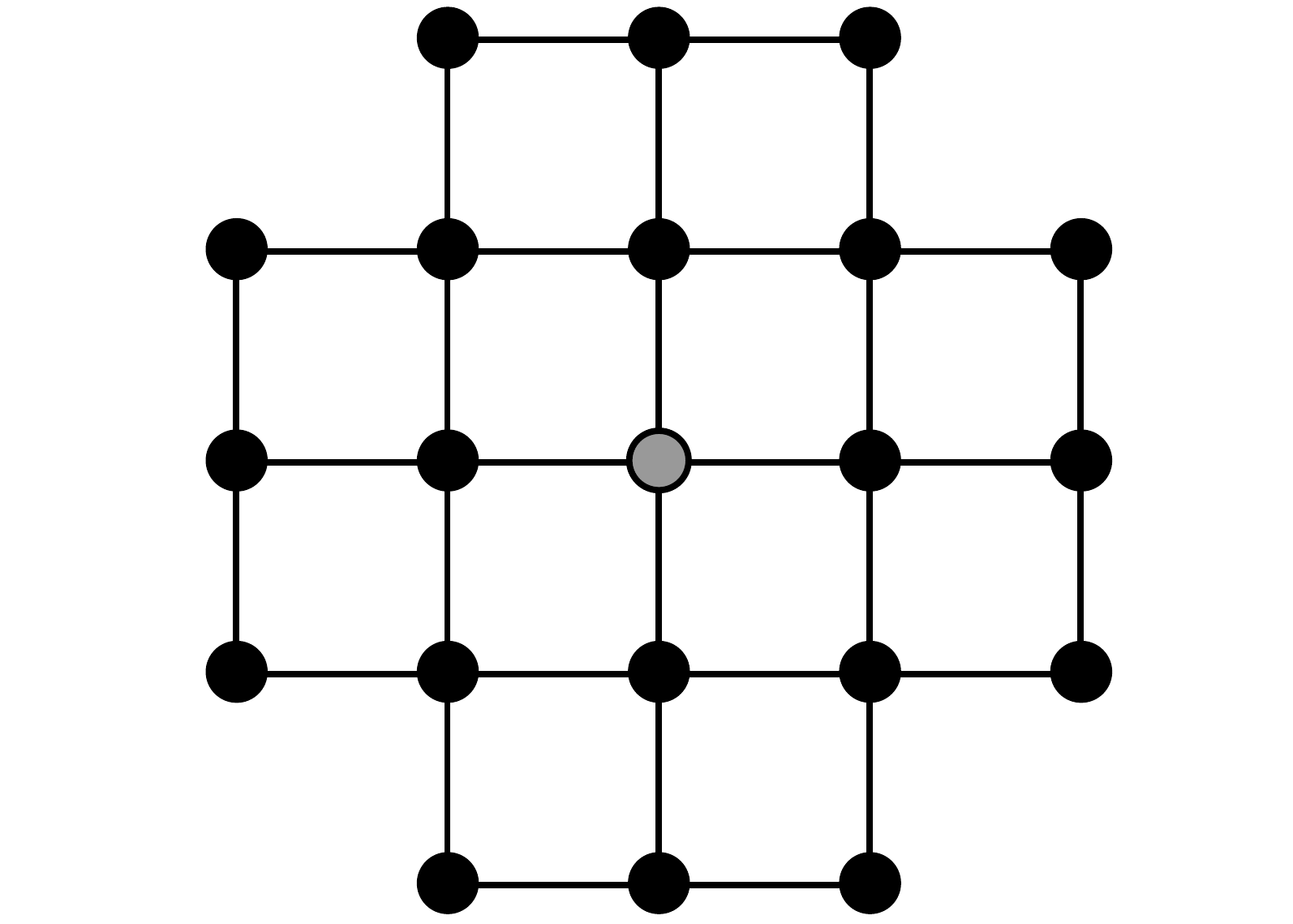}} &
\multirow{3}{\constraintwidth}{\includegraphics[width=\constraintwidth]{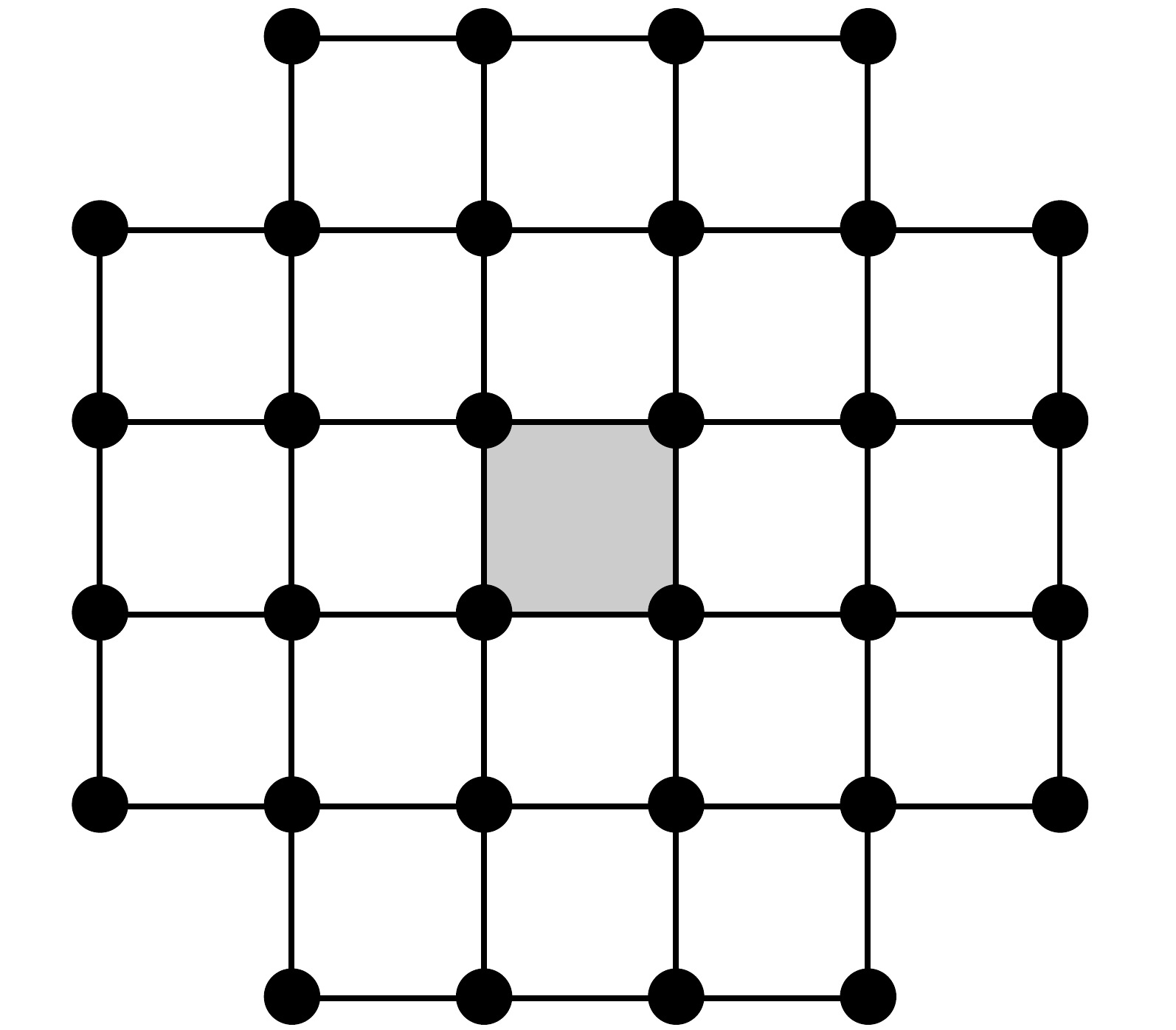}} \\[-3ex]
$N^+$ & \includegraphics[width=\rulefigwidth]{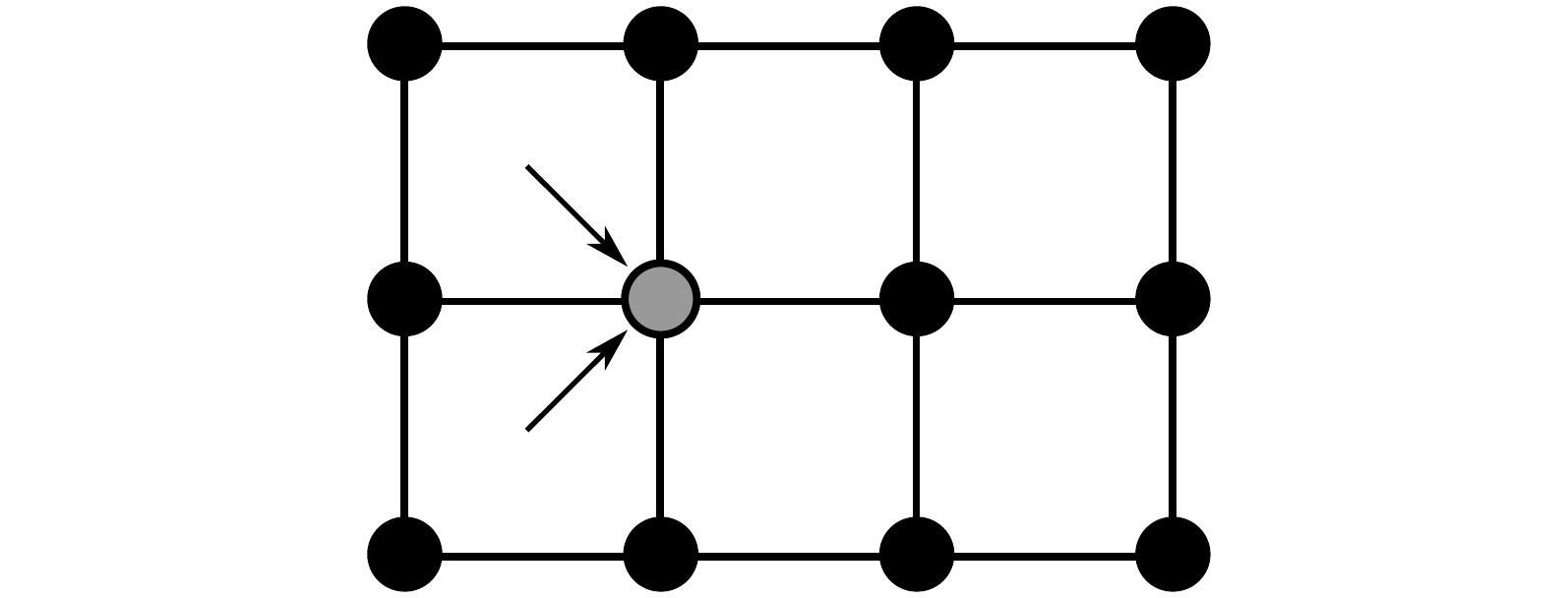} & & \\
$J_{2}^+$ & \includegraphics[width=\rulefigwidth]{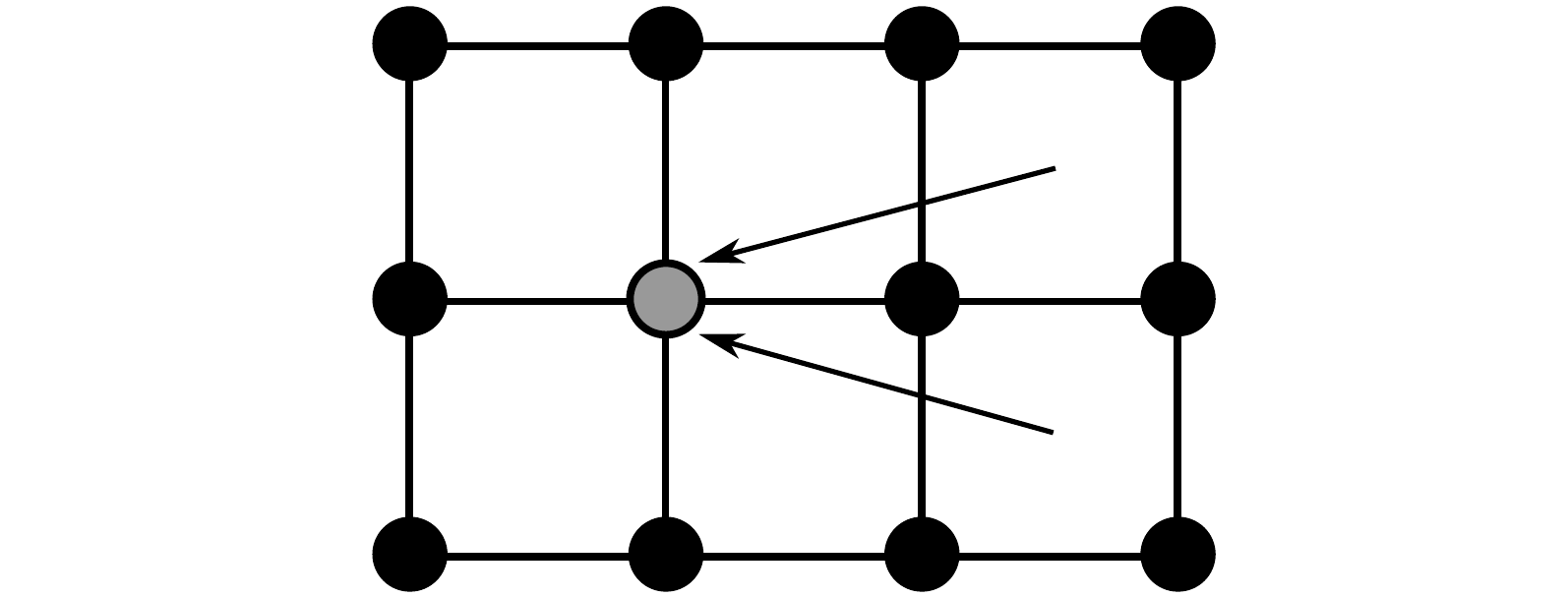} & & \\[5ex]
\hline
\end{tabular}
}

\caption{\label{tab:squarerules}Various Rules and Constraint Configurations in the Square Grid.}
\end{table}

\begin{table}[p]\small

\centering
\def\rulefigwidth{0.85in}
\def\constraintwidth{0.85in}
\def\bigconstraintwidth{1.7in}
\def\medconstraintwidth{1.3in}
\mbox{
\begin{tabular}[h]{|m{0.1in}m{\rulefigwidth}|m{\constraintwidth}m{\constraintwidth}|}
\hline
\multicolumn{2}{|c|}{\textbf{Rules} }
&
\multicolumn{2}{c|}{\textbf{Constraint Configurations}}\\
\hline&&&\\[-3ex]
\hline&&&\\[-2ex]
$V_1$ & 
\includegraphics[width=\rulefigwidth]{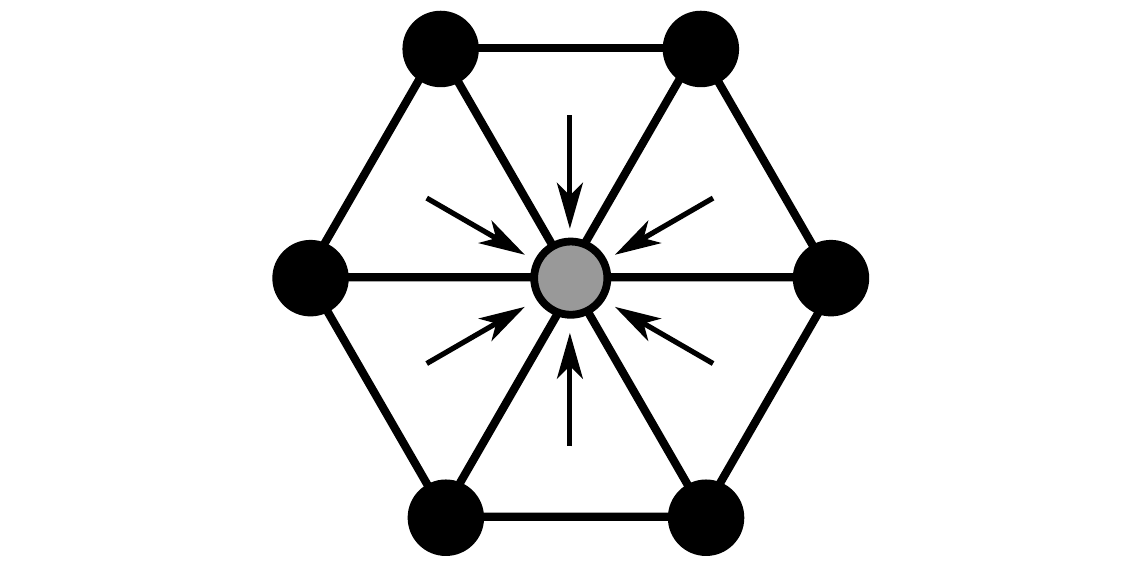} 
& 
\includegraphics[width=\constraintwidth]{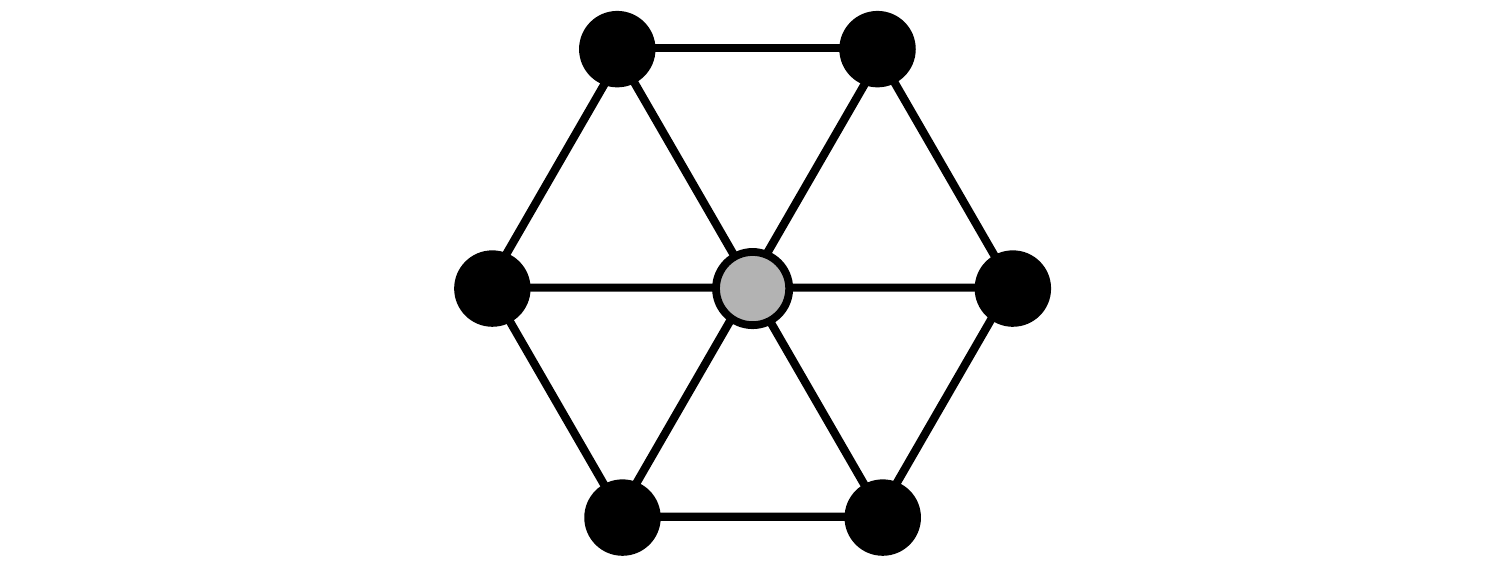} & 
\includegraphics[width=\constraintwidth]{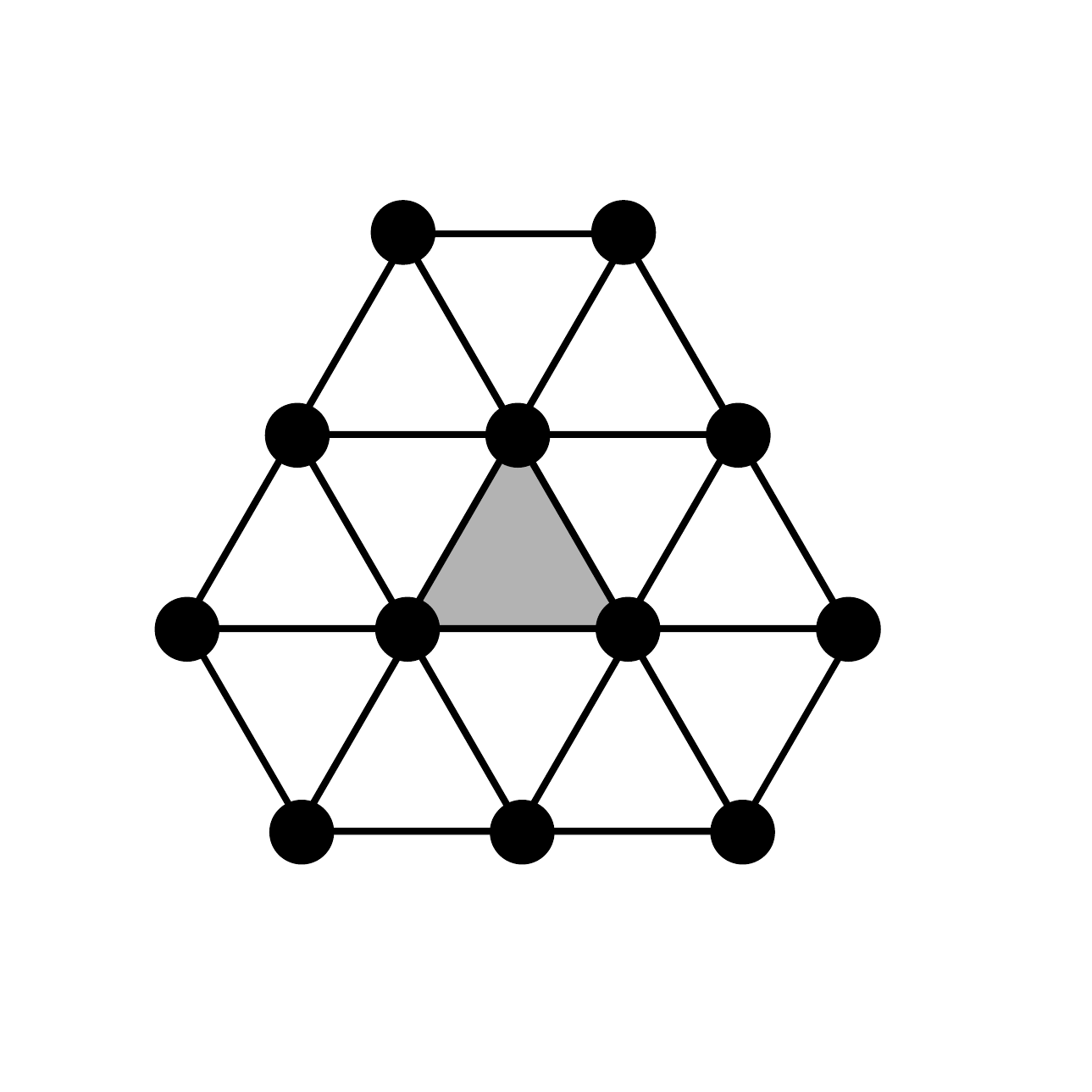} \\
\hline&&&\\[-2ex]
$S$ & 
\includegraphics[width=\rulefigwidth]{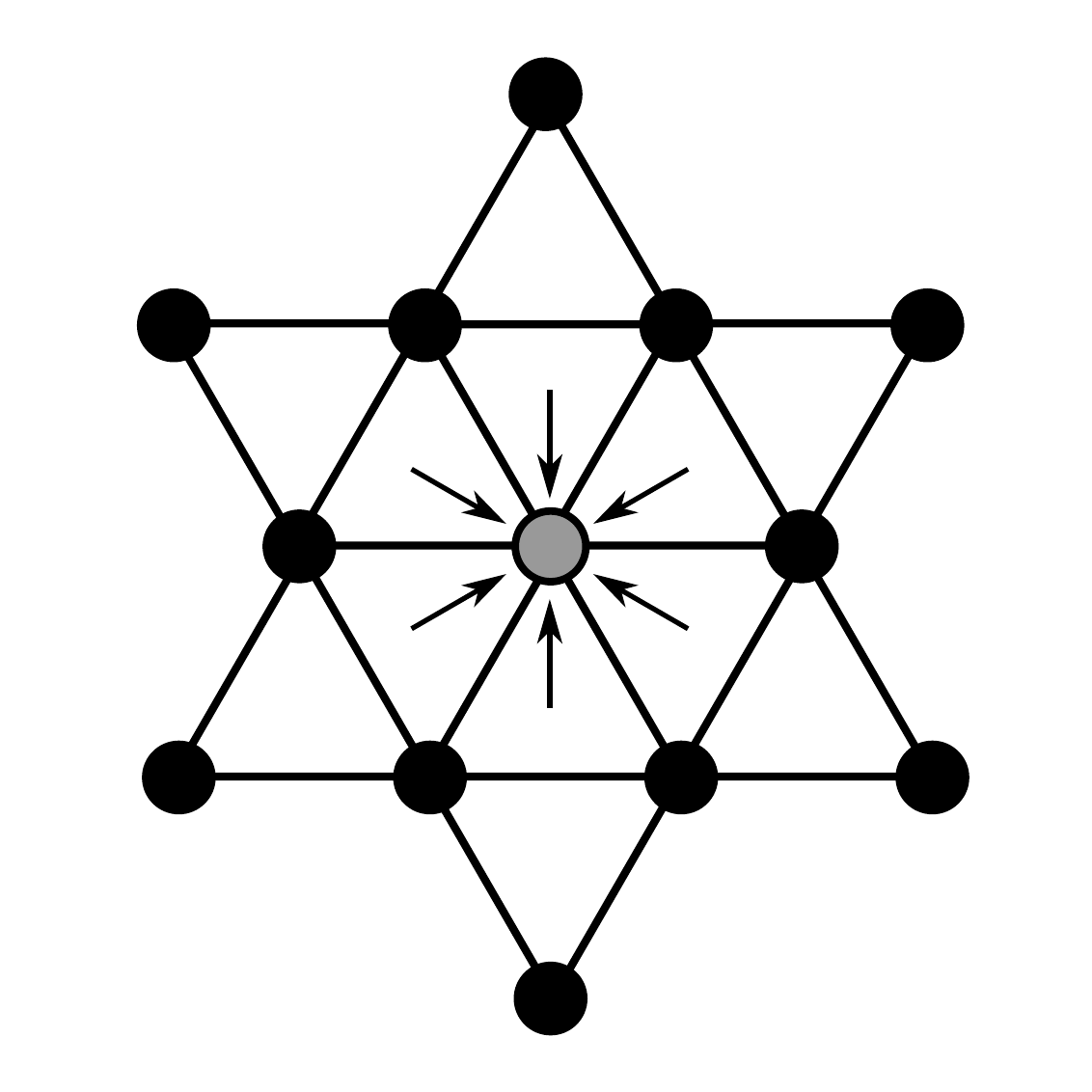} 
& 
\includegraphics[width=\constraintwidth]{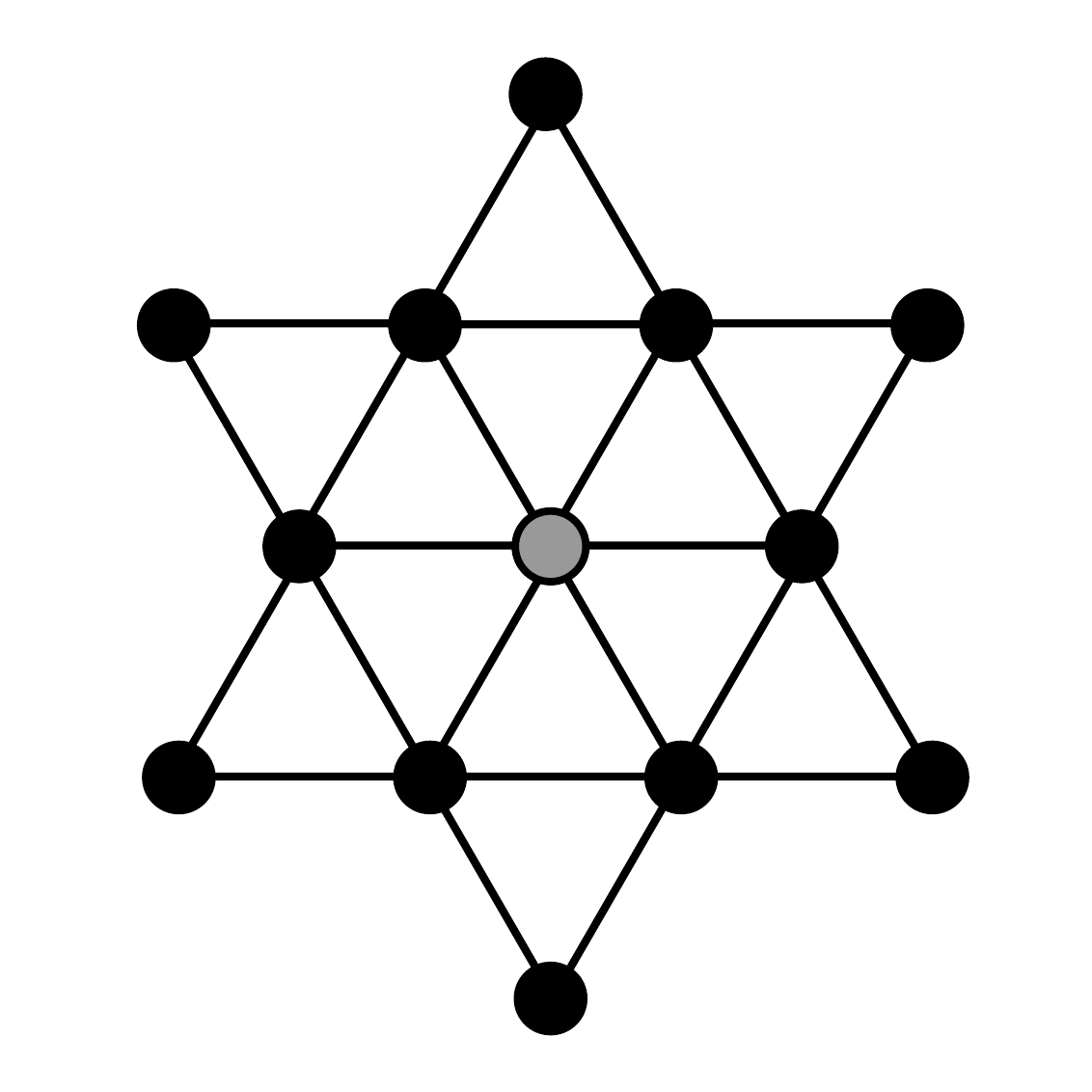} 
& 
\includegraphics[width=\constraintwidth]{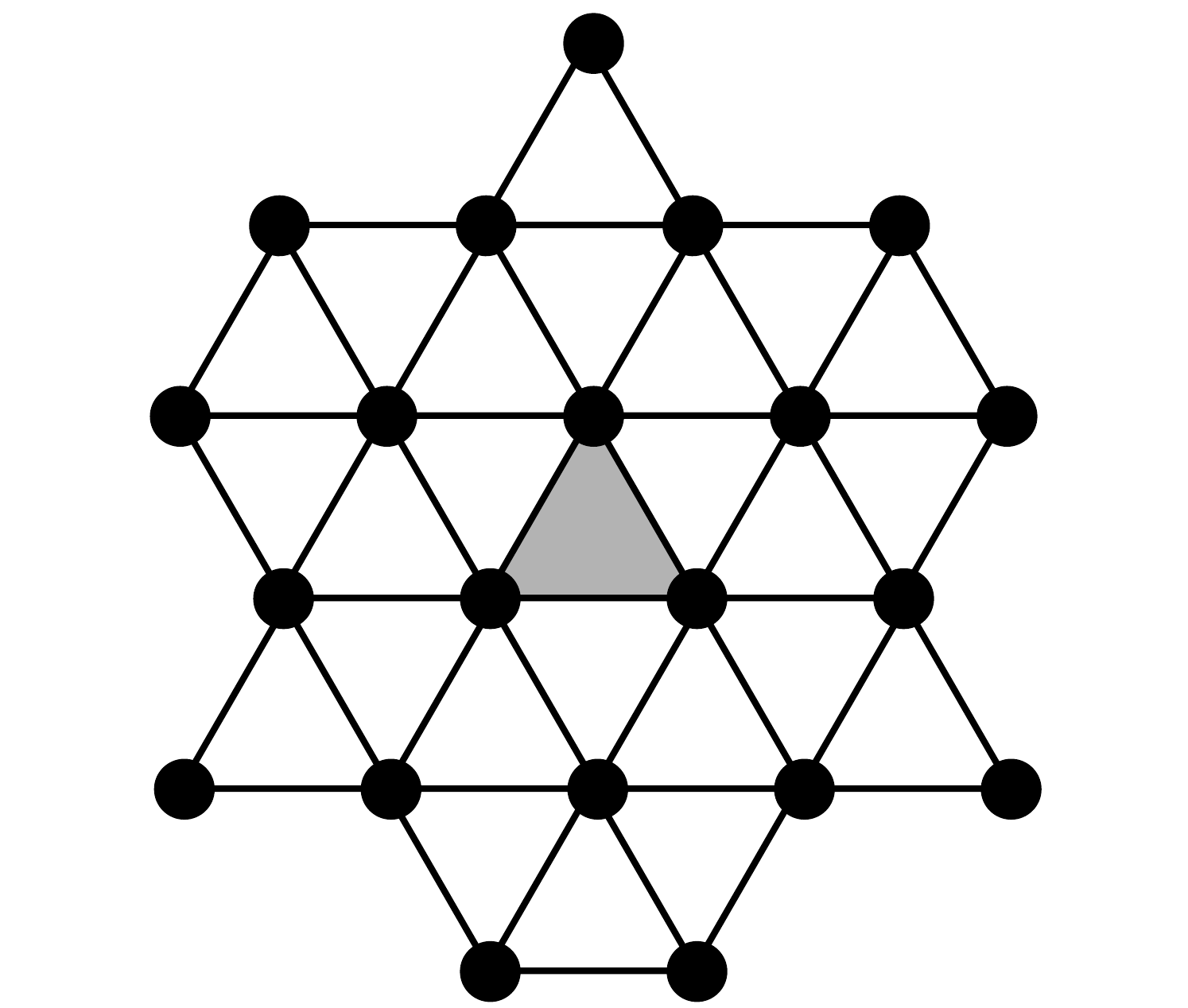} 
\\
\hline&&&\\[-2ex]
$N^+$ & 
\includegraphics[width=\rulefigwidth]{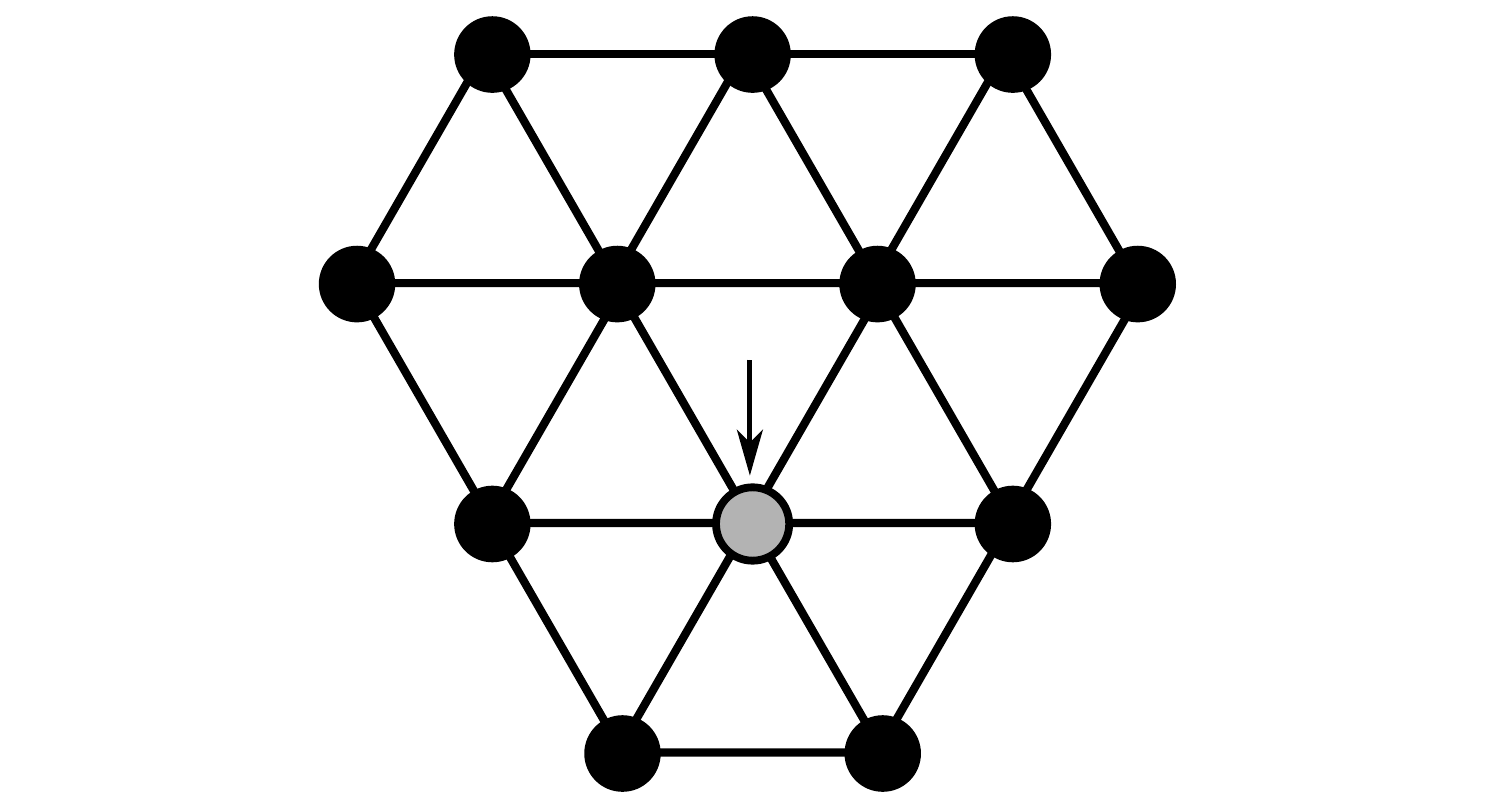} 
& 
\includegraphics[width=\constraintwidth]{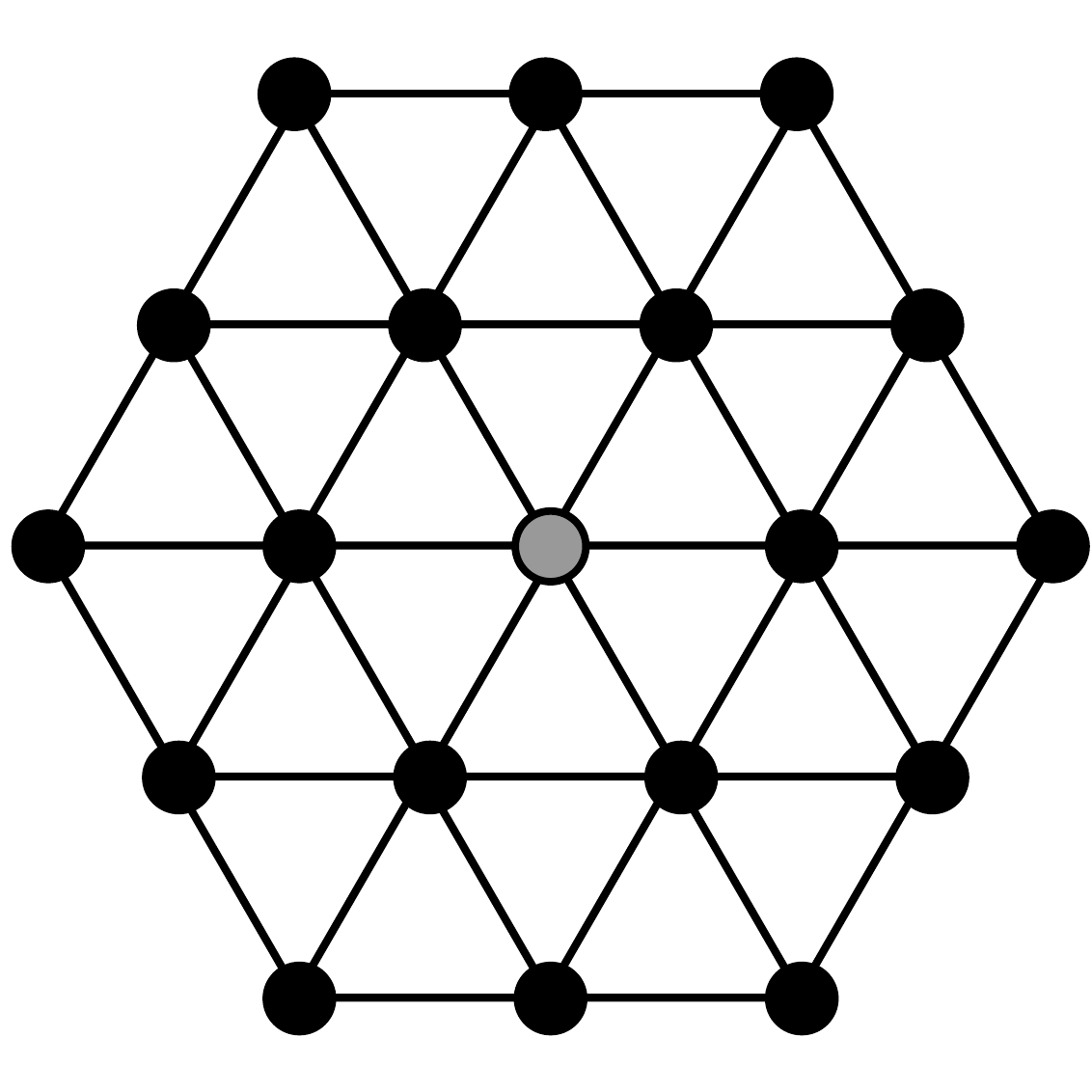} 
& 
\includegraphics[width=\constraintwidth]{trigrid-constraint-face-v1.pdf} 
\\
\hline&&&\\[-2ex]
$V_2$ & 
\includegraphics[width=\rulefigwidth]{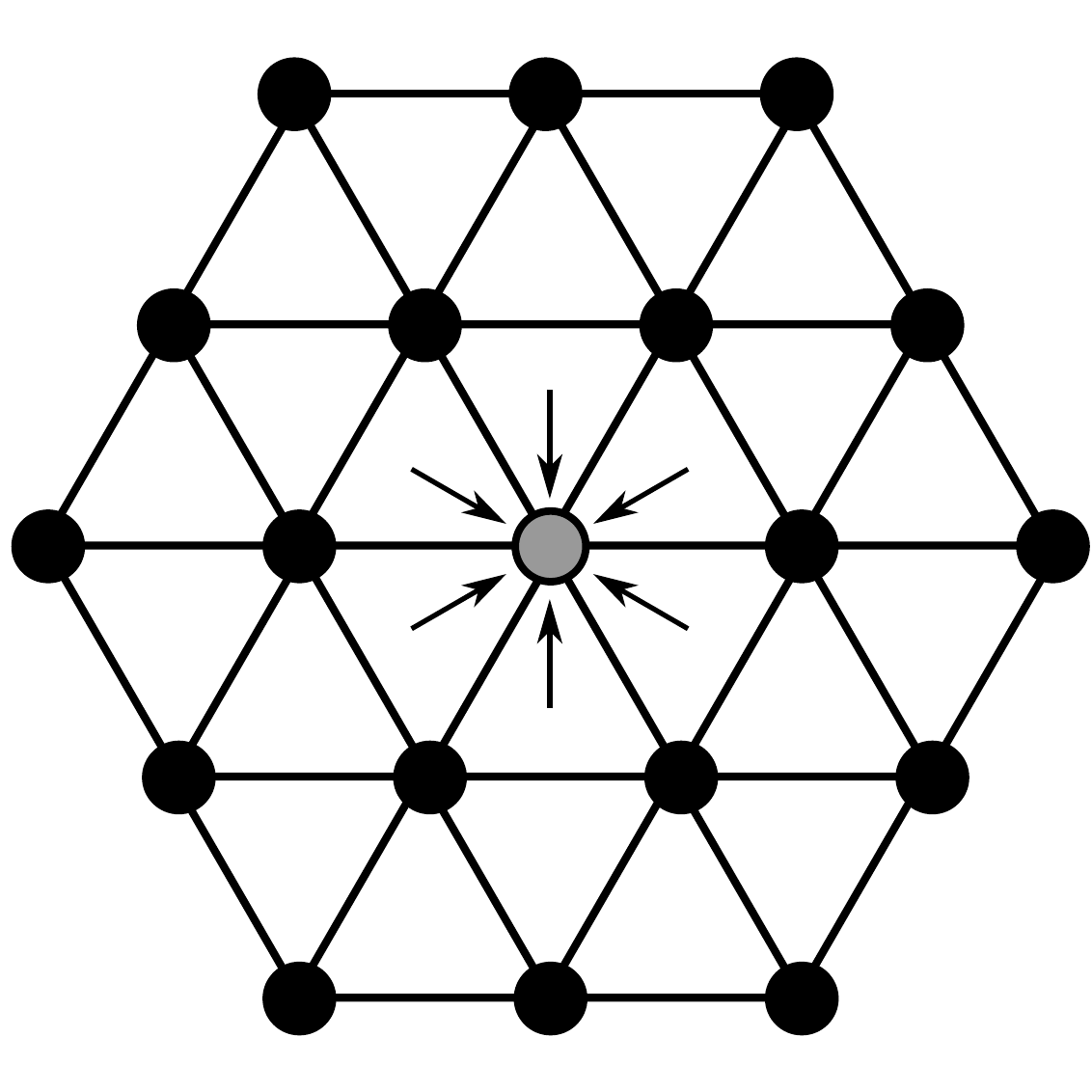} 
& 
\includegraphics[width=\constraintwidth]{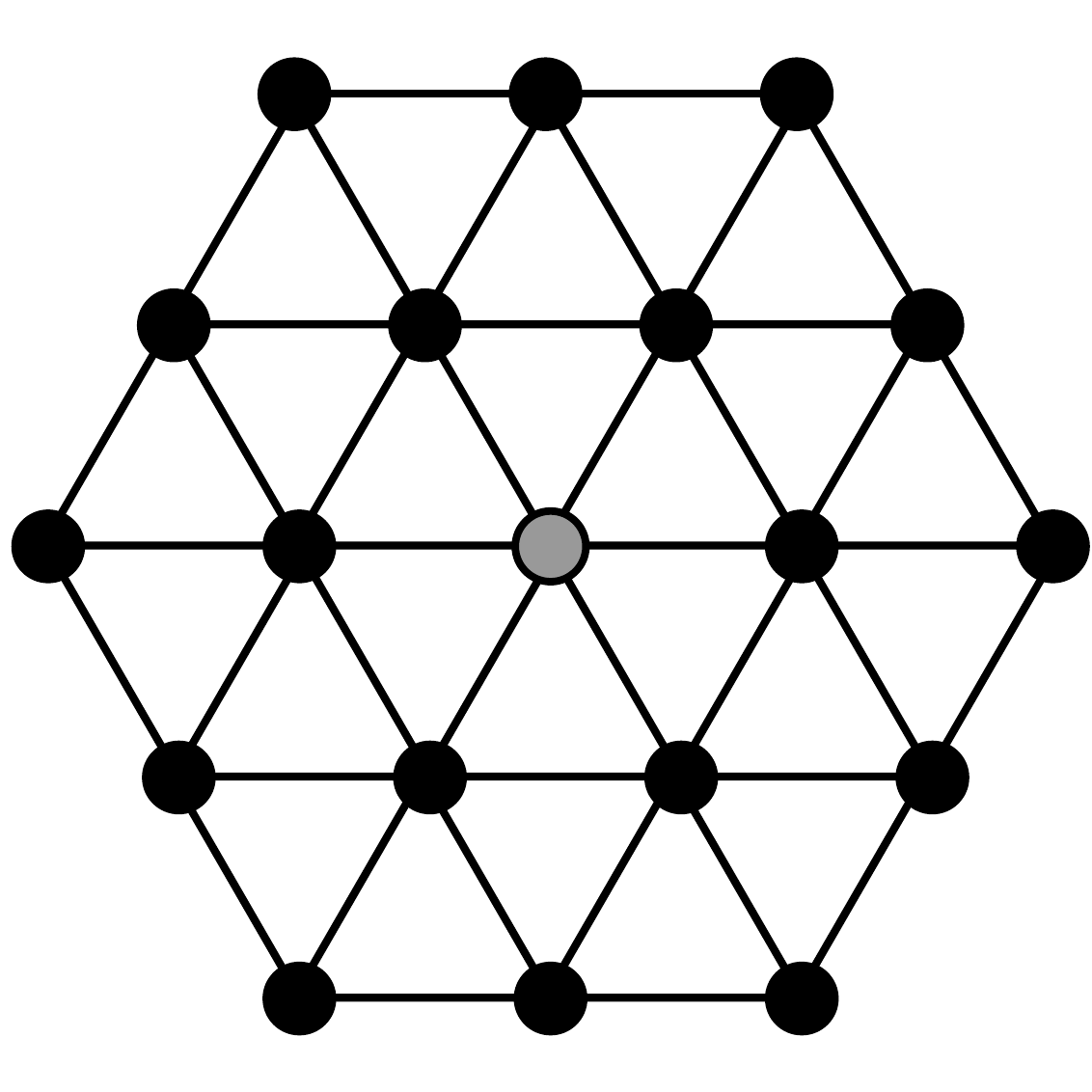} 
& 
\includegraphics[width=\constraintwidth]{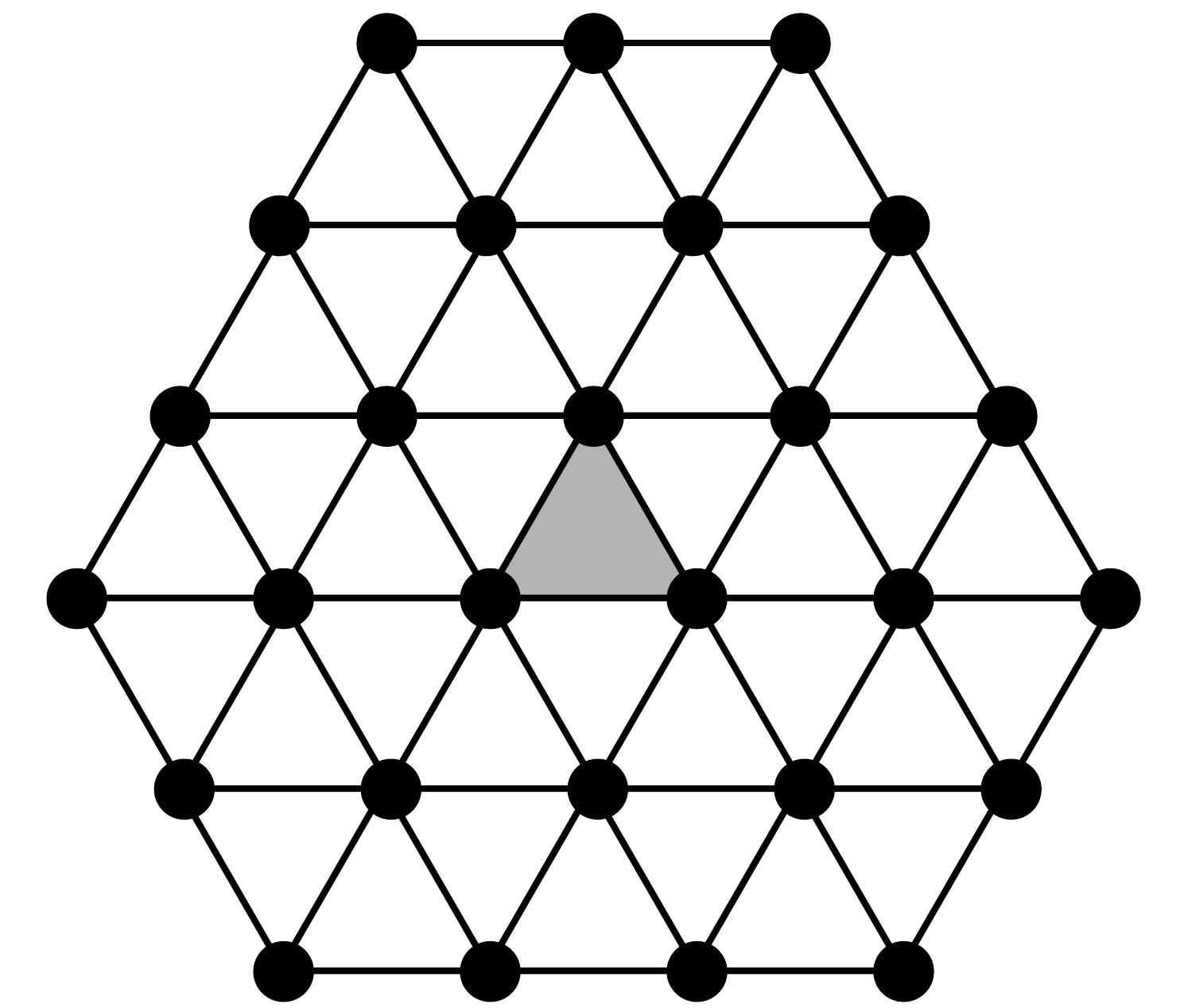} 
\\
\hline&&&\\[-2ex]
&& 
\multicolumn{2}{c|}{\multirow{4}{\bigconstraintwidth}{\includegraphics[width=\bigconstraintwidth]{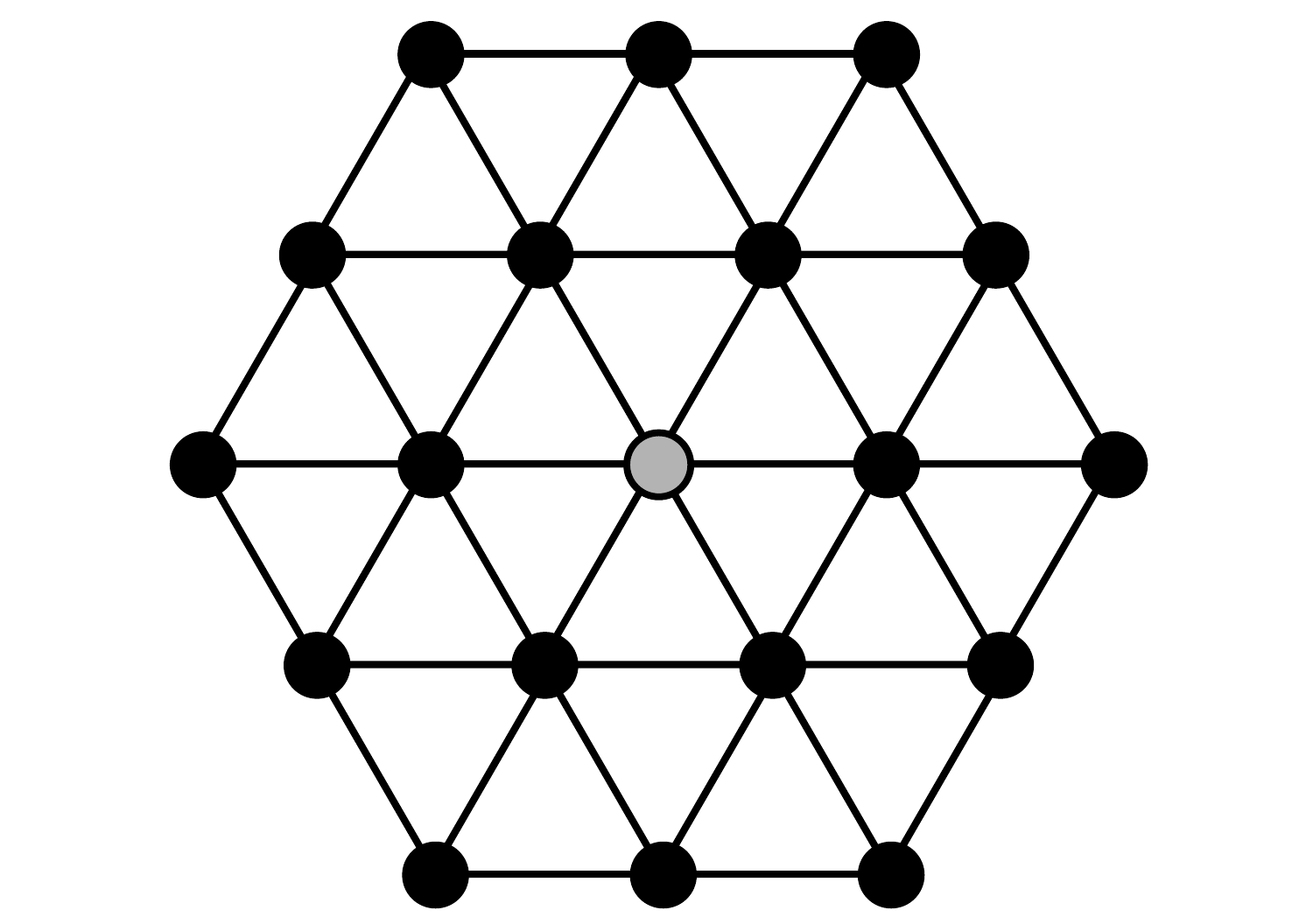}}}
\\[-3ex]
$C_{1}$ & 
\includegraphics[width=\rulefigwidth]{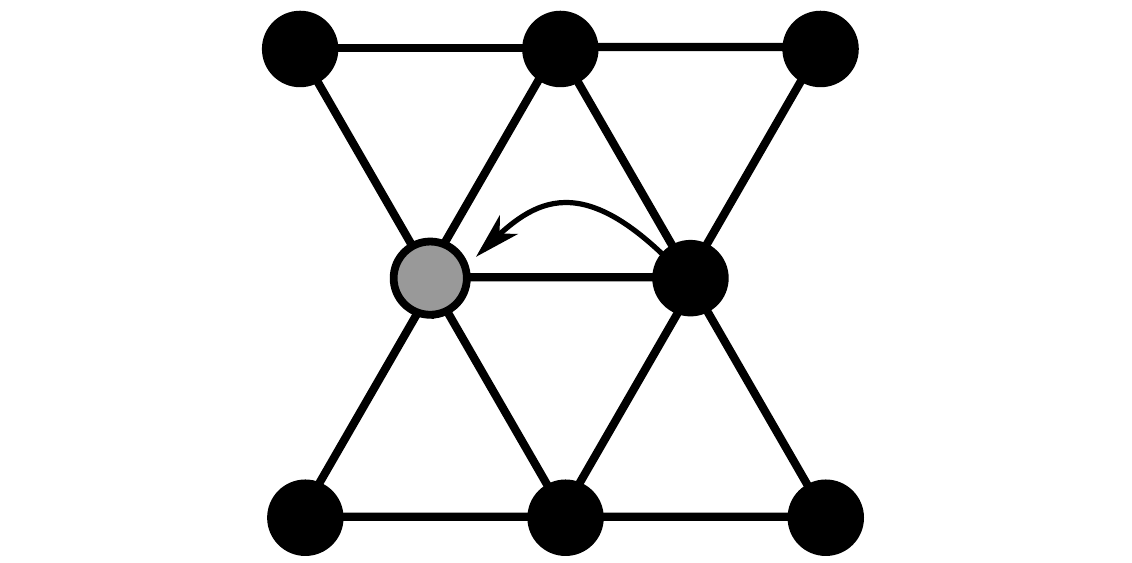} 
& &
 \\
\multirow{2}{0.1in}{$C_{2}$} & 
\includegraphics[width=\rulefigwidth]{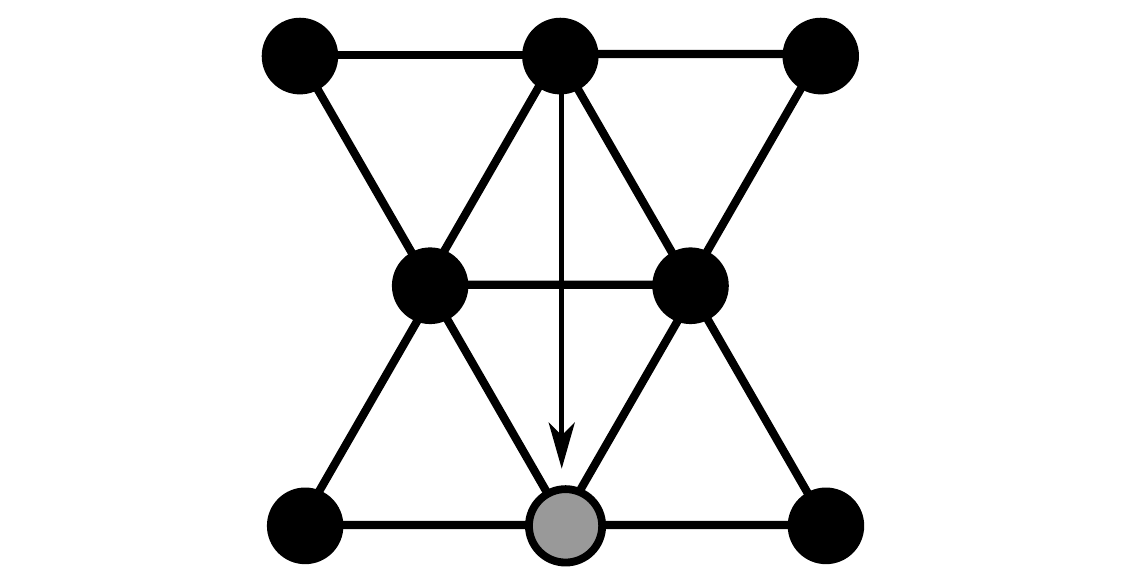} 
&& \\
& 
\includegraphics[width=\rulefigwidth]{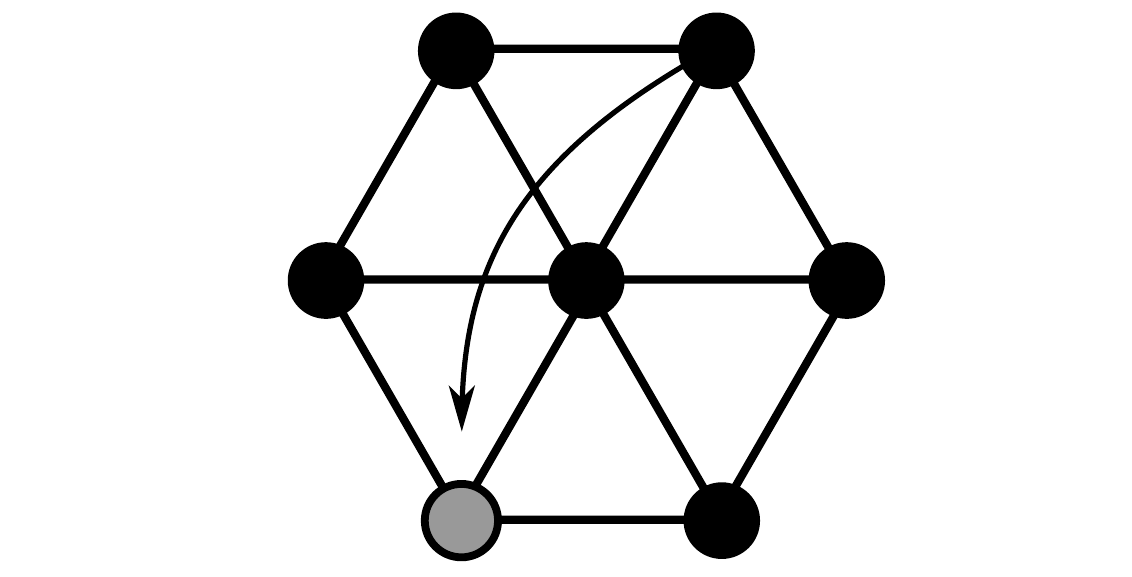} 
&& \\
\hline&&&\\[-2ex]
$C_{1}^+$ & 
\includegraphics[width=\rulefigwidth]{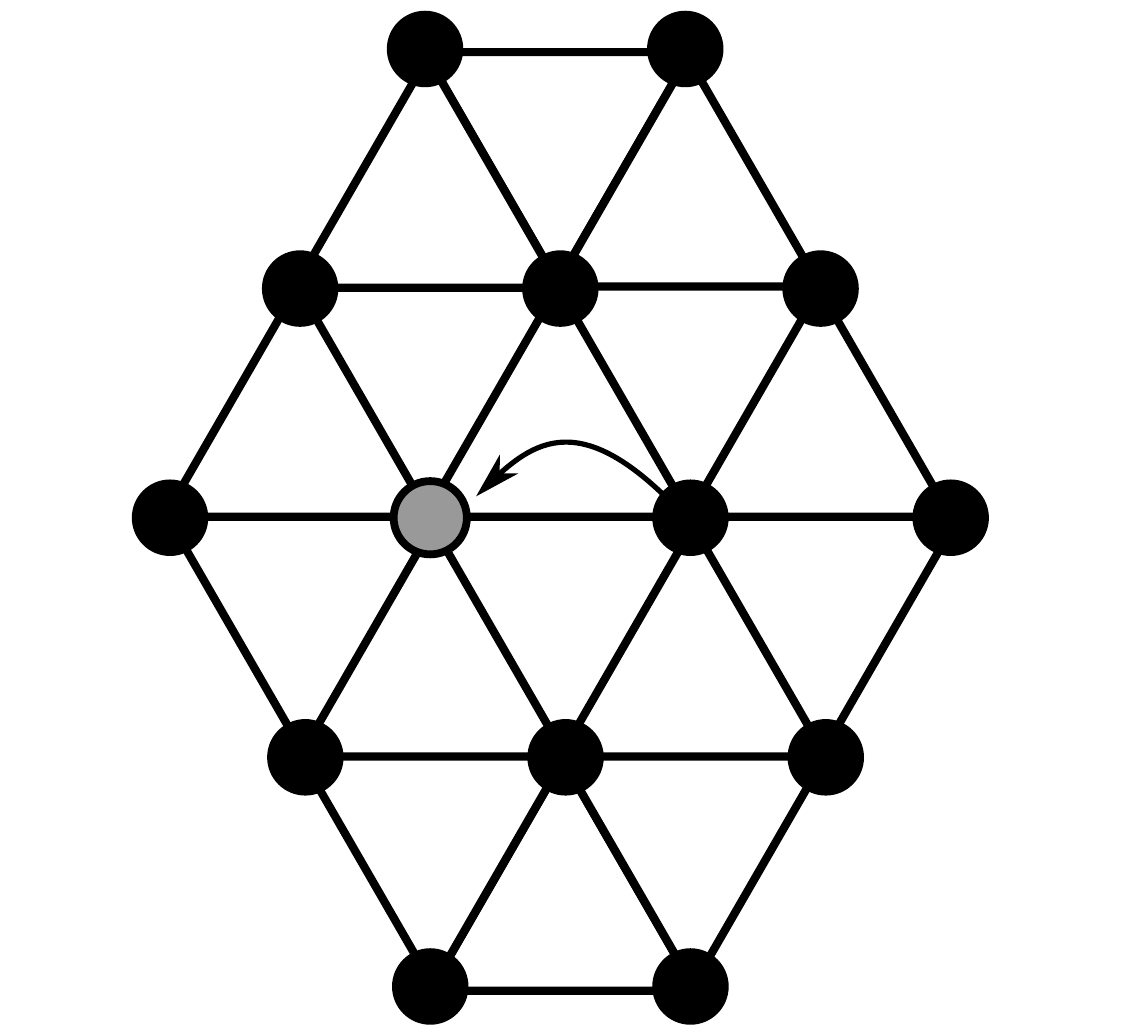} 
& 
\multicolumn{2}{c|}{\multirow{3}{\bigconstraintwidth}{\includegraphics[width=\bigconstraintwidth]{trigrid-constraint-vert-c1p2p.pdf}}} 
\\
\multirow{2}{0.1in}{$C_{2}^+$} & 
\includegraphics[width=\rulefigwidth]{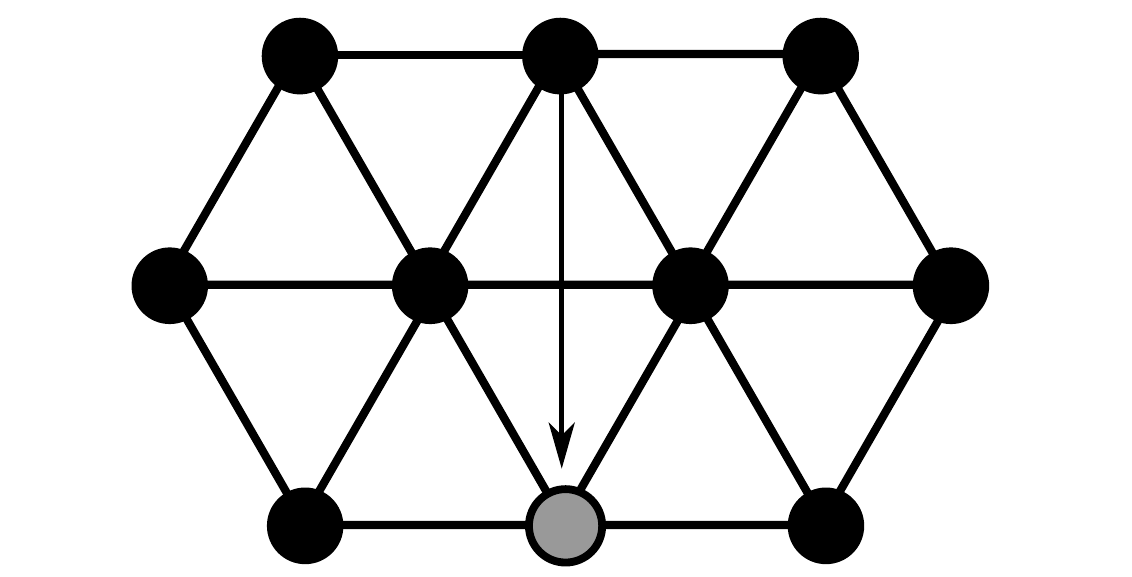} 
&& \\
& 
\includegraphics[width=\rulefigwidth]{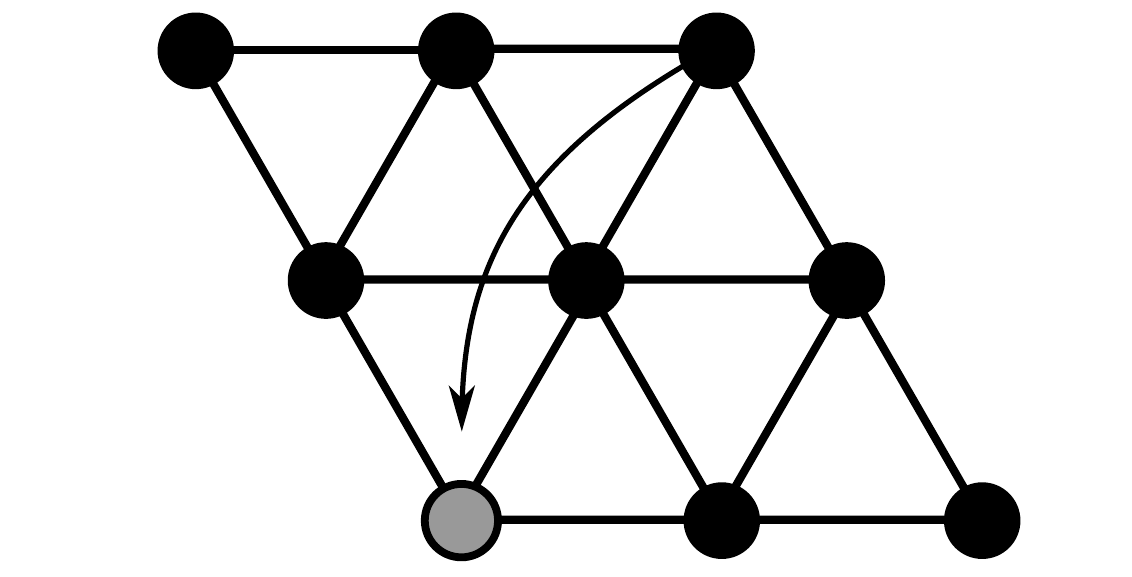} 
&& \\
\hline
\end{tabular}
\quad
\begin{tabular}[h]{|m{0.1in}m{\rulefigwidth}|m{\constraintwidth}m{\constraintwidth}|}
\hline
\multicolumn{2}{|c|}{\textbf{Rules} }
&
\multicolumn{2}{c|}{\textbf{Constraint Configurations}}\\
\hline&&&\\[-3ex]
\hline&&&\\[-2ex]
$C_1$ & 
\includegraphics[width=\rulefigwidth]{trigrid-rule-c1.pdf} 
& 
\multicolumn{2}{c|}{\multirow{4}{\bigconstraintwidth}{\includegraphics[width=\bigconstraintwidth]{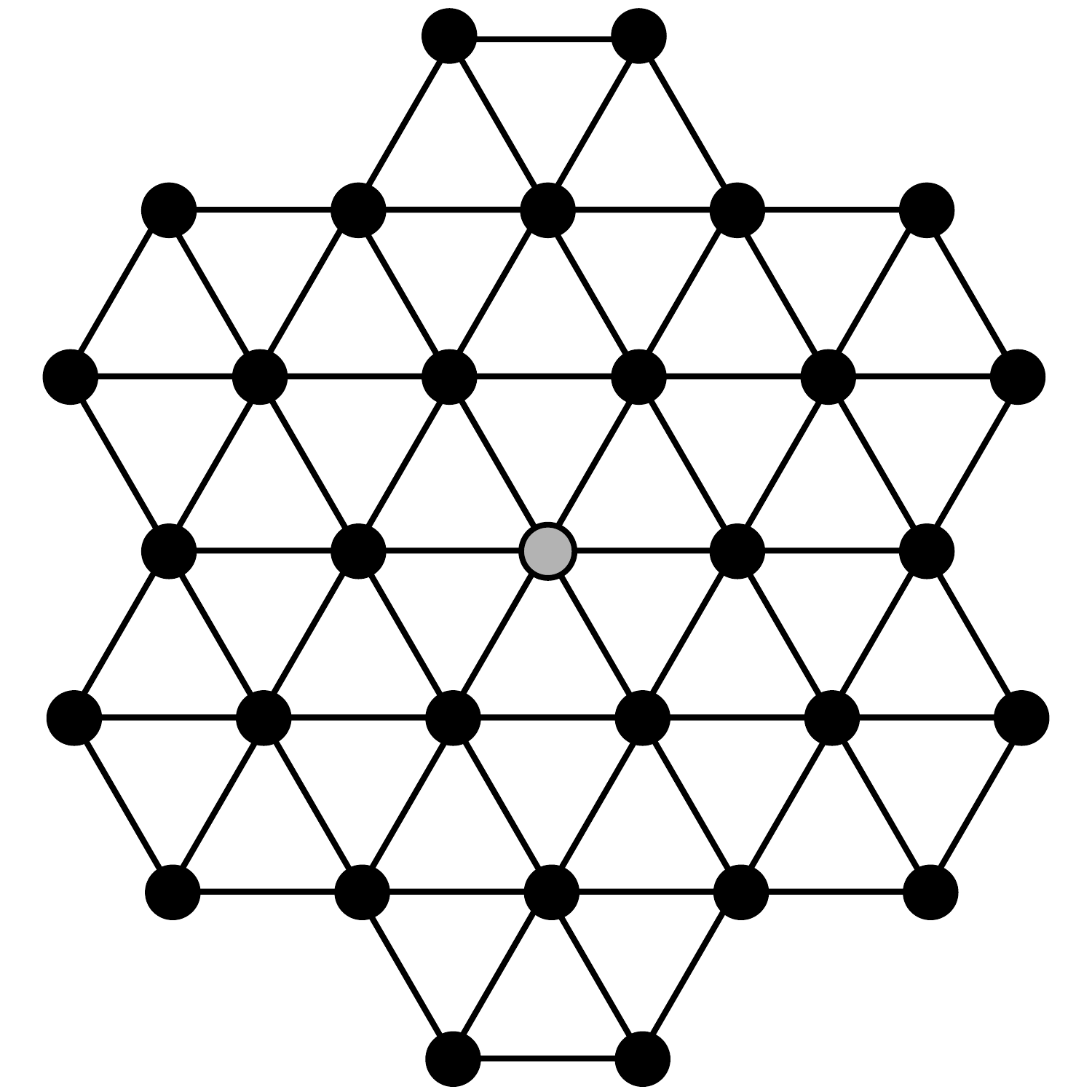}}} 
\\[4ex]
\multirow{2}{0.1in}{$C_2$} & 
\includegraphics[width=\rulefigwidth]{trigrid-rule-c2a.pdf} 
&& \\
& 
\includegraphics[width=\rulefigwidth]{trigrid-rule-c2b.pdf} 
&& \\[4ex]
\multirow{2}{0.1in}{$C_3$} & 
\includegraphics[width=\rulefigwidth]{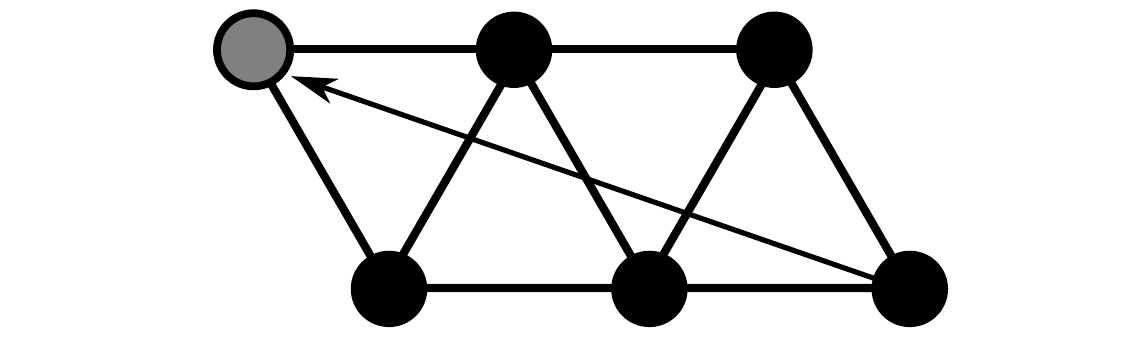} 
&& \\
 & 
\includegraphics[width=\rulefigwidth]{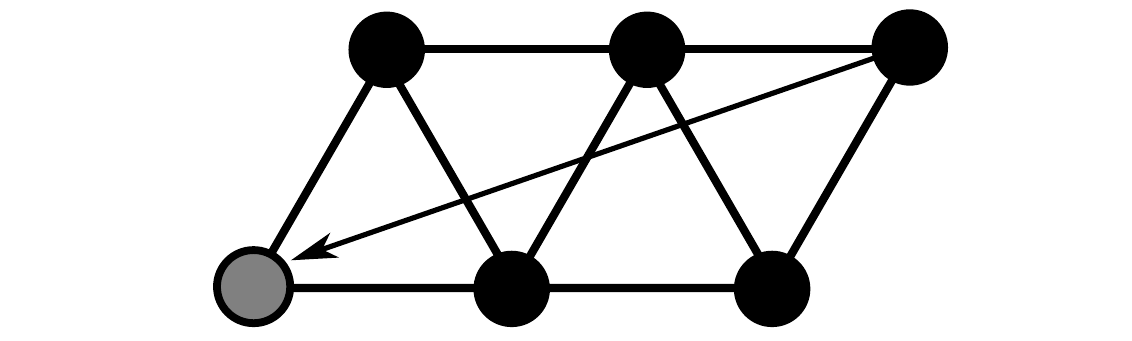} 
&& \\[5ex]
\hline&&&\\[-2ex]
$C_{1}^+$ & 
\includegraphics[width=\rulefigwidth]{trigrid-rule-c1p.pdf} 
& 
\multicolumn{2}{c|}{\multirow{4}{\bigconstraintwidth}{\includegraphics[width=\bigconstraintwidth]{trigrid-constraint-vert-c1p2p3.pdf}}} \\
\multirow{2}{0.1in}{$C_{2}^+$} & 
\includegraphics[width=\rulefigwidth]{trigrid-rule-c2ap.pdf} 
&& \\
& 
\includegraphics[width=\rulefigwidth]{trigrid-rule-c2bp.pdf} 
&& \\
\multirow{2}{0.1in}{$C_{3}^+$} & 
\includegraphics[width=\rulefigwidth]{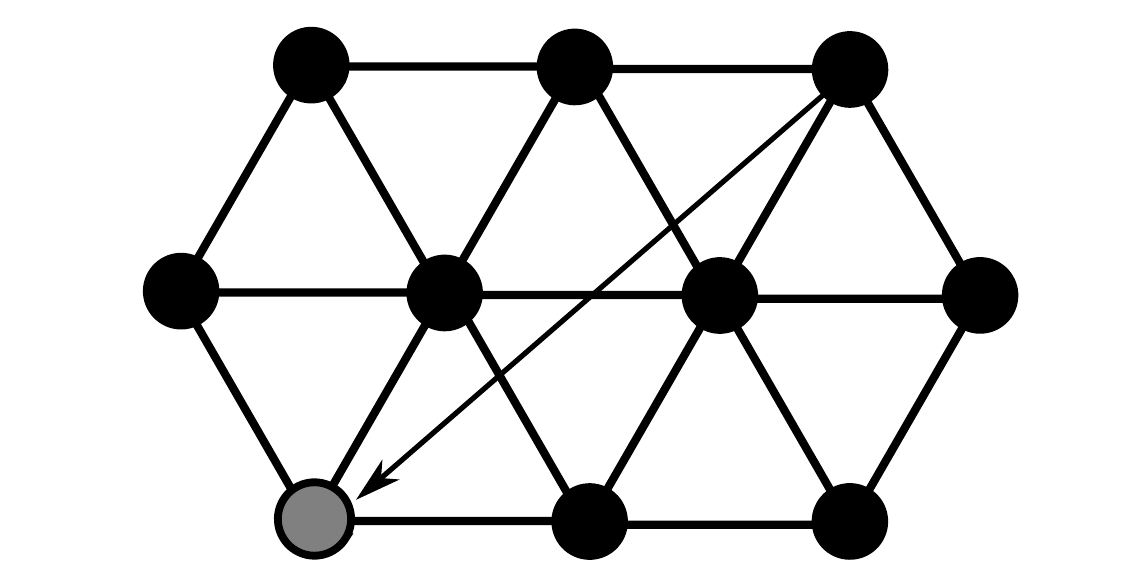} 
&& \\[4ex]
& 
\includegraphics[width=\rulefigwidth]{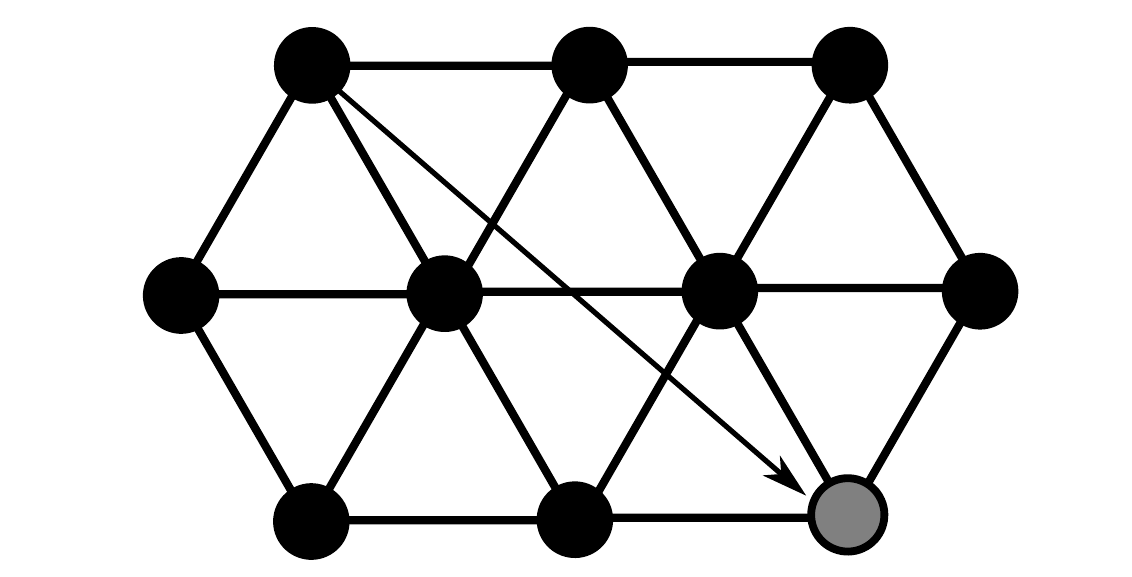} 
&& \\[4ex]
\hline&&&\\[-2ex]
$V_1$ & \includegraphics[width=\rulefigwidth]{trigrid-rule-v1.pdf} & 
\multirow{2}{\constraintwidth}{\includegraphics[width=\constraintwidth]{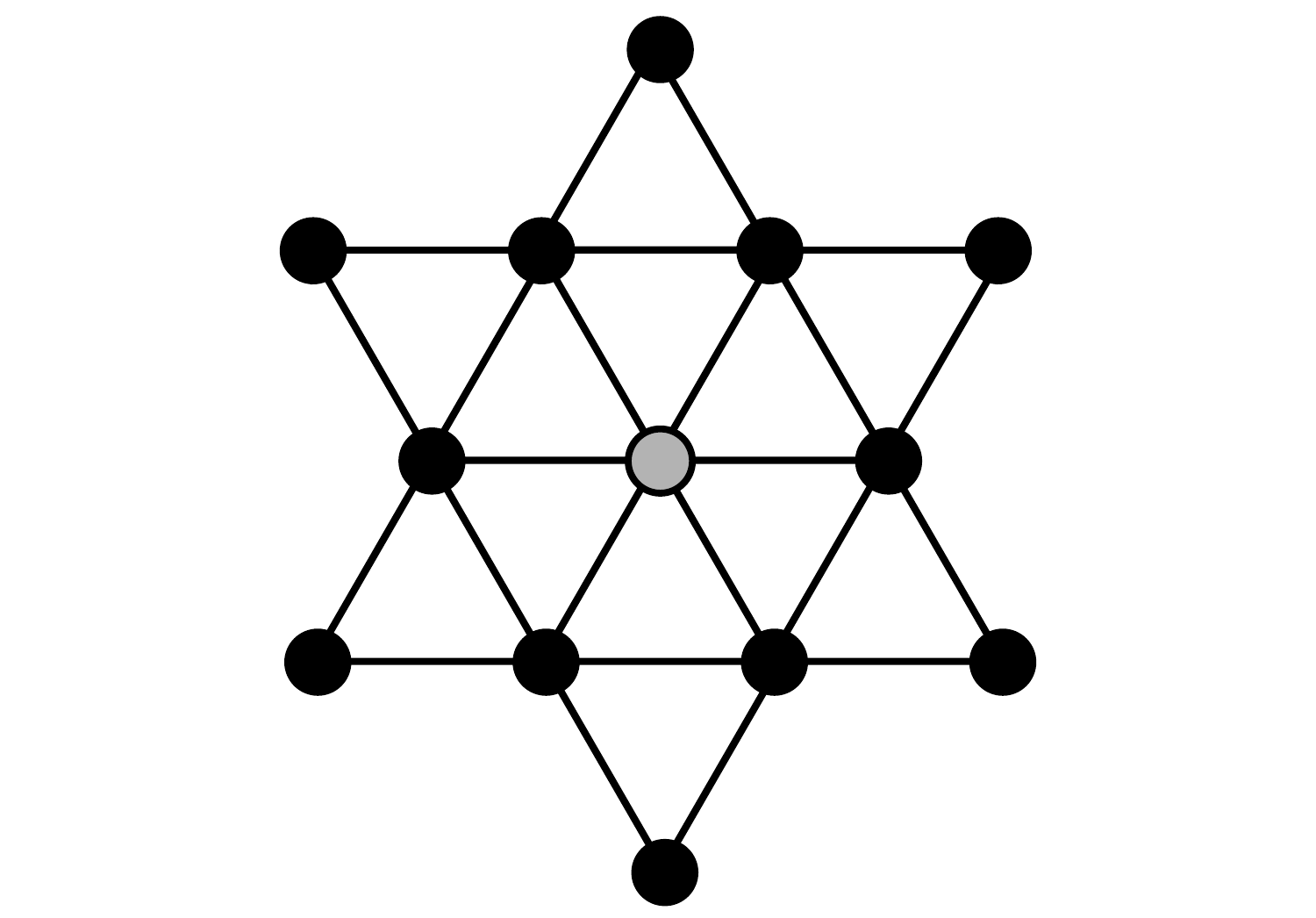}}
 &
\multirow{2}{\constraintwidth}{\includegraphics[width=\constraintwidth]{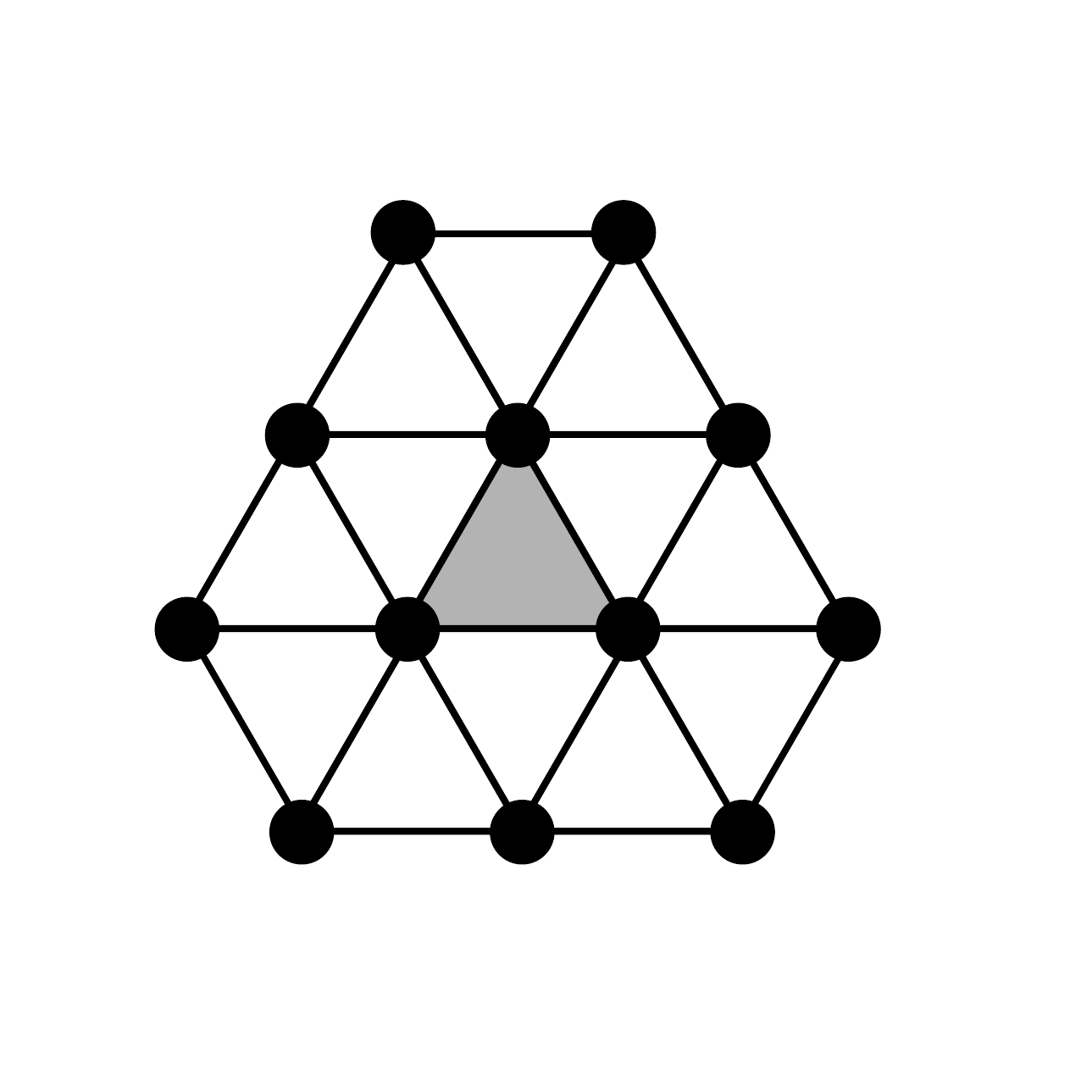}} \\[3ex]
$J_2$ & 
\includegraphics[width=\rulefigwidth]{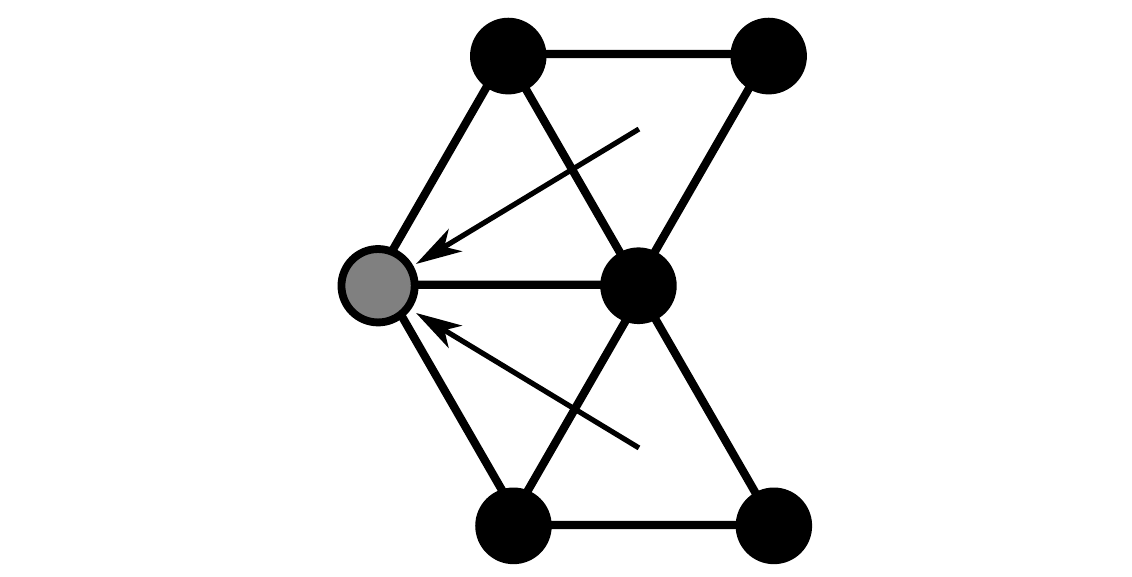} 
&& \\[5ex]
\hline&&&\\[-2ex]
$S$ & \includegraphics[width=\rulefigwidth]{trigrid-rule-star.pdf} & 
\multirow{2}{\constraintwidth}{\includegraphics[width=\constraintwidth]{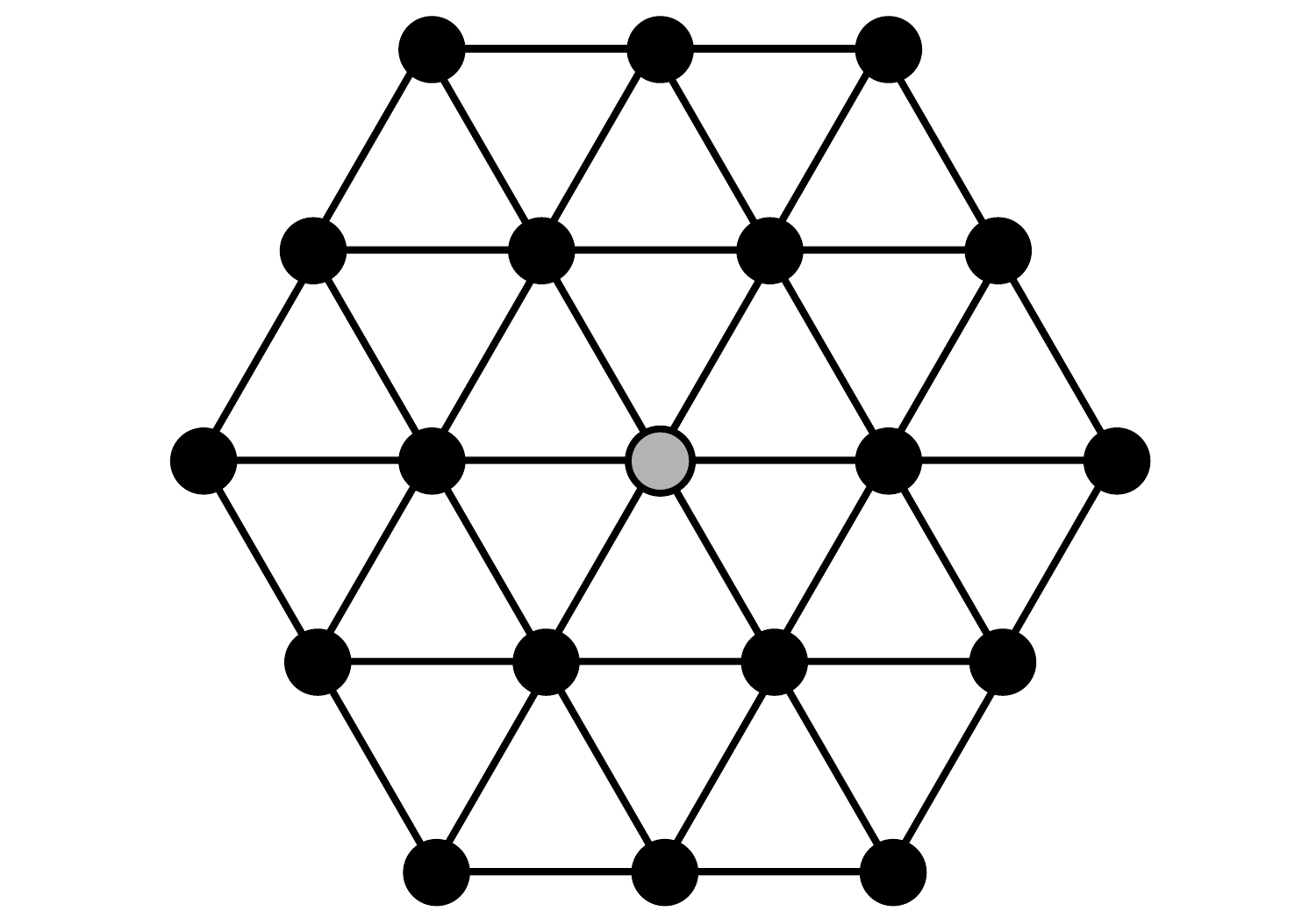}}
 &
\multirow{2}{\constraintwidth}{\includegraphics[width=\constraintwidth]{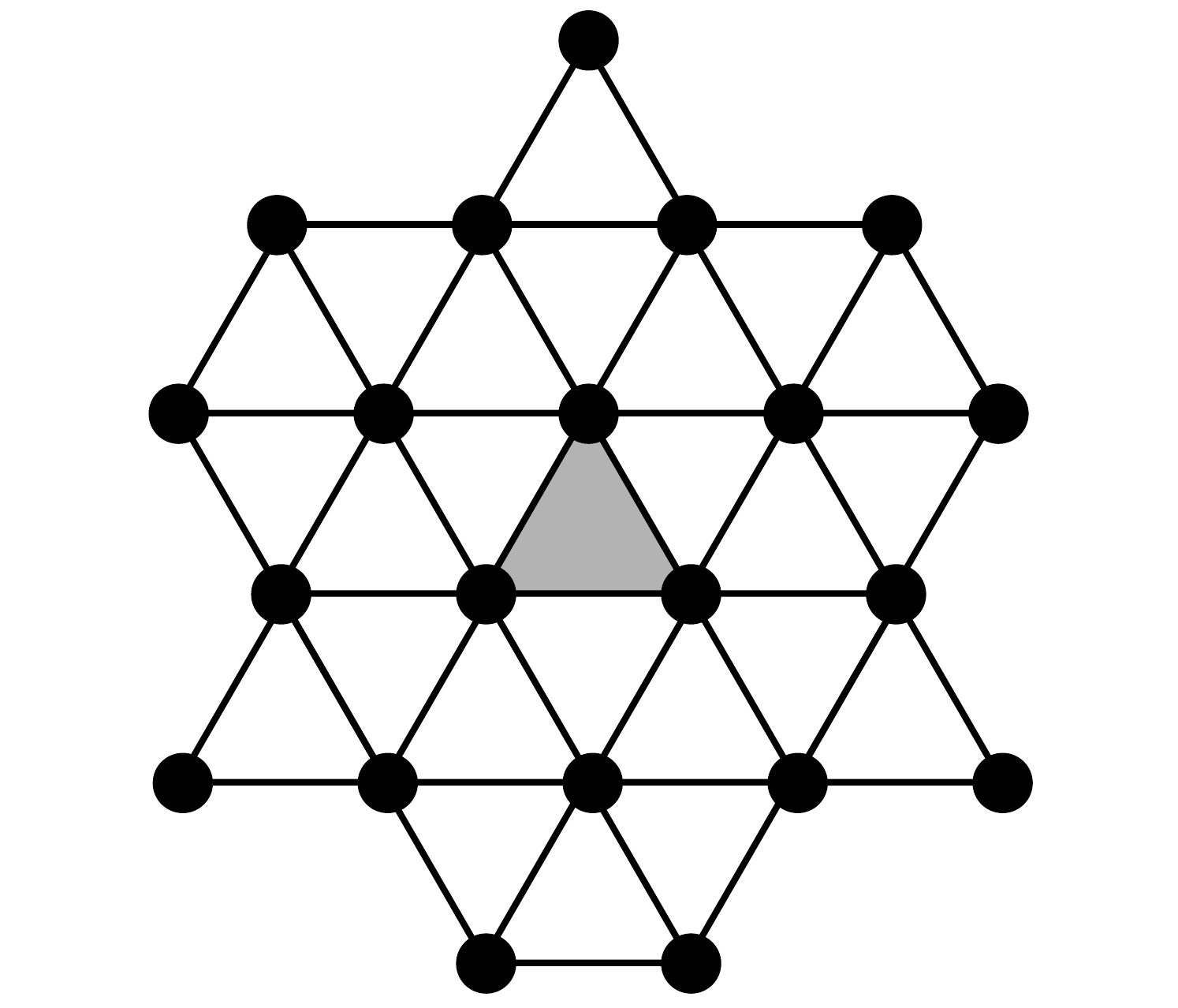}} 
\\
$J_{2}^+$ & 
\includegraphics[width=\rulefigwidth]{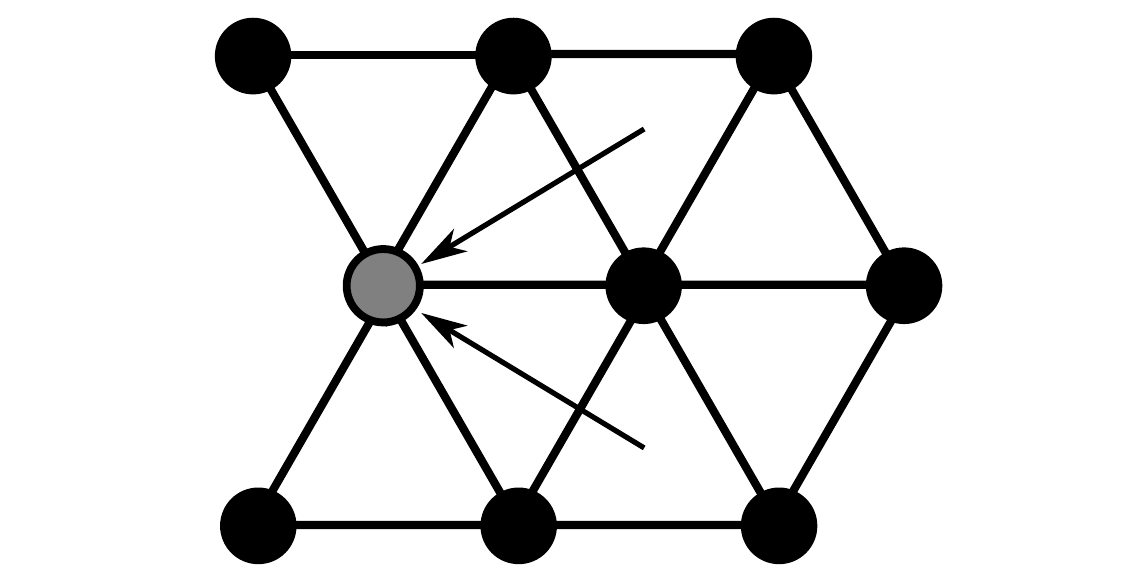} 
&& \\
\hline
\end{tabular}
}

\caption{\label{tab:triangularrules}Various Rules and Constraint Configurations in the Triangular Grid.}
\end{table}

\section{Bounds on Codes Using ADAGE}\label{sec:bounds}

\begin{table}[H]
{\small
\centering
\def\colonewidth{0.75in}
\renewcommand{\arraystretch}{1.2}
\begin{tabular}[h]{|c|rr@{\ $\approx$\ }l|rr@{\ $\approx$\ }l|rr@{\ $\approx$\ }l|}
\hline
\emph{Set Type} & \multicolumn{3}{c|}{\bf Hexagonal Grid} &\multicolumn{3}{c|}{\bf Square Grid} &\multicolumn{3}{c|}{\bf Triangular Grid} \\
\hline&\multicolumn{3}{c|}{\ }&\multicolumn{3}{c|}{\ }&\multicolumn{3}{c|}{\ }\\[-3ex]
\hline&\multicolumn{3}{c|}{\ }&\multicolumn{3}{c|}{\ }&\multicolumn{3}{c|}{\ }\\[-2ex]
\multirow{1}{\colonewidth}{Dominating Set} & 
$V_1$ & $\frac{1}{4}$ & 0.250000$^*$ & 
$V_1$ & $\frac{1}{5}$ & 0.200000$^*$ & 
$V_1$ & $\frac{1}{7}$ & 0.142857$^*$ \\[3ex]
\hline&\multicolumn{3}{c|}{\ }&\multicolumn{3}{c|}{\ }&\multicolumn{3}{c|}{\ }\\[-2ex]
\multirow{8}{\colonewidth}{Identifying Code} & 
$V_1$ & $\frac{2}{5}$ & 0.400000 & 
$V_1$ & $\frac{3}{10}$ & 0.300000 & 
$V_1$ & $\frac{1}{4}$ & 0.250000$^*$ \\
& $V_2$ &  $\frac{33}{80}$  &  0.412500 & 
$N$ &    $\frac{8}{25}$ &  0.320000  & 
$C_{1}^+$ & $\frac{1}{4}$ & 0.250000$^*$    \\
& $C_1$ & $\frac{33}{80}$ & 0.412500 & 
$V_2$ &   $\frac{7}{20}$ & 0.350000$^*$  &
$C_1, C_2$ & $\frac{1}{4}$ & 0.250000$^*$ 
\\
& $C_1, C_2^+$ & $\frac{33}{80}$ & 0.412500 & 
$C_{1}$ & $\frac{7}{20}$ & 0.350000$^*$ &
\multicolumn{3}{c|}{\ }
 \\ 
& $N$ &  $\frac{23}{55}$  & 0.418182$^\dagger$ & 
\multicolumn{3}{c|}{\ }
&
\multicolumn{3}{c|}{\ }
\\[0.5ex]
&$N, J_2$ &   $\frac{23}{55}$  & 0.418182$^\dagger$ & 
\multicolumn{3}{c|}{\ } & 
\multicolumn{3}{c|}{\ }\\
%& $E_6, J_2^+$ &  & & 
%\multicolumn{3}{c|}{\ } & 
%\multicolumn{3}{c|}{\ } \\
&\emph{Upper} \cite{cohen2000bounds}: & $\frac{3}{7}$   & 0.428571$^\ddagger$  & 
\emph{Upper} \cite{ben2005exact}: &   $\frac{7}{20}$ & 0.350000$^*$ & 
\emph{Upper} \cite{karpovsky1998new}: &  $\frac{1}{4}$  &  0.250000 \\[1ex]
\hline&\multicolumn{3}{c|}{\ }&\multicolumn{3}{c|}{\ }&\multicolumn{3}{c|}{\ }\\[-2ex]
\multirow{7}{\colonewidth}{Strong Identifying Code} & 
$V_1$ & $\frac{17}{40}$ & 0.425000 & 
$V_1$ & $\frac{1}{3}$ & 0.333333 & 
$V_1$ & $\frac{1}{4}$ & 0.250000 \\
& $V_2$ &  $\frac{8}{17}$ & 0.470588 & 
$V_2$ &     &  & 
$S$ &    $\frac{22}{73}$ & 0.301370  \\
& $C_1$ & $\frac{8}{17}$ & 0.470588 & 
$C_1$ & $\frac{22}{57}$ & 0.385964 & 
$C_{1}^{+}$ & $\frac{5}{16}$ & 0.312500\\
& $C_1, C_2$ & $\frac{8}{17}$ & 0.470588 &  % AND C_1, C_2^+
$C_1, C_2$ & $\frac{7}{18}$ &  0.388889 & 
$C_{1}^+, C_{2}^+$ & $\frac{4}{13}$ & 0.307692 \\
%& $N$ &   & & 
%$N$ &    &  & 
%$V_1, J_2$ &  $\frac{1}{4}$ & 0.250000 \\
%&$N, J_2$ &   &  & 
%$N, J_2$ &    &  & 
%$S, J_2^+$ &   &  \\
&\emph{Upper} \cite{honkala2010optimal}: &   $\frac{1}{2}$ & 0.50000$^*$  & 
\emph{Upper} \cite{honkala2010optimal}: &   $\frac{2}{5}$ & 0.400000$^*$  & 
\emph{Upper} \cite{honkala2010optimal}: &  $\frac{6}{19}$  & 0.315789$^*$  \\[1ex]
\hline&\multicolumn{3}{c|}{\ }&\multicolumn{3}{c|}{\ }&\multicolumn{3}{c|}{\ }\\[-2ex]
\multirow{6}{\colonewidth}{Locating-Dominating Code} & 
$V_1$ & $\frac{4}{13}$ & 0.307692 & 
$V_1$ & $\frac{1}{4}$ & 0.250000 & 
$V_1$ & $\frac{2}{11}$ & 0.181818 \\
& $V_2$ &  $\frac{1}{3}$ &0.333333$^*$  & 
$V_2$ &   $\frac{3}{10}$ & 0.300000$^*$ & 
$S$ & $\frac{9}{41}$   & 0.219512  \\
& $C_1$ & $\frac{1}{3}$&0.333333$^*$ & 
$C_1$ & $\frac{3}{10}$ & 0.300000$^*$ & 
$C_{1}^+$ & $\frac{2}{9}$ & 0.222222 \\
& $C_1, C_2$ & $\frac{1}{3}$&0.333333$^*$ & 
$C_1, C_2$ & $\frac{3}{10}$  &0.300000$^*$ & 
$C_1, C_2$ & $\frac{12}{53}$ & 0.226415 \\ % also C_1, C_2, C_3
%& $N$ & $\frac{1}{3}$ &0.333333$^*$ &
%$N$ &    &  & 
%$S, J_2^+$ &  $\frac{2}{11}$  &  0.226415 \\
&\emph{Upper} \cite{honkala2006locating}: & $\frac{1}{3}$   &  0.333333$^*$ & 
\emph{Upper} \cite{slater2002fault}: &   $\frac{3}{10}$ & 0.300000$^*$ & 
\emph{Upper} \cite{honkala2006optimal}: &    $\frac{13}{57}$ & 0.228070$^*$  \\[1ex]
\hline&\multicolumn{3}{c|}{\ }&\multicolumn{3}{c|}{\ }&\multicolumn{3}{c|}{\ }\\[-2ex]
\multirow{4}{\colonewidth}{Open-Locating-Dominating Code} & 
$V_1$ & $\frac{4}{9}$ & 0.444444 & 
$V_1$ & $\frac{1}{3}$ & 0.333333 & 
$V_1$ & $\frac{2}{9}$ & 0.222222 \\
& $V_2$ & $\frac{1}{2}$   & 0.500000$^*$ & 
$N$ & $\frac{1}{3}$   &  0.333333 & 
$C_1, C_2$ & $\frac{7}{23}$   &  0.304347  \\
& $C_{1}$ & $\frac{1}{2}$ &0.500000$^*$ & 
$C_{1}^+$ & $\frac{2}{5}$ & 0.400000$^*$   &
$C_{1}^+$ & $\frac{4}{13}$ & 0.307692$^*$ \\
&\emph{Upper} \cite{seo2010open}: &  $\frac{1}{2}$  &  0.500000$^*$ & 
\emph{Upper} \cite{seo2010open}: & $\frac{2}{5}$   & 0.400000$^*$ & 
\emph{Upper} \cite{kincaid2014optimal}: &  $\frac{4}{13}$&0.307692$^*$  \\[1ex]
%\hline&\multicolumn{3}{c|}{\ }&\multicolumn{3}{c|}{\ }&\multicolumn{3}{c|}{\ }\\[-2ex]
%%
%%
%%
%\multirow{4}{\colonewidth}{Neighbor-Identifying Code} & 
%$V_1$ & $\frac{3}{8}$ & 0.375000 & 
%$V_1$ & $\frac{3}{11}$ & 0.272727 & 
%$V_1$ & $\frac{1}{4}$ & 0.250000 \\
%& $V_2$ & $\frac{3}{8}$   & 0.375000 & 
%$V_2$ &  $\frac{3}{11}$ & 0.272727  & 
%$C_{1^+,2}^+$ &  $\frac{1}{4}$  &  0.250000 \\
%&$N, J_2$ &    &   & 
%$N, J_2$ &    &   & 
%$V_1, J_2$ &  $\frac{1}{4}$ & 0.250000 \\
%&\emph{Upper}: &  $\frac{3}{7}$  &  0.428571$^\ddagger$ & 
%\emph{Upper}: &    &  & 
%\emph{Upper}: &    &  \\[1ex]
\hline
\end{tabular}\\
}

{
\footnotesize
$^*$ Bound given is optimal lower bound on density.

$^\dagger$ Bound given is current-best lower bound on density, but may not be optimal.

$^\ddagger$ Bound given is current-best upper bound on density, but may not be optimal.
}

\caption{\label{tab:bounds}Density lower bounds for various set types in various grids, using various discharging rules.}
\end{table}

%
%
%\clearpage
%\section{Implementation Details}\label{sec:details}
%
%
%\subsection{Grids}
%
%\subsection{Configurations}
%
%\subsection{Rules}
%
%\subsection{Constraints}
%
%\subsection{Linear Program}
%
%GLPK, CPLEX.
%
%\subsection{Computation Time}
%
% table!

\label{apx:end}
\end{document}